\documentclass[conference]{IEEEtran}
\IEEEoverridecommandlockouts

\usepackage{booktabs} 
\usepackage{times}
\usepackage{color}
\usepackage{balance}
\usepackage{url}
\usepackage{tabularx}
\usepackage{siunitx}
\usepackage{tabularx,subfigure,multirow, graphicx}
\usepackage{epsfig}
\usepackage{epstopdf}
\usepackage{enumitem}
\usepackage{caption}
\usepackage[bookmarks=false]{hyperref}

\usepackage{amsthm,amssymb,amsmath}
\newcolumntype{L}[1]{>{\raggedright\arraybackslash}p{#1}}
\newcolumntype{C}[1]{>{\centering\arraybackslash}p{#1}}
\newcolumntype{R}[1]{>{\raggedleft\arraybackslash}p{#1}}

\usepackage[linesnumbered,ruled,vlined]{algorithm2e}
\SetKwRepeat{Do}{do}{while}

\SetCommentSty{mycommfont}

\long\def\comment#1{}

\newcommand{\nop}[1]{}

\newcommand{\figureCaptionMargin}{\vspace{-1ex}}
\newcommand{\figureBelowMargin}{\vspace{-2ex}}

\newtheorem{theorem}{\bf Theorem}[section]
\newtheorem{lemma}{\bf Lemma}[section]

\newtheorem{example}{\bf Example}

\theoremstyle{remark}

\theoremstyle{definition}
\newtheorem{definition}{\bf Definition}

\newcommand{\entity}[1]{\mathcal{#1}}
\newcommand{\algvar}[1]{\mathcal{#1}}
\newcommand{\constvar}[1]{\mathbb{#1}}

\newcommand{\vectorfont}[1]{\boldsymbol{#1}}

\newcommand{\basicProblem}{PA-TA}
\newcommand{\approximateBasicProblem}{PAA-TA}

\newcommand{\solutionA}{PUCE}
\newcommand{\solutionATotalName}{Private Utility Conflict-Elimination}

\newcommand{\solutionB}{PGT}
\newcommand{\solutionBTotalName}{Private Game Theoretic Approach}
\newcommand{\solutionCMP}{PDCE}
\newcommand{\solutionCMPTotalName}{Private Distance Conflict-Elimination}
\definecolor{EmphColor}{rgb}{0,0,0}

\definecolor{EmphColorB}{rgb}{0,0,1}

\newcommand{\solutionNPATotalName}{Utility Conflict-Elimination}
\newcommand{\solutionNPBTotalName}{Game Theory}
\newcommand{\solutionNPCMPTotalName}{Distance Conflict-Elimination}
\newcommand{\solutionNPGDTotalName}{Greedy}

\newcommand{\EXPSolutionA}{\solutionA{}}
\newcommand{\EXPSolutionB}{\solutionB{}}
\newcommand{\EXPSolutionCMP}{\solutionCMP{}}
\newcommand{\EXPSolutionNPA}{UCE}
\newcommand{\EXPSolutionNPB}{GT}
\newcommand{\EXPSolutionNPCMP}{DCE}
\newcommand{\EXPSolutionNPGD}{GRD}

\def\BibTeX{{\rm B\kern-.05em{\sc i\kern-.025em b}\kern-.08em
    T\kern-.1667em\lower.7ex\hbox{E}\kern-.125emX}}

\begin{document}

\title{Dynamic Private Task Assignment under Differential Privacy
}

\author{\IEEEauthorblockN{Leilei Du}
\IEEEauthorblockA{
\textit{East China Normal University}\\
Shanghai, China \\
leileidu@stu.ecnu.edu.cn}
\and
\IEEEauthorblockN{Peng Cheng}
\IEEEauthorblockA{
	\textit{East China Normal University}\\
	Shanghai, China \\
	pcheng@sei.ecnu.edu.cn}
\and
\IEEEauthorblockN{Libin Zheng}
\IEEEauthorblockA{
\textit{Sun Yat-sen University}\\
Guangzhou, China \\
zhenglb6@mail.sysu.edu.cn}
\and
\IEEEauthorblockN{Wei Xi}
\IEEEauthorblockA{
\textit{Xi'an Jiaotong University}\\
Shaanxi, China \\
xiwei@xjtu.edu.cn}
\and
\IEEEauthorblockN{Xuemin Lin}
\IEEEauthorblockA{
	\textit{Shanghai Jiao Tong University}\\
	Shanghai, China \\
	xuemin.lin@gmail.com}
\and
\IEEEauthorblockN{Wenjie Zhang}
\IEEEauthorblockA{
\textit{The University of New South Wales}\\
Sydney, Australia \\
wenjie.zhang@unsw.edu.au}
\and
\IEEEauthorblockN{Jing Fang}
\IEEEauthorblockA{
	\textit{East China Normal University}\\
	Shanghai, China \\
	jingfang@stu.ecnu.edu.cn}
}

\maketitle

\begin{abstract}
Data collection is indispensable for spatial crowdsourcing services, such as resource allocation, policymaking, and scientific explorations.
However, privacy issues make it challenging for users to share their information unless receiving sufficient compensation.
Differential Privacy (DP) is a promising mechanism to release helpful information while protecting individuals' privacy.
However, most DP mechanisms only consider a fixed compensation for each user's privacy loss.
In this paper, we design a task assignment scheme that allows workers to dynamically improve their utility with dynamic distance privacy leakage.
Specifically, we propose two solutions to improve the total utility of task assignment results, namely Private Utility Conflict-Elimination (PUCE) approach and Private Game Theory (PGT) approach, respectively.
We prove that PUCE achieves higher utility than the state-of-the-art works.
We demonstrate the efficiency and effectiveness of our PUCE and PGT approaches on both real and synthetic data sets compared with the recent distance-based approach, Private Distance Conflict-Elimination (PDCE).
PUCE is always better than PDCE slightly.
PGT is 50\% to 63\% faster than PDCE and can improve 16\% utility on average when worker range is large enough.
\end{abstract}


\newcommand{\mycolor}{\color{black}}
\section{Introduction}
With the popularity of mobile computing, spatial crowdsourcing has emerged as a new paradigm for spatial task solutions involving human participation.
Workers are {encouraged} to share their data with servers in exchange for benefits.
However, sometimes workers are reluctant to share due to vital privacy leakage (e.g., location) which can
lead to extensive attacks such as identity theft, physical surveillance and stalking and leakage of other sensitive information (e.g., individual health status, racial types, and religion views).
For example, in ride-sharing, if a taxi driver submits his locations to the platform for task requests over a period of time (i.e., a month), a malicious platform attacker is able to guess the driver's range of activity and surveil him {or her.}

Differential Privacy (DP) \cite{DBLP:conf/icalp/Dwork06} is often used to protect individual data.
It trades off utility and privacy by well designing the privacy budget ($\epsilon$).
However, different people have different demands for {both} utility and privacy.
For example, some confidential agencies pay great attention to privacy.
They would rather {gain} high-level privacy protection by sacrificing some utility.
In ride-sharing, some taxi drivers would like to sacrifice some personal location privacy for higher incomes {by} serving more passengers.
{Users need} to adjust their utility by altering the privacy protection level themselves.

\begin{figure}[t!]
	\centering\vspace{-2ex}
	\includegraphics[width=0.2\textwidth]{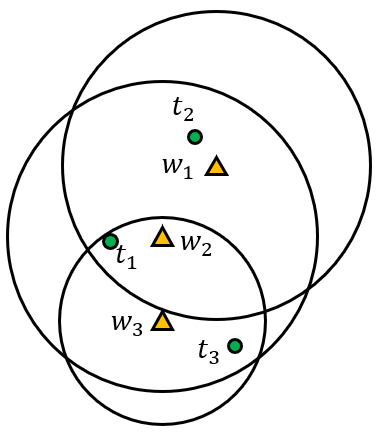}
	\caption{Workers' locations with service areas and tasks' locations.}\label{fig_Introduction_A}
\end{figure}

In this paper, we propose a dynamic private task assignment scheme such that workers can trade their location privacy for higher utilities. Consider the motivation example as follows:

\begin{example}\label{example_1}
As shown in Figure~\ref{fig_Introduction_A}, there are three workers: $w_1$, $w_2$ and $w_3$, and three tasks: $t_1$, $t_2$  and $t_3$.
Each worker $w_j$ competes for tasks with smaller distances.
However, in order to protect their locations, all workers employ a differential privacy mechanism to obfuscate their distances to tasks, and send the obfuscated distances to the server. Assume the server will assign workers to tasks based on their reported obfuscated distances to minimize the total distance as: \{$\langle t_1, w_3 \rangle$, $\langle t_2, w_1 \rangle$, $\langle t_3, w_2 \rangle$\}.
We also assume that the server is untrusted in this example, which means the obfuscated distances on the server can be accessed by workers if they want. Then, worker $w_3$ can sacrifice some {of} his location privacy to report a closer obfuscated distance with task $t_3$. The updated  assignment will be \{$\langle t_1, w_2 \rangle$, $\langle t_2, w_1 \rangle$, $\langle t_3, w_3 \rangle$\} with a smaller total distance.
\end{example}

In this paper, we study the privacy-aware task assignment (PA-TA) problem in spatial crowdsourcing, where workers can dynamically adjust their privacy protection levels for higher utilities. Specifically, we assume a privacy setting where spatial crowdsourcing workers are curious and want to protect their location privacy thus only report obfuscated distances to the server during the task assignment phase without relying on a trusted server. Similar to the existing location protection studies in spatial crowdsourcing \cite{DBLP:conf/icde/ToS018, DBLP:journals/tmc/WangHLWWYQ19, li2021privacy}, we assume the server is untrusted, and thus cannot guarantee the security of received obfuscated distances, which means that other entities (e.g., curious workers) {have} access to the obfuscated distances from the server. In this paper, we handle the task assignment in a multi-proposal enabled batch-based style, where in a given time window each worker can propose to an available task for {multiple} times with different obfuscated distances to improve their utilities until the end of the time window. We first formally define the PA-TA problem. To improve the accuracy of comparing obfuscated distances, we propose a new comparison method, \emph{Partial Probability Comparison Function} (PPCF), which can resolve the comparison between a real distance and an obfuscated distance. We prove our PPCF is better than the existing method, \emph{Probability Compare Function} (PCF) \cite{DBLP:journals/tmc/WangHLWWYQ19}, {both theoretically and practically.} 
To solve PA-TA, we propose two solutions, namely \solutionATotalName{} (\solutionA{}) and \solutionBTotalName{} (\solutionB{}). \solutionA{} is a greedy-based algorithm.
\solutionB{} is on a game-theoretic approach and can achieve higher accuracy than \solutionA{} when the worker range is larger than 1.4 on synthetic data sets.
The contributions of this paper are as follows.

{(1) We formally define the privacy-aware task assignment problem to {support} dynamic privacy budget adjustment for workers in spatial crowdsourcing in Section \ref{sec:problem_definition}.}

{(2) We propose a greedy-based algorithm, namely \solutionATotalName{} (\solutionA{}), in Section~\ref{SolutionA}, and a game-theoretic approach, namely \solutionBTotalName{} (\solutionB{}), in Section~\ref{SolutionB}.

(3) We test our approaches on both synthetic and real data sets to show their efficiency and effectiveness in Section~\ref{Experiment}.

\section{Related Work}\label{RelatedWork}

\noindent\textbf{Task Assignment in Spatial Crowdsourcing.}
Most task assignments in spatial crowdsourcing {focus} on maximizing total utility.
Deng et al. \cite{DBLP:conf/gis/DengSD13} define the total utility as the total number of performed tasks.
Zhang et al. \cite{DBLP:conf/kdd/ZhangHMWZFGY17} maximize the total acceptance ratio of workers.
Zhao et al. \cite{DBLP:conf/aaai/ZhaoXSTZZ19} propose algorithms to maximize the total rewards of the assigned tasks.
Tong et al. \cite{DBLP:conf/icde/TongSDWC16} and Wang et al. \cite{DBLP:conf/icde/WangTLXXL19} maximize the total expected rewards of the assigned tasks. The conventional methods to achieving optimized total utility are exact methods \cite{DBLP:conf/gis/KazemiS12,DBLP:journals/tsas/ToSK15} and greedy methods \cite{DBLP:journals/tkde/SheTCC16,DBLP:journals/tkde/ChengLCHZ16,DBLP:conf/stoc/KarpVV90}.
In order to achieve higher utilities, game-theoretic methods for task assignment are proposed recently.
Ni et al. \cite{DBLP:conf/icde/NiCC020} declare that the tasks may have some dependencies among them and give the definition of dependency-aware spatial crowdsourcing (DA-SC). They propose a game-theoretic approach to solve DA-SC, and {the experiment demonstrates that the game-theoretic approach is superior to the greedy algorithms.}
Zhao et al. \cite{DBLP:conf/icde/ZhaoZGYPJ21} focus on the problem of Fairness-aware Task Assignment, which is to minimize the payoff difference among workers and to maximize the average worker payoff.
They model the problem as a multiplayer game and propose two game-theoretic methods.

\noindent\textbf{Privacy Protection in Spatial Crowdsourcing.}
Differential Privacy \cite{DBLP:conf/icalp/Dwork06} is a golden tool for privacy protection and private data release.
According to the existing of the trusted entity, it can be classified into two categories:
1) Central Differential Privacy (CDP) \cite{DBLP:journals/fttcs/DworkR14}; 2) Geo-Indistinguishability (Geo-I) \cite{DBLP:conf/ccs/AndresBCP13} and Local Differential Privacy (LDP) \cite{DBLP:conf/focs/DuchiJW13}.

To et al. \cite{DBLP:journals/pvldb/ToGS14} adopt \emph{Private Spatial Decomposition} (PSD) \cite{DBLP:journals/sensors/KimCK18} to create obfuscated data releases of workers and devise a geocast mechanism for task request dissemination to protect the privacy of workers' locations.
However, it needs a trusted entity to help sanitize workers' location data.
Wang et al. \cite{DBLP:conf/icdm/WangZYLM16} study \emph{Bayesian attack} \cite{DBLP:conf/ccs/AndresBCP13, DBLP:conf/ccs/BordenabeCP14} on sparse mobile crowdsourcing and propose a privacy-preserving framework to reduce the data quality loss caused by differential location obfuscation. They provide the method to get the optimal location obfuscation matrix satisfying $\epsilon$-differential privacy. It can be used to protect workers' location without relying on the trust entity.
To et al. \cite{DBLP:conf/icde/ToS018} propose a privacy-aware framework that protects the privacy of both tasks and workers in spatial crowdsourcing without any trusted entity. It employs Geo-I to transform both tasks' and workers' locations into obfuscated locations. The platform can identify a set of candidate workers for the task requester through these obfuscated locations without knowing the real locations of both workers and the task. 
Wang et al. \cite{DBLP:journals/tmc/WangHLWWYQ19} also assume that no trusted entity exist, but they propose a method that achieves local differential privacy.
These works get rid of reliance on trusted third parties.
However, they only protect individual privacy without inspiring tasks or workers to participate in the platform.

\noindent\textbf{Private Data Compensation.}
In order to motivate requesters and workers to join spatial crowdsourcing platforms while protecting their location privacy, we need a connection between their utility and privacy cost.
Jin and Zhang \cite{DBLP:journals/ton/JinZ18} provide a framework for spectrum-sensing participants selection, which achieves differential location privacy, approximate social cost minimization, and truthfulness simultaneously.
Ghosh et al. \cite{DBLP:conf/sigecom/GhoshR11} model the utility of competing agents considering privacy cost.
They hold the privacy cost related to some unknown quantities $v$ and suppose the privacy cost is changing linearly with privacy budget $\epsilon$ ($\epsilon v$).
Nissim et al. \cite{DBLP:conf/sigecom/NissimOS12} argue that $\epsilon v$ should be the upper bound rather than the total privacy cost. They propose a privacy-aware mechanism with $v$ below a certain threshold.
Xiao \cite{DBLP:conf/innovations/Xiao13} proposes two models for quantifying an agent's privacy cost using mutual information and max divergence, respectively.
However, it requires the privacy variable $\delta>0$.
Wang et al. \cite{DBLP:journals/tmc/WangHLWWYQ19} propose a personalized privacy-preserving task allocation method. They define \emph{Probability Compare Function} (PCF) to compare two noise values with the acknowledgment of their privacy budget. Besides, they propose \emph{Probabilistic Winner Selection Mechanism}  to minimize the total travel distance and \emph{Vickrey Payment Determination Mechanism} to determine the appropriate payment to each winner of workers satisfying truthfulness, profitability, and probabilistic individual rationality. However, all workers can only have a fixed budget for each task and cannot dynamically compete for tasks with higher utilities.

\section{Problem definition}
\label{sec:problem_definition}

\begin{table}[t!]\vspace{-3ex}
	\begin{center}
		{\small\scriptsize  \vspace{1ex}
			\caption{\small Notations.} \label{varible_description}
			\begin{tabular}{c|c}
                \hline
                \textbf{Variable}                      & \textbf{Description}                                     \\ \hline
                $t_i$                         & the $i$-th   task\\
                $w_j$                         & the $j$-th   worker\\
                $d_{i,j}$                     & the real distance from $t_i$   to $w_j$\\
                $\hat{d}_{i,j}$               & the obfuscated distance from   $t_i$   to $w_j$ \\
                $\tilde{d}_{i,j}$             & the effective obfuscated distance from   $t_i$   to $w_j$              \\
                $\vectorfont{\epsilon_{i,j}}$ & the privacy budget vector owned   by $w_j$   to propose to $t_i$\\
                $\epsilon_{i,j}^{(u)}$        & the $u$-th   element in $\vectorfont{\epsilon_{i,j}}$\\
                $\tilde{\epsilon}_{i,j}$      & the effective privacy budget\\
                $\vectorfont{b_{i,j}}$        & the state vector corresponding to $\vectorfont{\epsilon_{i,j}}$\\
                $b_{i,j}^{(u)}$               & the $u$-th   element in $\vectorfont{b_{i,j}}$ recording whether $\epsilon_{i,j}^{(u)}$ has been used \\
                $s_{i,j}$                     & the state recording whether $t_i$ matches $w_j$ \\ \hline
            \end{tabular}
		}\vspace{-1ex}
	\end{center}
\end{table}

\begin{definition}[Spatial Tasks]\label{SpatialTasks} Let {$t_i$} denote a task.
Its location and value are denoted as {$l_i$} and {$v_i$}, respectively.
\end{definition}

Here,
$v_i$ is an inherent property of $t_i$, and a worker will gain $v_i$ revenue if he serves $t_i$.

\begin{definition}[Spatial Workers]\label{SpatialWorkers} Let {$w_j$} denote a worker located at {$l_j$}.
His service area is denoted as {$A_j$} with a service radius {$r_j$}.
\end{definition}

$A_j$ is a circle area centered at $l_j$ with radius $r_j$ (also called \emph{worker range} in the experiment). Let set $R_j$ denote all tasks in $A_j$.
$w_j$ only proposes to those tasks in $R_j$.

To make the distance and the privacy budget comparable with the task value, we define the Distance Value Function ($f_d$) in Definition~\ref{DVF} and Privacy Budget Function ($f_p$) in Definition~\ref{PBVF} to unify the measurement.

\begin{definition}[Distance Value Function, $f_d$]\label{DVF}
Given a distance $d\in R^{*}$, a function $f_d: R^{*}\to R^{*}$ is called \emph{distance value function}, which takes $d$ as the input and outputs a value $v$. It satisfies that $f_d(0)=0$, $f_d^{'}(\cdot)\geq 0$.
\end{definition}

\begin{definition}[Privacy Budget Value Function, $f_p$]\label{PBVF} Given a privacy budget $\epsilon\in R^{*}$, $f_p: R^{*}\to R^{*}$ is a \emph{privacy budget value function}, which takes $\epsilon$ as input and outputs a value $v$. It satisfies that $f_p(0)=0$, $f_p^{'}(\cdot)\geq 0$ and $\forall \epsilon_1, \epsilon_2 \in R, f_p(\epsilon_1)+f_p(\epsilon_2)=f_p(\epsilon_1+\epsilon_2)$.
\end{definition}

$f_d$ transforms a distance value into a task value. $f_p$ transforms a privacy budget value into a task value.
$f_d$ and $f_p$ are defined as monotone increasing functions and $f_d(0)= f_p(0)=0$.
Besides, $f_p$ is a linear function in this paper and we will consider other types of functions in the future work.

\begin{definition}[Privacy-aware Task Assignment Problem]\label{PrblmDf} Given a set of tasks $\entity{T}$, a set of workers $\entity{W}$, and a set of obfuscated worker-and-task distances {$\{\hat{d}_{i,j}|i\in[m],j\in[n]\}$}, where each $\hat{d}_{i,j}$ is added with a noise {$\eta_{i,j}$} subjecting to distribution {$DF(\epsilon_{i,j})$}, a \basicProblem{} problem is to find a match $M$ between workers and tasks subject to the working area constraint of workers, such that

{\scriptsize\begin{equation}\notag
    \begin{aligned}
    \textrm{max}\;\; & \sum_{t_i\in \entity{T}}\sum_{w_j\in\entity{W}} (s_{i,j}\cdot (v_{i}  - f_d(d_{i,j})) - f_p(\vectorfont{b_{i,j}}\cdot\vectorfont{\epsilon_{i,j}})) &\\
    s.t.\;\; & \sum_{t_i\in\entity{T}}s_{i,j}\leq 1, \;\;\;\; \forall i = 1,2,...,m \\
             & \sum_{w_j\in\entity{W}}s_{i,j}\leq 1, \;\;\;\; \forall j = 1,2,...,n \\
             & \sum_{z\in Z}b_{i,j}^{(z)}\leq Z, \;\;\;\; \forall z = 1,2,...,Z \\
             & s_{i,j}, b_{i,j}\in\{0,1\},\;\forall i = 1,2,...,m;\forall j = 1,2,...,n
    \end{aligned}
\end{equation}}
where {$s_{i,j}$} is the matching state representing whether task {$t_i$} is allocated to worker {$w_j$}. {$s_{i,j}=1$}, if {$t_i$} is allocated to {$w_j$};  otherwise, $s_{i,j}=0$.
$v_i$ is the value of task $t_i$.
{$f_d$} is a Distance Value Function transforming distance to value cost.
{$f_p$} is a Privacy Budget Value Function transforming privacy cost to value cost.
{$\vectorfont{\epsilon_{i,j}}=\langle\epsilon_{i,j}^{(1)},...,\epsilon_{i,j}^{(Z)}\rangle$} is the privacy budget vector between task {$t_i$} and worker {$w_j$}, where {$\epsilon_{i,j}^{(u)}$}({$u\in Z$}) stands for the {$u$}-th proposal of worker {$w_j$} to task {$t_i$}.
{$\vectorfont{b_{i,j}}=\langle b_{i,j}^{(1)},...,b_{i,j}^{(Z)}\rangle$} is the state vector corresponding to {$\vectorfont{\epsilon_{i,j}}$}.
Take {$\vectorfont{b_{1,2}}=\langle 1,1,0,0,0\rangle$} as an example.
It means in the total competition, {$w_2$} can propose to {$t_1$} five times and has already proposed twice with the privacy leakage {$\epsilon_{1,2}^{(1)}$} and {$\epsilon_{1,2}^{(2)}$}.
\end{definition}
The objective of \basicProblem{} is to find a one-to-one match that maximizes the total profit on the platform.
	In the objective function, there are three important parts to construct the matching profit between $t_i$ and $w_j$: task value $v_i$, distance value cost $f_d(d_{i,j})$ and privacy cost $f_p(\vectorfont{b_{i,j}}\cdot\vectorfont{\epsilon_{i,j}})$.
	We model the matching profit as the linear combination of the three parts.
	Note that, the privacy cost is concerned for the process of ``$w_j$ proposing to $t_i$" but not for the final matching state.
	Thus, $f_p(\vectorfont{b_{i,j}}\cdot\vectorfont{\epsilon_{i,j}})$ is not affected by $s_{i,j}$.

We give some of the frequently used variables in Table~\ref{varible_description}.

\section{Review of Conflict Elimination Algorithm}\label{TechniquesConceptions}

Conflict Elimination Algorithm (CEA) \cite{DBLP:journals/tmc/WangHLWWYQ19} is a related work that can resolve the winner conflict problem and can be used as a subroutine in our proposed algorithm, thus we first quickly review CEA. Here, workers are regarded as competitors. When there are more than one worker competing for one task, there will be a conflict, called winner conflict. The problem of resolving all these conflicts is called winner conflict problem.

Given all distances from each task-worker pair, CEA constructs the distance rank matrix $A_{m\times n}=(a_{i,k})_{m\times n}$ where $a_{i,k}$ stands for the index of the worker who is the $k$-th nearest from $t_i$.
For example, $a_{i,k}=j$ means $w_j$ is the $k$-th nearest worker of $t_i$.

For any conflict worker $w_c$ selected by $\varphi$ tasks, CEA allocates only one task to $w_c$ and finds another candidate other than $w_c$ for each of the rest $\varphi-1$ conflict tasks. Thus, for each conflict worker $w_c$, there will be $\varphi$ candidate distance choices as shown in equation~\ref{equal_conflict}:
{\scriptsize\begin{equation}
    \label{equal_conflict}
    \begin{split}
        \left\{
            \begin{array}{l}
                C_1: D_{c_1}=D(a_{c_1,1})+D(a_{c_2,2})+...+D(a_{c_\varphi,2})\\
                C_2: D_{c_2}=D(a_{c_1,2})+D(a_{c_2,1})+...+D(a_{c_\varphi,2})\\
                ...\\
                C_\varphi: D_{c_\varphi}=D(a_{c_1,2})+D(a_{c_2,2})+...+D(a_{c_\varphi,1})
            \end{array}
        \right.
    \end{split}
\end{equation}}
where {\scriptsize$C_u (1\leq u\leq\varphi)$} stands for the $u$-th solution: allocating  $w_c$ to $t_{c_u}$ and other tasks are allocated to the successive workers.

To choose the best solution from $\varphi$ choices in equation~\ref{equal_conflict}, we need to compare four distance values. For example, to compare $C_u$ and $C_v$ $(1\leq u,v\leq\varphi)$, we need to compare $D(a_{c_u,1})+D(a_{c_v,2})$ with $D(a_{c_v,1})+D(a_{c_u,2})$.

If the distances are obfuscated distances, we have to compare four Laplace random variables.
In CEA, it supposes that the difference between the travel distances for different tasks is relatively small for the same worker (i.e., {\scriptsize$D(a_{c_u,1})\simeq D(a_{c_v,1})$}).
Then, CEA only needs to compare two Laplace random variables, which can be calculated by \emph{Probability Compare Function} \cite{DBLP:journals/tmc/WangHLWWYQ19}.

\begin{definition}[Probability Compare Function \cite{DBLP:journals/tmc/WangHLWWYQ19}]\label{def_PCF} Given two values $d_a$ and $d_b$ with their obfuscated values {$\hat{d}_a=d_a+Lap(x,1/\epsilon_a)$} and {$\hat{d}_b=d_b+Lap(x,1/\epsilon_b)$} ({$Lap(x,y)$} is a random variable drawn from Laplace distribution with parameters {$x,y$}), a function {$f:$} {$R^4\to[0,1]$} is called a probability compare function (PCF) if $PCF(\hat{d_a},\hat{d_b},\epsilon_a,\epsilon_b)=\textrm{Pr}[d_a<d_b]$.
\end{definition}

For Example, suppose there are 3 tasks and 3 workers, their distance rank matrix is shown in Table~\ref{tab_DRM}. Each element in the table stands for the worker and his relative distance to the corresponding task.

For $w_3$, both $t_2$ and $t_3$ will choose him first.
Thus $w_3$ is a conflict worker.
We have $C_1:D_{2}=D(a_{2,1})+D(a_{3,2})$ and $C_2: D_{3}=D(a_{2,2})+D(a_{3,1})$.
To make a choice between $C_1$ and $C_2$ (choose the minimal one), it supposes $D(a_{2,1})\simeq D(a_{3,1})$, and thus only needs to compare $D(a_{3,2})$ with $D(a_{2,2})$.
Since $D(a_{2,2})<D(a_{3,2})$, $C_2$ is selected.
}

\begin{table}[t!]
\caption{Distance rank matrix.}
\label{tab_DRM}
\centering
\scalebox{1}{
\begin{tabular}{c|ccc}
\hline
\textbf{Task/Rank} & $1$ & $2$ & $3$ \\ \hline
$t_1$              & $w_1$ (9.06)     &$w_2$ (9.85) & $w_3$ (12.04) \\
$t_2$              & $w_3$ (2.09)  &$w_1$ (10.44) & $w_2$ (12.59) \\
$t_3$              & $w_3$ (2.00)  &$w_2$ (11.28) & $w_1$ (18.87) \\ \hline
\end{tabular}
}
\end{table}

\section{\solutionATotalName{} (\solutionA{})}\label{SolutionA}
A direct method to solve our matching problem is collecting all workers' proposals to tasks with privacy budgets and obfuscated distances and using the Hungarian algorithm to get the optimal matching. Here, the Hungarian matching algorithm \cite{DBLP:books/daglib/0022248}, also called the Kuhn-Munkres algorithm, is one classical method to exactly solve maximum bipartite matching problem with the time complexity of $O(n^3)$, where $n$ is the number of vertices in either part of the bipartite graph.
However, to use the Hungarian algorithm, we have to compare the path length calculated by summing many obfuscated distances, which needs complex comparisons and has low accuracy. In this section, we propose a private utility conflict-elimination (PUCE) algorithm to solve PA-TA problem. Due to each worker can propose to multiple tasks in each round, PUCE greedily chooses the worker-and-task pair that maximizes the subjective function of PA-TA.

\subsection{Comparison and Estimation of Obfuscated Distances}\label{comparisonAndEOD}

Before {introducing}  PUCE algorithm, we first explain three necessary techniques {for solving:}
1) how to calculate a suitable obfuscated distance when there is a series of obfuscated distances for a given task and a given worker;
2) how to compare a real distance with an obfuscated distance;
3) how to compare two utilities when knowing the obfuscated distances.

In this paper, according to the objective function of PA-TA, we define the utility of worker $w_j$ conducting task $t_i$ as:

{\small\begin{equation}\label{worker_utility}
	U_j(i) = v_i - f_d(d_{i,j}) - \sum_{t_i\in \entity{T}} f_p(\vectorfont{b_{i,j}}\cdot\vectorfont{\epsilon_{i,j}})
\end{equation}}

\noindent\textbf{Effective Obfuscated Distance and Effective Privacy Budget.}
In the process of our task assignment, $w_j$ may propose to $t_i$ many times, which means $w_j$ will submit more than one obfuscated distance $\hat{d}_{i,j}$ to the server. For the server, it needs to determine an obfuscated distance (we call it \emph{effective obfuscated distance}) for $d_{i,j}$ to make comparison. For other workers, they also need the \emph{effective obfuscated distance} to compare with the distances of themselves.
Thus, we need a method to calculate the \emph{effective obfuscated distance} in a series of obfuscated distances and ensure the \emph{effective obfuscated distance} supports comparison (i.e., supporting PCF).

We first adopt \emph{maximum likelihood estimation} (MLE) \cite{myung2003tutorial} to get a distance interval $\check{d}$ from a worker $w$'s release set $\vectorfont{DE}=\{(\hat{d}_1,\epsilon_1), (\hat{d}_2,\epsilon_2),...,(\hat{d}_u,\epsilon_u)\}$ for a task $t$.
Let $\vectorfont{DE}.\vectorfont{\hat{d}}$ denote the set $\{\hat{d}_1, \hat{d}_2..., \hat{d}_u\}$ in $\vectorfont{DE}$.
Let $L(X)=L(\hat{d}_1, \hat{d}_2,..., \hat{d}_u; X) = \prod_{k=1}^{u}\textrm{Pr}[\hat{d}_k; X]$, where $\textrm{Pr}[\hat{d}_k; X]$ is the probability function of $Lap(\epsilon_k)$.
When the server gets $\vectorfont{DE}$, it calculates the estimation of $d$ as follows.
{\scriptsize\begin{equation}\notag
		\begin{aligned}
			\check{d} &= \textrm{arg max}_d \prod_{k=1}^{u} \frac{\epsilon_k}{2}\textrm{exp}(-|\hat{d}_k-d|\cdot\epsilon_k)\\
			&= \textrm{arg min}_d \sum_{k=1}^{u} \epsilon_k\cdot|\hat{d}_k-d|.\\
		\end{aligned}
\end{equation}}\vspace{-2ex}

The value of $\check{d}$ is all points on a line segment.
We limit the domain of $d$ in $\vectorfont{DE}.\vectorfont{\hat{d}}$ to get the only estimation of $d$ (supporting comparison).
This estimation of $d$ is the \emph{effective obfuscated distance}, and we denote it as $\tilde{d}$.
We call the corresponding privacy budget (denoted by $\tilde{\epsilon}$) of $\tilde{d}$ in the pair as \emph{effective privacy budget} and call the pair $(\tilde{d},\tilde{\epsilon})$ as \emph{effective distance-budget pair}.

	For example, suppose $w_1$ releases 3 pairs of obfuscated distance and privacy budget to $t_1$: $\vectorfont{DE}=\{(0.1, 0.2), (0.2, 0.9), (0.3, 0.1)\}$. Then we can calculate the effective distance-budget pair as $(\tilde{d}=0.2,\tilde{\epsilon}=0.9)$.

\noindent\textbf{Partial Probability Compare Function (PPCF).}
If $w_{j_1}$ want to compare his distance from himself to $t_i$ with the effective obfuscated distance $\hat{d}_{i,j_2}$ of $w_{j_2}$ to $t_1$, $w_{j_1}$ can utilize the real distance $d_{i,j_1}$ instead of $\hat{d}_{i,j_1}$ or $\tilde{d}_{i,j_1}$ to achieve a more accurate comparison result.
Thus, we need a method for the comparison between a real distance and an obfuscated distance.
Suppose there are two values $d_i$ and $d_j$. The obfuscated value of $d_j$ is $\hat{d}_j$, which is calculated by adding noise $\eta_j$ drawn from a type of distribution $DF(\epsilon_j)$.
Then, we have
{\scriptsize\begin{align}
		\hat{d}_j &= d_j + \eta_j, \;\;\; \eta_j\sim DF(\epsilon_j),\notag\\
		\textrm{Pr}[d_i<d_j] &= \textrm{Pr}[d_i<\hat{d}_j-\eta_j]\notag\\
		&= \textrm{Pr}[\eta_j<\hat{d}_j-d_i].\notag
\end{align}}

Let $f(x)$ be the probability density function of $\eta_j$, then
{\scriptsize\begin{equation}\notag
		\begin{aligned}
			\textrm{Pr}[d_i<d_j] &= \int_{-\infty}^{\hat{d}_j-d_i}f(\eta_j)d\eta_j.
		\end{aligned}
\end{equation}}\vspace{-2ex}

Similar to PCF, we define PPCF$(d_i,\hat{d}_j,\epsilon_j) = \textrm{Pr}[d_i<d_j]$.
If the distribution of $DF(\epsilon_j)$ is symmetric about the y-axis (e.g., Laplace distribution), then
{\scriptsize\begin{equation}\label{PPCF_equation}
		\begin{aligned}
			\textrm{PPCF}(d_i,\hat{d}_j,\epsilon_j)>\frac{1}{2}\Leftrightarrow d_i<\hat{d}_j.
		\end{aligned}
\end{equation}}\vspace{-2ex}

Our PPCF is better than PCF as shown in Theorem~\ref{thrm:PPCF_PCF}. Please refer to the details of the proof in Appendix A.

\begin{theorem}\label{thrm:PPCF_PCF}
	For any given distance $d_x, d_y, \epsilon_x, \epsilon_y$ satisfying $d_x<d_y$. Let $\eta_x\sim Lap(0,1/\epsilon_x), \eta_y\sim Lap(0,1/\epsilon_y)$. Let $\hat{d}_x = d_x + \eta_x, \hat{d}_y = d_y + \eta_y$. Then $Pr[PCF(\hat{d}_x,\hat{d}_y,\epsilon_x,\epsilon_y)>\frac{1}{2}] \leq Pr[PPCF(d_x, \hat{d}_y, \epsilon_y)>\frac{1}{2}]$.
\end{theorem}

\noindent\textbf{Comparison Transformation from Utility to Distance.}
After receiving proposals of workers, the server needs to eliminate conflict among workers for each task.
We can easily use CEA directly to choose only one worker for each task.
However, in CEA, the comparison is based on obfuscated distances rather than utility functions, which does not satisfy our optimized goal.
If we use the utility directly as the comparison object, the server will know the utility value in each round, which leaks the real distance between tasks and workers.

\begin{algorithm}[t!]\scriptsize
	\DontPrintSemicolon
	\caption{\small WorkerProposal}
	\label{alg_worker_propose_multi}
	\KwIn{Not winning worker set $NWW$}
	\KwOut{Candidate list $CL$}
	{Initialize candidate list $CL$ as $m$ empty sets;}\\
	\For{each worker $w_j$ in $NWW$}{
		\For{each task $t_i$ in $R_j$}{
			\If{$w_j$'s privacy budget has been exhausted}{
				{continue;}\\
			}
			{$U_j(i) = v_i - f_d(d_{i,j}) - \sum_{t_i\in \entity{T}} f_p(\vectorfont{b_{i,j}}\cdot\vectorfont{\epsilon_{i,j}})$;}\\
			\If{$U_j(i)\leq 0$}{
				{continue;}\\
			}
			{Get $(\tilde{d}_{i,win(i)},\tilde{\epsilon}_{i,win(i)})$ of $w_{win(i)}$ from the server;}\label{WorkerProposalUtilityComparisonStart}\\
			{Calculate new $(\tilde{d}_{i,j},\tilde{\epsilon}_{i,j})$;}\\
			{Calculate $\tilde{d}'_{i,win(i),j}$ by Equation~\ref{effective_utility_equation};}\\
			\If{$\textrm{PPCF}(d_{i,j},\tilde{d}'_{i,win(i),j},\epsilon_{i,win(i)})\leq 0.5$}{
				{continue;}\\
			}
			\If{$\textrm{PCF}(\tilde{d}_{i,j},\tilde{d}'_{i,win(i),j}, \tilde{\epsilon}_{i,win(i)}, \tilde{\epsilon}_{i,j})\leq 0.5$}{
				continue;\label{WorkerProposalUtilityComparisonEnd}\\
			}
			{Add $\tilde{d}_{i,j}$ to $CL[i]$}
		}
	}
	\Return $CL$;
\end{algorithm}

In order to handle the problem above, we convert the utility comparison into the distance comparison and then use CEA to choose the high-utility one under the distance form.
For any two workers $w_a$ and $w_b$, they hold tasks $t_x$ and $t_y$, respectively.
Their utilities are $U_{a}(x)$ and $U_{b}(y)$, respectively.
Let $V_{a}(x)=U_{a}(x)+f_d(d_{x,a})$ and $V_{b}(y)=U_{b}(y)+f_d(d_{y,b})$. Then we have
{\scriptsize\begin{equation*}
		\begin{aligned}
			\textrm{Pr}(U_{a}(x)>U_{b}(y))  &= \textrm{Pr}(V_{a}(x)-f_d(d_{x,a})>V_{b}(y)-f_d(d_{y,b})) \\
			&= \textrm{Pr}(f_d^{-1}(V_{a}(x))- d_{x,a}>f_d^{-1}(V_{b}(y))-d_{y,b}) \\
			&= \textrm{Pr}(d_{x,a}<d_{y,b} + f_d^{-1}(V_{a}(x)) - f_d^{-1}(V_{b}(y))).
		\end{aligned}
\end{equation*}}
Let

\vspace{-4ex}
{\scriptsize\begin{equation}
		\begin{aligned}
			\hat{d}'_{y,b,a}=\hat{d}_{y,b} + f_d^{-1}(V_{a}(x)) - f_d^{-1}(V_{b}(y)),\label{effective_utility_equation}
		\end{aligned}
\end{equation}}
thus,\vspace{-3.5ex}
{\scriptsize\begin{equation}\notag
		\begin{aligned}
			\textrm{Pr}(U_{a}(x)>U_{b}(y))  &= \textrm{Pr}(\eta_{x,a} -\eta_{y,b}>\hat{d}_{x,a}-\hat{d}'_{y,b,a})\\
			&= \textrm{PCF}(\hat{d}_{x,a},\hat{d}'_{y,b,a}, \epsilon_{x,a}, \epsilon_{y,b}).
		\end{aligned}
\end{equation}}

Therefore, we can calculate $\hat{d}'_{y,b,a}$ for each pair of $w_a$ and $w_b$ with the same task $t_y$ and use PCF function to compare the utility.
Similarly, we can compare $U_a(x)$ and $U_b(y)$ through PPCF:\vspace{-1ex}
{\scriptsize\begin{equation*}
		\begin{aligned}
			\textrm{Pr}(U_{a}(x)>U_{b}(y))  &= \textrm{Pr}(\eta_{y,b}<\hat{d}'_{y,b,a} - d_{x,a})\\
			&= \textrm{PPCF}(d_{x,a}, \hat{d}'_{y,b,a}, \epsilon_{y,b}).
		\end{aligned}
\end{equation*}}

\subsection{The PUCE Algorithm}

We suppose that $w_j$ will propose to all tasks $R_j$ within area $A_j$.
In order to further decline unnecessary privacy costs, we add an extra judgment for workers through the PPCF function.

\begin{algorithm}[t]\scriptsize
	\DontPrintSemicolon
	\caption{\small WinnerChosen}
	\label{alg_winner_chosen}
	\KwIn{Candidate list $CL$, last term allocation list $AL'$}
	\KwOut{Allocation list $AL$, updating state $upd$}
	
	\If{All set in $CL$ are empty} {
		\Return {$(AL',false)$}
	}
	{Initialize $AL$ as $m$ null values;}\\
	{Initialize competing table $CT$ as empty table;}\\
	\For{Each candidate set $CS_i$ in $CL$}{
		\If{$CS_i$ is empty}{\label{WinnerChosenNoApplyStart}
			{Set $AL[i]=AL'[i]$;}\label{WinnerChosenNoApplyEnd}\\
		}
		\Else{
			{Set $CT[i]=CS_i\cup\{\tilde{d}_{i,win(i)}\}$;}\label{WinnerChosenPCFStart}\\
			{Calculate $\hat{d}'_{i,a,b}$ for each pair in $CT[i]$;}\\
			{Sort $CT[i]$ in descending order by $\textrm{PCF}(\hat{d}'_{i,a,b}, \hat{d}_{i,b}, \epsilon_{i,a}, \epsilon_{i,b})$;}\label{WinnerChosenPCFEnd}\\
		}
	}
	{Get updated matching $M$ set by using CEA for $CT$;}\\
	{Add $M$ to $AL$;}\\
	\Return $(AL, true)$;
\end{algorithm}

The worker proposal process and winner-chosen algorithm are respectively shown in Algorithm~\ref{alg_worker_propose_multi} and Algorithm~\ref{alg_winner_chosen}.

In Algorithm~\ref{alg_worker_propose_multi}, each worker $w_j$ checks all the tasks in his service area and judges whether it is worth to complete for the tasks (check whether $U_j(i)>0$ for $t_i\in R_j$).
Besides, he also judges whether he has advantages over the before-winner worker for these tasks by utility comparison. The utility comparison is shown from line \ref{WorkerProposalUtilityComparisonStart} to line \ref{WorkerProposalUtilityComparisonEnd}.
If the two conditions are satisfied, $w_j$ will propose to this task with a new privacy budget and obfuscated distance.

Algorithm~\ref{alg_winner_chosen} takes candidate list $CL$ (constructed by Algorithm~\ref{alg_worker_propose_multi}) and last term allocation list $AL'$ as the input.
It outputs the updating allocation list with the updating state $upd$.
The $false$ value of $upd$ means there is no change for $AL$.
The candidate list will be partitioned into two parts.
Ones with no workers' proposal are the same as the last term ones, which is shown from line \ref{WinnerChosenNoApplyStart} to line \ref{WinnerChosenNoApplyEnd}.
The others containing workers' proposals will be added to a new competing table with the winners of the last term.
Each set of workers for applied tasks in competing table will be sorted by the utility value (compared by $\textrm{PCF}(\hat{d}'_{x,a,b},\hat{d}_{x,b}, \epsilon_{x,a}, \epsilon_{x,b})$) in descending order.
The process is shown from line \ref{WinnerChosenPCFStart} to line \ref{WinnerChosenPCFEnd}.

By executing Algorithm~\ref{alg_worker_propose_multi} and Algorithm~\ref{alg_winner_chosen}, we can construct our \solutionA{} algorithm as shown in Algorithm~\ref{alg_solutionA}.
The total task set and the total worker set can be divided into several time window slices. We execute \solutionA{} on each time window in a batch-based style.
In the beginning, the not winning worker set $NWW$ is initialized as the whole worker set $W$, and the allocation list $AL$ is initialized as an empty set.
We execute Algorithm~\ref{alg_worker_propose_multi} to get candidate allocation list $CL$.
Then we execute Algorithm~\ref{alg_winner_chosen} to pick a new allocation list $AL$ and get a updating state $upd$.
When there are still some workers proposing to tasks ($CL$ is not empty), $upd$ will be set as $true$.
We also update $NWW$ by removing the new winner workers and adding the new loser workers.
When no workers propose to any task, $upd$ will be set as $false$.
Thus, we get the final task-worker matching pairs $TWM$ as $AL$.

\begin{algorithm}[t]\scriptsize
\DontPrintSemicolon
\caption{\small\solutionA{}}
\label{alg_solutionA}
\KwIn{Tasks $\entity{T}$ and workers $\entity{W}$ in the current time window}
\KwOut{The task-worker matching pairs $TWM$}

{Initialize not winning worker set $NWW$ as $\entity{W}$;}\\
{Initialize halt state $hs$ as $false$;}\\
{Initialize allocation list $AL$ as $m$ empty set list;}\\
\While{$hs$ is not $true$}{
    {Get $CL$ by executing Algorithm~\ref{alg_worker_propose_multi};}\\
    {Get $AL$ and $upd$ by executing Algorithm~\ref{alg_winner_chosen};}\\
    {Update $NWW$ by removing new winners and adding new losers;}\\
    {Set $hs=upd$;}\\
}
{Set $TWM$ as $AL$;}\\
\Return $TWM$;
\end{algorithm}

\begin{example}[Running Example of PUCE]
	We give a running example of the whole process of \solutionA{} following the motivation example.
	As shown in Figure \ref{fig_Introduction_A}, three workers $w_1$, $w_2$ and $w_3$ have service areas $15$, $15$ and $10$, respectively.
	Three tasks $t_1$, $t_2$ and $t_3$ have task values $12.4$, $11$ and $13$, respectively.
	The distance between each task and worker is shown in Table~\ref{tab_puce_solution}.

\begin{table}[t!]\centering	\vspace{-2ex}
	\caption{Task-worker distances.}
	\label{tab_puce_solution}
		\scalebox{1}{
			\begin{tabular}{c|ccc}
				\hline
				\textbf{Worker/Task} & $t_1$ & $t_2$ & $t_3$ \\ \hline
				$w_1$               & 12.2  & 3.61  & 17.12 \\
				$w_2$              & 5     & 10.44 & 12.21 \\
				$w_3$               & 9.43  & 18.25 & 7.28  \\ \hline
			\end{tabular}
	}
\end{table}

	Suppose there are three privacy budgets for each task-worker pair.
	The corresponding effective distance, the privacy budget and utility are shown in Table~\ref{tab_puce_solution2}.

	At the beginning, $NWW$ is set as $\{w_1, w_2, w_3\}$. $CL$ is set to NULL.
	$w_1$ firstly judges whether the tasks within his service area will be added to the $CL$.
	He calculates the utility for $t_1$ as $U_1(1)=0.1>0$ and adds $\tilde{d}_{1,1}$ to $CL[1]$.
	Besides, he adds $\tilde{d}_{2,1}$ to $CL[2]$.
	And $w_2$, $w_3$ also add their selected tasks (by the judgement in Algorithm~\ref{alg_worker_propose_multi}). And we can get the data in $CL$ as shown in Table~\ref{tab_CL} (the utility values are shown in square brackets).
	Then we get $CT$ by sorting $CL$, which is shown in Table~\ref{tab_CT}.

	After that, we find $t_1$ is allocated to $w_3$.
	Besides, $t_2$ and $t_3$ fall into conflict for $w_2$.
	After the comparison of CEA, $t_3$ is allocated to $w_2$.
	In the next round, there is only $t_2$ unallocated.
	And $w_1$ has not matched any task yet.
	$w_1$ can only propose for $t_2$. However the utility of $U_1(2)$ in this round is $-3.1\leq 0$.
	Thus there is no worker proposing to any tasks in this round. And the process is end.
\end{example}

\noindent\textbf{Privacy Analysis.}
We define the query data set of worker $w_j$ as $X_j$, which consists of all tasks in the service area of $w_j$ (i.e., $R_j$).
The neighboring data set of $X_j$ is noted as $X_j'$.
It satisfies that $\|X_j-X_j'\|=1$, which means there is only one different task item between $X_j$ and $X_j'$.
We focus on the query $f$ as `\emph{Get each distance from $w_j$ to his service tasks $R_j$}'.
That means $f(X_j)=[d_{i_1,j},...,d_{i_{|R_j|},j}]$.

\begin{table}[t!]\small\vspace{-2ex}
		\caption{\small Effective obfuscated distance, privacy budget and utility.}
		\label{tab_puce_solution2}
		\centering
		\scalebox{0.8}{
			\begin{tabular}{|c|cc|cc|cc|}
				\hline
				\textbf{Matchable   pair} & \multicolumn{2}{c|}{$(\tilde{d},\epsilon^{(1)})$, utility} & \multicolumn{2}{c|}{$(\tilde{d},\epsilon^{(2)})$, utility} & \multicolumn{2}{c|}{$(\tilde{d},\epsilon^{(3)})$, utility} \\ \hline
				$(t_1,w_1)$               & \multicolumn{1}{c|}{(12.7,0.1)}            & 0.1          & \multicolumn{1}{c|}{(12.4,0.3)}            &               & \multicolumn{1}{c|}{(12.3,0.4)}              &             \\ \hline
				$(t_1,w_2)$               & \multicolumn{1}{c|}{(5.5,4.6)}             & 2.8           & \multicolumn{1}{c|}{(5.3,4.65)}             &               & \multicolumn{1}{c|}{(5.1,4.8)}               &             \\ \hline
				$(t_1,w_3)$               & \multicolumn{1}{c|}{(9.93,0.1)}            & 2.87          & \multicolumn{1}{c|}{(9.63,0.4)}            &               & \multicolumn{1}{c|}{(9.53,0.4)}              &             \\ \hline
				$(t_2,w_1)$               & \multicolumn{1}{c|}{(4.11,6.99)}           & 0.4           & \multicolumn{1}{c|}{(4.01,7.1)}           & -3.1          & \multicolumn{1}{c|}{(3.81,7.2)}             &             \\ \hline
				$(t_2,w_2)$               & \multicolumn{1}{c|}{(10.94,0.1)}           & 0.46          & \multicolumn{1}{c|}{(10.64,0.2)}           &               & \multicolumn{1}{c|}{(10.54,0.5)}             &             \\ \hline
				$(t_3,w_2)$               & \multicolumn{1}{c|}{(12.71,0.1)}           & 0.69          & \multicolumn{1}{c|}{(12.51,0.3)}           &               & \multicolumn{1}{c|}{(12.31,0.4)}             &             \\ \hline
				$(t_3,w_3)$               & \multicolumn{1}{c|}{(7.78,5.4)}            & 0.32          & \multicolumn{1}{c|}{(7.58,5.5)}            &               & \multicolumn{1}{c|}{(7.38,5.6)}              &             \\ \hline
			\end{tabular}
	}
\end{table}

\begin{table}[t!]\centering\vspace{-2ex}
		\caption{Candidate list $CL$.}
		\label{tab_CL}
		
		\scalebox{1}{
			\begin{tabular}{c|c|c|c}
				\hline
				\textbf{CL} &                                &                                      &\\ \hline
				1           & $\tilde{d}_{1,1}$=12.7, [0.1]    & $\tilde{d}_{1,2}$=5.5, [2.8]    & $\tilde{d}_{1,3}$=9.93 [2.87]       \\
				2           & $\tilde{d}_{2,1}$=4.11, [0.4]   & $\tilde{d}_{2,2}$=10.94, [0.46]  &  \\
				3           & $\tilde{d}_{3,2}$=12.71, [0.69] & $\tilde{d}_{3,3}$=7.78, [0.32]   &  \\ \hline
			\end{tabular}
	}
\end{table}

\begin{table}[t!]\centering\vspace{-3ex}
		\caption{Competing table $CT$.}
		\label{tab_CT}
		
		\scalebox{1}{
			\begin{tabular}{c|c|c|c}
				\hline
				\textbf{CT} &                  1              &                2                &   3\\ \hline
				1           & $\tilde{d}_{1,3}$=9.93 [2.87]    & $\tilde{d}_{1,2}$=5.5, [2.8]   & $\tilde{d}_{1,1}$=12.7, [0.1]  \\
				2           & $\tilde{d}_{2,2}$=10.94, [0.46]   & $\tilde{d}_{2,1}$=4.11, [0.4] &\\
				3           & $\tilde{d}_{3,2}$=12.71, [0.69] & $\tilde{d}_{3,3}$=7.78, [0.32]  &\\ \hline
			\end{tabular}
	}
\end{table}

\begin{theorem}\label{SolutionA_PA}
\solutionA{} satisfies {\scriptsize$(\sum_{t_i\in R_j}\vectorfont{b_{i,j}\epsilon_{i,j}}r_j)$}-local differential privacy for each worker $w_j$.
\end{theorem}

\begin{proof}
Let $\algvar{A}_j$ be the mechanism \solutionA{} applying to $w_j$ with query $f$ defined above.
Let $X_j$ be the location of $w_j$.
For query $f(X_j)=[d_{i_1,j},...,d_{i_{|R_j|},j}]$, we extend it to an equivalent query $\hat{f}(X_j)=f(X_j)\cdot \algvar{J}$, where $\algvar{J}$ is a block diagonal matrix:

\vspace{-1ex}
{\scriptsize\begin{equation}\notag
\algvar{J}=
    \begin{bmatrix}
        CP(\vectorfont{b_{i_1,j}})  &       &       &   \\
        &       CP(\vectorfont{b_{i_2,j}})  &       &   \\
        &       &       \ddots  &   \\
        &       &       &       CP(\vectorfont{b_{i_{|R_j|},j}})
    \end{bmatrix}
\end{equation}}
Here, {\scriptsize$CP(\vectorfont{b})$} is the compression of {$\vectorfont{b}$}, which means removing all zero element of {$\vectorfont{b}$}. For example, if {$\vectorfont{b}=[1,1,0,0,0]$, then {\scriptsize$CP(\vectorfont{b})=[1,1]$}.}  $\hat{f}(X_j)$ means query $d_{i_u,j}$ for $sum(\vectorfont{b_{i_u,j}})$ times for $u\in[|R_j|]$, where $sum(\vectorfont{b_{i_u,j}})$ means the sum of all elements in $\vectorfont{b_{i_u,j}}$. We denote the size of $\hat{f}(X_j)$ as $|\hat{f}|$ and the $a$-th element of $\hat{f}(X_j)$ as $\hat{f}(X_j)_a$.

Let $Y_j$ denote the set of all published obfuscated distances of the worker $w_j$ to tasks in $R_j$.
Then we have $Y_j=\hat{f}(X_j)+[\eta_1, \eta_2,...,\eta_{|\hat{f}|}]$, where $\eta_a(1\leq a\leq |\hat{f}|)$ is an i.i.d random variable drawn from $Lap(1/\epsilon_a)$.
Hence we have
{\scriptsize\begin{equation}\notag
    \begin{aligned}
        \frac{\textrm{Pr}[\algvar{A}_j(X_j)=Y_j]}{\textrm{Pr}[\algvar{A}_j(X'_j)=Y_j]}
        &= \prod_{a\in[|\hat{f}|]} (\frac{\textrm{exp}(-\epsilon_a|Y_{j,a}-\hat{f}(X_{j})_a|)}{\textrm{exp}(-\epsilon_a|Y_{j,a}-\hat{f}(X_{j}')_a|)})
           \\
        &=\prod_{t_i\in R_j}\prod_{u\in[sum(\vectorfont{b_{i,j}})]}
           (\frac{\textrm{exp}(-\epsilon_{i,j}^{(u)}|\tilde{d}_{i,j}^{(u)}-d_{i,j}|)}{\textrm{exp}(-\epsilon_{i,j}^{(u)}|\tilde{d}_{i,j}^{(u)}-d'_{i,j}|)}) \\
        &\leq\prod_{t_i\in R_j}\prod_{u\in[sum(\vectorfont{b_{i,j}})]}
           (\textrm{exp}(\epsilon_{i,j}^{(u)}(|d_{i,j}-d'_{i,j}|))) \\
        &=\prod_{t_i\in R_j}
           \textrm{exp}(\vectorfont{b_{i,j}\epsilon_{i,j}}(|d_{i,j}-d'_{i,j}|)) \\
        &\leq\textrm{exp}(\sum_{t_i\in R_j}\vectorfont{b_{i,j}\epsilon_{i,j}}r_j).
    \end{aligned}
\end{equation}}
Because $X_j$ contains only one element, then we have \solutionA{} satisfies $(\sum_{t_i\in R_j}\vectorfont{b_{i,j}\epsilon_{i,j}}r_j)$-local differential privacy for each worker $w_j$.
\end{proof}

\noindent\textbf{Time Cost Analysis.}
There are $m$ tasks and $n$ workers.
{Each worker has $Z$ privacy budget for each task.}
Therefore, the worst time cost for \solutionA{} is $O(m\cdot n\cdot Z)$.

\section{\solutionBTotalName{} (\solutionB{})}\label{SolutionB}
In this section, we declare that each worker can compete for each task within their service area, whether they have already won a task.
We {model} our problem as an exact potential game with at least one Nash equilibrium in pure strategy.
{To make the utility value support comparison under a privacy circumstance,} we approximate our utility function by replacing real distance with effective obfuscated distance.

\subsection{Cases of Utility Change in Competition}
There are three cases of utility change in each time of competition for each task-worker pair.
They are \emph{Winning Change}, \emph{Abandoned Change} and \emph{Defeated Change}.
We denote them as $\Delta U_{j}^{W}(i)$, $\Delta U_{j}^{A}(i)$ and $\Delta U_{j}^{D}(i)$ respectively, which are expressed as follows:\vspace{-1.5ex}
{\scriptsize\begin{align}
	\Delta U_{j}^{W}(i) &= v_i - f_d(\tilde{d}_{i,j})-f_p(\epsilon_{i,j}^{(z)}),\notag\\
	\Delta U_{j}^{A}(i) &= - v_i + f_d(\tilde{d}_{i,j})\notag,\\
	\Delta U_{j}^{D}(i) &= - v_i + f_d(\tilde{d}_{i,j}).\notag
\end{align}}

$\Delta U_{j}^{W}(i)$ means the utility change of winning task $t_i$ for worker $w_j$.
$\Delta U_{j}^{A}(i)$ means the utility change of abandoning task $t_i$ {(because each worker can only match one task at most)} for worker $w_j$.
$\Delta U_{j}^{D}(i)$ means the utility change of being defeated by some other competitor in competing for task $t_i$ for worker $w_j$.
It is the same with $\Delta U_{j}^{A}(i)$.
We use $\Delta U_{j}^{W(k)}(i)$, $\Delta U_{j}^{A(k)}(i)$ and $\Delta U_{j}^{D(k)}(i)$ to denote the above three utility change in $k$-th competition.

We give examples of these three utility changes.
Suppose there are two workers $w_1, w_2$ and two tasks $t_1, t_2$.
At the first stage, $w_1$ competes for $t_1$ and $w_2$ competes for $t_2$.
Then the corresponding $\Delta U_{1}^{W}(1)$ and $\Delta U_{2}^{W}(2)$ are shown in Figure~\ref{subfig:delta_utility_A}.
At the second stage, $w_1$ competes for $t_2$ and gets it successfully. As is shown in Figure~\ref{subfig:delta_utility_B}.
The utility change between $w_1$ and $t_2$ is $\Delta U_1^W(2)$.
The utility change between $w_1$ and $t_1$ is $\Delta U_1^A(1)$.
The utility change between $w_2$ and $t_2$ is $\Delta U_2^D(2)$.

\begin{figure}[t!]\centering \vspace{-2ex}
	\subfigure[][{\scriptsize Stage 1}]{
		\scalebox{0.5}[0.5]{\includegraphics{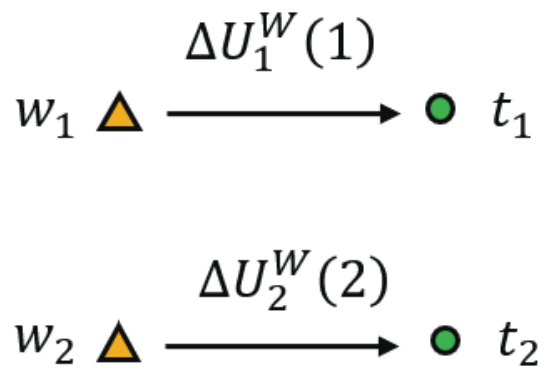}}
		\label{subfig:delta_utility_A}}
	\subfigure[][{\scriptsize Stage 2}]{
		\scalebox{0.5}[0.5]{\includegraphics{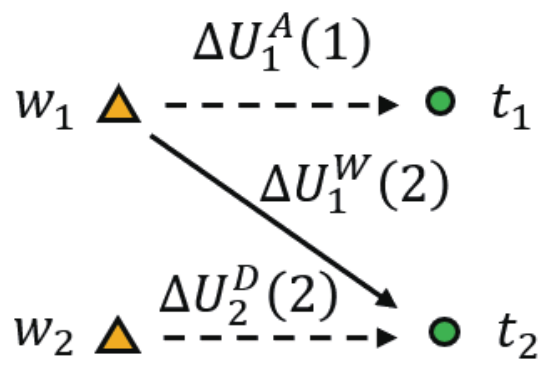}}
		\label{subfig:delta_utility_B}}
	\caption{\small Utility change.}
	\label{fig:delta_utility}
\end{figure}

\subsection{Game Modeling and Nash Equilibrium}
We approximate our \basicProblem{} as \emph{Privacy-aware Approximate Task Assignment} (\approximateBasicProblem{}) problem by replacing the real distance as effective distance.
We formulate \approximateBasicProblem{} as an $n$-player strategic game, $\mathcal{G}=<\entity{W},\vectorfont{S},\vectorfont{UT}>$.
$\mathcal{G}$ consists of players $\entity{W}$, strategy spaces $\vectorfont{S}$, and utility functions $\vectorfont{UT}$.
We specify these three components as follows:

  (1) {\scriptsize$\entity{W}=\{w_1,...,w_n\}$} denotes the finite set of $n$ workers with $n\geq 2$.
  We will use worker and player interchangeably in the rest of the paper.

  (2) {\scriptsize$\vectorfont{S}=\{S_j\}_{j=1}^{n}$} is the strategy spaces (i.e., the overall strategy set of all players).
  $S_j$ is the finite set of strategies available to worker $w_j$. Here, one strategy of worker $w_j$ indicates an action that he proposes to some task $t_i$ with a privacy budget $\epsilon^{(u)}_{i,j}$ for the $u$-th proposal.

  (3) {\scriptsize$\vectorfont{UT}=\{UT_j^{(k)}\}_{j=1}^{n}$} is the utility functions of all players $w_j$ where $k$ is the total competition number.
  For each chosen strategy {\scriptsize$\vectorfont{st}\in S$, $UT_j^{(k)}(\vectorfont{st})\in \mathbb{R}$} is the utility of player $w_j$.
  We calculate {\scriptsize$UT_j^{(k)}(\vectorfont{st})$} as follows:\vspace{-2ex}

      {\scriptsize\begin{equation}
        \label{formula_utility}
        \begin{aligned}
        UT_j^{(k)}(\vectorfont{st})  &= \Delta U_j^{W(k)}(i_2) + \Delta U_{win(i_2)}^{D(k-1)}(i_2) + \Delta U_j^{A(k-1)}(i_1)\\
                                     &= v_{i_2} - f_d(\tilde{d}_{i_2,j}^{(k)}) - f_p(\epsilon_{i_2,j}^{(z_{k})}) - v_{i_2} + f_d(\tilde{d}_{i_2,win(i_2)}^{(k-1)})\\
                                     & - v_{i_1} + f_d(\tilde{d}_{i_1,j}^{(k-1)})\\
                                     &= - f_d(\tilde{d}_{i_2,j}^{(k)}) - f_p(\epsilon_{i_2,j}^{(z_{k})}) + f_d(\tilde{d}_{i_2,win(i_2)}^{(k-1)}) - v_{i_1} + f_d(\tilde{d}_{i_1,j}^{(k-1)}).\\
        \end{aligned}
      \end{equation}}

In equation \ref{formula_utility}, $w_j$ wins $t_{i_1}$ and $w_{win(i_2)}$ wins $t_{i_2}$ in $(k-1)$-th competition. $w_j$ will compete for $t_{i_2}$ in $k$-th competition.

In the following part, we define \emph{exact potential game} (EPG) and prove that \approximateBasicProblem{} is an EPG.
\begin{definition}[Exact Potential Game] \label{EPG_definition}
A strategic game, $\mathcal{G}=<\entity{W},\vectorfont{S},\vectorfont{UT}>$, is an Exact Potential Game (EPG) if there exists a function, $\Phi: \vectorfont{S}\to \mathbb{R}$, such that for all $\vectorfont{st}_j\in\vectorfont{S}$, it holds that, $\forall w_j\in\entity{W}$, $\forall k\in N^{+}$,
{\scriptsize\begin{equation}\notag
    \begin{aligned}
        &UT_j^{(k)}(st'_j,\vectorfont{st}_{-j}) - UT_j^{(k)}(st_j,\vectorfont{st}_{-j})\\
        =&\Phi^{(k)}(st'_j,\vectorfont{st}_{-j}) - \Phi^{(k)}(st_j,\vectorfont{st}_{-j}).
    \end{aligned}
\end{equation}}
\end{definition}

\begin{algorithm}[t]\scriptsize
	\DontPrintSemicolon
	\caption{\small\solutionB{}}
	\label{alg_game}
	\KwIn{Tasks $\entity{T}$ and workers $\entity{W}$ in the current time window}
	\KwOut{The allocation list $AL$}
	{Initialize $AL$ as a list with $m$ $null$ value;}\\
	{Initialize halt state $hs$ as $false$;}\\
	\While{$hs$ is $false$}{
		{Set $hs$ as $true$;}\\
		\For{each worker $w_j\in\entity{W}$}{
			{Get the maximal $UT_j$ for each task $t_i\in R_j\setminus\{AL[b]\}$;}\label{ut_j}\\
			\If{$UT_j(\vectorfont{st})$ is $null$ or $UT_j(\vectorfont{st})\leq 0$}{
				continue;
			}
			{Set $hs$ as $false$;}\\
			{Set $t_c$ as $w_j$'s already mateched task;}\\
			{Set $t_b$ as the task with maximal $UT_j$;}\\
			{Set $w_f$ as the worker matched $t_b$ before;}\\
			{Update effective distance-budget pair between $t_b$ and $w_j$;}\\
			{Set $AL[c]=null$;}\\
			{Set $AL[b]=w_j$;}\\
		}
	}
	
	\Return $AL$;
\end{algorithm}

\begin{theorem}
\approximateBasicProblem{} is an Exact Potential Game (EPG).
\end{theorem}

\emph{Proof:} We define a potential function as {\scriptsize$$\Phi^{(k)}(\vectorfont{st})=\sum_{t_i\in\entity{T}}\sum_{w_j\in\entity{W}}(s_{i,j}^{(k)}\cdot (v_i-f_d(\tilde{d}_{i,j}))-f_p(\vectorfont{b_{i,j}^{(k)}}\cdot\vectorfont{\epsilon_{i,j}}))$$} which represents the total utility value of the matching result in $k$-th competition that all worker gain.
Let {\scriptsize$\tilde{U}_j^{(k)}(i)=v_i-f_d(\tilde{d}_{i,j}^{(k)})-\sum_{t_i\in\entity{W}}f_p(\vectorfont{b_{i,j}^{(k)}}\cdot\vectorfont{\epsilon_{i,j}})$} be the approximate value of $U_j(i)$ by replacing the real distance $d_{i,j}$ with the effective obfuscated distance $\tilde{d}_{i,j}$.
Then we get the recurrence relation of $\tilde{U}_j^{(k)}(i)$ for $k$ as\vspace{-2ex}
{\scriptsize\begin{equation}\notag
    \begin{aligned}
        \tilde{U}_j^{(k)}(i) = \left\{
            \begin{array}{l}
                \tilde{U}_j^{(k-1)}(i) + v_i - f_d(\tilde{d}_{i,j}^{(k)}) - f_p(\epsilon_{i,j}^{(z_k)})\;\;\;\sharp 1   \\
                \tilde{U}_j^{(k-1)}(i) - v_i + f_d(\tilde{d}_{i,j}^{(k-1)})\;\;\;\sharp 2 \\
                \tilde{U}_j^{(k-1)}(i) \;\;\;\sharp 3
            \end{array}
        \right.
    \end{aligned}
\end{equation}}
where condition $\sharp$1 means $w_j$ wins $t_i$ in $k$-th competition, condition $\sharp$2 means $w_j$ \textrm{gives up his original task or is defeated} in $k$-th competition and condition $\sharp$3 means {there is no change between $t_i$ and $w_j$}.
Suppose that $w_j, w_{j_x}, w_{j_y}$ wins $t_{i_1},t_{i_2},t_{i_3}$ in $(k-1)$-th competition respectively and $w_j$ will compete for $t_{i_2}$ ($st_j$) or $t_{i_3}$ ($st'_j$) in $k$-th competition, then we obtain
{\scriptsize\begin{equation}\notag
    \begin{aligned}
            & \Phi^{(k)}(st'_j,\vectorfont{st}_{-j}) - \Phi^{(k)}(st_j,\vectorfont{st}_{-j}) \\
           =& \tilde{U}_j^{(k)}(i_1) + \tilde{U}_j^{(k-1)}(i_2) + \tilde{U}_j^{(k)}(i_3) + \tilde{U}_{j_y}^{(k)}(i_3) + \tilde{U}_{j_x}^{(k-1)}(i_2) \\
            & - (\tilde{U}_j^{(k)}(i_1) + \tilde{U}_j^{(k)}(i_2) + \tilde{U}_j^{(k-1)}(i_3)  + \tilde{U}_{j_y}^{(k-1)}(i_3) + \tilde{U}_{j_x}^{(k)}(i_2)) \\
           =& \tilde{U}_j^{(k)}(i_3) - \tilde{U}_j^{(k-1)}(i_3) - (\tilde{U}_j^{(k)}(i_2) - \tilde{U}_j^{(k-1)}(i_2)) \\
            & + \tilde{U}_{j_y}^{(k)}(i_3) - \tilde{U}_{j_y}^{(k-1)}(i_3) - (\tilde{U}_{j_x}^{(k)}(i_2) - \tilde{U}_{j_x}^{(k-1)}(i_2))\\
           =& v_{i_3} - f_d(\tilde{d}_{i_3,j}^{(k)})- f_p(\epsilon_{i_3,j}^{(z_k)}) - (v_{i_2} - f_d(\tilde{d}_{i_2,j}^{(k)})- f_p(\epsilon_{i_2,j}^{(z_k)}))\\
            & - v_{i_3} + f_d(\tilde{d}_{i_3,j_y}^{(k-1)}) - ( - v_{i_2} + f_d(\tilde{d}_{i_2,j_x}^{(k-1)}))\\
           =& - f_d(\tilde{d}_{i_3,j}^{(k)})- f_p(\epsilon_{i_3,j}^{(z_k)}) + f_d(\tilde{d}_{i_3,j_y}^{(k-1)})\\
            & + f_d(\tilde{d}_{i_2,j}^{(k)}) + f_p(\epsilon_{i_2,j}^{(z_k)}) - f_d(\tilde{d}_{i_2,j_x}^{(k-1)})\\
           =& v_{i_3}- f_d(\tilde{d}_{i_3,j}^{(k)})- f_p(\epsilon_{i_3,j}^{(z_k)}) - v_{i_3} + f_d(\tilde{d}_{i_3,j_y}^{(k-1)}) - v_{i_1} + f_d(\tilde{d}_{i_1,j}^{(k-1)})\\
            & - (v_{i_2} - f_d(\tilde{d}_{i_2,j}^{(k)}) - f_p(\epsilon_{i_2,j}^{(z_k)}) - v_{i_2} + f_d(\tilde{d}_{i_2,j_x}^{(k-1)})- v_{i_1} + f_d(\tilde{d}_{i_1,j}^{(k-1)}))\\
           =& \Delta U_{j}^{W(k)}(i_3) + \Delta U_{j_y}^{D(k-1)}(i_3) + \Delta U_{j}^{A(k-1)}(i_1)\\
            & - (\Delta U_{j}^{W(k)}(i_2) + \Delta U_{j_x}^{D(k-1)}(i_2) + \Delta U_{j}^{A(k-1)}(i_1))\\
           =& UT_j^{(k)}(st'_j,\vectorfont{st}_{-j}) - UT_j^{(k)}(st_j,\vectorfont{st}_{-j})
    \end{aligned}
\end{equation}}

According to Definition~\ref{EPG_definition}, the strategic game of the \approximateBasicProblem{} is an exact potential game.
Therefore, \approximateBasicProblem{} has pure Nash equilibrium according to Theorem~2.3 in Ref~\cite{chew2016potential}.

\noindent\textbf{PGT Algorithm.}
The server executes the competition process with the aid of workers.
Each worker $w_j$ needs to repeat choosing the best task $t_b$ for the maximal utility value.
If the maximal value is positive, $w_j$ will update his effective distance-budget pair for $t_b$  and ask the server to update the allocation list.

We give the process in Algorithm~\ref{alg_game}.
The critical step is to calculate the best response information (maximal $UT_j$) shown in line \ref{ut_j}.
The state variable $hs$ is a boolean variable that indicates whether there still exists a task that can improve a utility function $UT_j$ for any $w_j\in\mathcal{W}$.
If there is no such task, the process will halt.

\begin{example}[Running Example of PGT]
Consider the example in Table~\ref{tab_puce_solution}, and the effective obfuscated distance and privacy budgets are shown in Table~\ref{tab_puce_solution2}.
As shown in Table~\ref{tab_gt_allocation_list}, suppose in the $k$-th competition, the winners of $t_1$, $t_2$ and $t_3$ are $w_1$, $w_2$ and $w_3$ respectively. And $w_1$, $w_2$ and $w_3$ have consumed their first privacy budgets $\epsilon_1$ for all three tasks.
Besides, they public the obfuscated distances relevant to $\epsilon_1$  for all tasks
(so that all the effective obfuscated distances related to $\epsilon_1$ are able to calculated by the server and all workers).
Suppose $f_d$ and $f_p$ are both identity functions (i.e., {\scriptsize$f_d(x)=x$, $f_p(x)=x$}).

In the $(k+1)$-th competition, it is $w_1$'s turn to compete.
$w_1$ can only compete for $t_2$.
He first uses his new privacy budget {\scriptsize$\epsilon_{2,1}^{(z_{k+1})}=\epsilon_{2,1}^{(2)}=7.1$} and calculates the new effective obfuscated distance {\scriptsize$\tilde{d}_{2,1}^{(k+1)}=4.01$}.
After that, he calculates {\scriptsize$UT_1^{(k+1)}=-f_d(\tilde{d}_{2,1}^{(k+1)})-f_p(\epsilon_{2,1}^{(z_{k+1})})+f_d(\tilde{d}_{2,2}^{(k)})-v_1+f_d(\tilde{d}_{1,1}^{(k)})=0.13>0$}.
Then, he publishes his privacy budget {\scriptsize$\epsilon_{2,1}^{(2)}=7.1$} with the corresponding obfuscated distance {\scriptsize$\hat{d}_{2,1}^{(2)}$} to the server.
The server can also calculate the new effective obfuscated distance {\scriptsize$\tilde{d}_{2,1}^{(k+1)}$} and {\scriptsize$UT_1^{(k+1)}$}.
It finds that {\scriptsize$UT_1^{(k+1)}$} is positive, which means $w_1$ wins $t_2$.
The server then alters the allocation table $AL$ by setting the winner of $t_2$ as $w_1$ and the winner of $t_1$ as NULL.

In the $(k+2)$-th competition, it is $w_2$'s turn to compete.
$w_2$ can compete for both $t_1$ and $t_3$. He calculates {\scriptsize$UT_2^{(k+2)}[t_1]=v_1-f_d(\tilde{d}_{1,2}^{(k+2)})-f_p(\epsilon_{1,2}^{(z_{k+2})})=2.45>0$} and {\scriptsize$UT_2^{(k+2)}[t_3]=-f_d(\tilde{d}_{3,2}^{(k+2)})-f_p(\epsilon_{3,2}^{(z_{k+2})})+f_d(\tilde{d}_{3,3}^{(k+1)})=-5.03<0$}.
After that, $w_2$ sets {\scriptsize$UT_2^{(k+2)}$} as {\scriptsize$UT_2^{(k+2)}[t_1]$}, which is the maximal positive value in set {\scriptsize$\{UT_2^{(k+2)}[t_1], UT_2^{(k+2)}[t_3]\}$}.
Then, $w_2$ applies to the server for $t_1$ by proposing {\scriptsize$(\hat{d}_{2,1}^{(2)},\epsilon_{2,1}^{(2)})$}.
After similar calculations, the server alters $AL$ by setting the winner of $t_1$ as $w_2$.

In the $(k+3)$-th competition, it is $w_3$'s turn to compete.
$w_3$ can only propose to $t_1$.
However, the value {\scriptsize$UT_3^{(k+3)}=-9.95<0$}.
Therefore, $w_3$ does not compete for any tasks.

These three steps are repeated until all workers do not propose to any tasks (i.e., until the $6$-th competition).
Table~\ref{tab_gt_distance_budget_change} records the changing of effective obfuscated distances and privacy budgets.
The red one (with $UT>0$) means there is a new winner who publishes a new privacy budget and updates the corresponding effective obfuscated distance. The green one (with $UT\leq 0$) means the competitor fails to compete for the task and will publish neither his new obfuscated distance nor his new privacy budget.

\begin{table}[t!]
	\caption{Allocation list from the $k$-th competition.}
	\label{tab_gt_allocation_list}
	\centering
	\scalebox{1}{
		\begin{tabular}{c|ccc}
			\hline
			Task  & $k$-th & $(k+1)$-th & $(k+2)$-th --   $(k+6)$-th \\ \hline
			$t_1$ & $w_1$  & NULL       & $w_2$                      \\
			$t_2$ & $w_2$  & $w_1$      & $w_1$                      \\
			$t_3$ & $w_3$  & $w_3$      & $w_3$                      \\ \hline
		\end{tabular}
	}
\end{table}

\begin{table}[t!]
	\caption{The timeline of effective distances and privacy budgets.}
	\label{tab_gt_distance_budget_change}
	\centering
	\scalebox{0.7}{
		\begin{tabular}{|c|ccccccc|}
			\hline
			\textbf{Pair/Times}          & \multicolumn{1}{c|}{$k$}                         & \multicolumn{1}{c|}{$k+1$}                                          & \multicolumn{1}{c|}{$k+2$}                                             & \multicolumn{1}{c|}{$k+3$}                                            & \multicolumn{1}{c|}{$k+4$}                                            & \multicolumn{1}{c|}{$k+5$}                                             & $k+6$                                            \\ \hline
			\multirow{2}{*}{$(t_1,w_1)$} & \multicolumn{4}{c|}{\multirow{2}{*}{(12.7,0.1)}}                                                                                                                                                                                                                        & \multicolumn{1}{c|}{(12.7,0.1)}                                       & \multicolumn{2}{c|}{\multirow{2}{*}{(12.7,0.1)}}                                                                          \\ \cline{6-6}
			& \multicolumn{4}{c|}{}                                                                                                                                                                                                                                                   & \multicolumn{1}{c|}{\textcolor{green}{(12.4,0.3)}} & \multicolumn{2}{c|}{}                                                                                                     \\ \hline
			\multirow{2}{*}{$(t_1,w_2)$} & \multicolumn{2}{c|}{\multirow{2}{*}{(5.5,4.6)}}                                                                        & \multicolumn{1}{c|}{(5.5,4.6)}                                         & \multicolumn{4}{c|}{\multirow{2}{*}{(5.3,4.65)}}                                                                                                                                                                                                                           \\ \cline{4-4}
			& \multicolumn{2}{c|}{}                                                                                                  & \multicolumn{1}{c|}{\textcolor{red}{(5.3,4.65)}}     & \multicolumn{4}{c|}{}                                                                                                                                                                                                                                                     \\ \hline
			\multirow{2}{*}{$(t_1,w_3)$} & \multicolumn{3}{c|}{\multirow{2}{*}{(9.93,0.1)}}                                                                                                                                                & \multicolumn{1}{c|}{(9.93,0.1)}                                       & \multicolumn{2}{c|}{\multirow{2}{*}{(9.93,0.1)}}                                                                                               & (9.93,0.1)                                       \\ \cline{5-5} \cline{8-8}
			& \multicolumn{3}{c|}{}                                                                                                                                                                           & \multicolumn{1}{c|}{\textcolor{green}{(9.63,0.4)}} & \multicolumn{2}{c|}{}                                                                                                                          & \textcolor{green}{(9.63,0.4)} \\ \hline
			\multirow{2}{*}{$(t_2,w_1)$} & \multicolumn{1}{c|}{\multirow{2}{*}{(4.11,6.99)}} & \multicolumn{1}{c|}{(4.11,6.99)}                                     & \multicolumn{5}{c|}{\multirow{2}{*}{(4.01,7.1)}}                                                                                                                                                                                                                                                                                                   \\ \cline{3-3}
			& \multicolumn{1}{c|}{}                            & \multicolumn{1}{c|}{\textcolor{red}{(4.01,7.1)}} & \multicolumn{5}{c|}{}                                                                                                                                                                                                                                                                                                                              \\ \hline
			\multirow{2}{*}{$(t_2,w_2)$} & \multicolumn{5}{c|}{\multirow{2}{*}{(10.94,0.1)}}                                                                                                                                                                                                                                                                                               & \multicolumn{1}{c|}{(10.94,0.1)}                                       & \multirow{2}{*}{(10.94,0.1)}                     \\ \cline{7-7}
			& \multicolumn{5}{c|}{}                                                                                                                                                                                                                                                                                                                           & \multicolumn{1}{c|}{\textcolor{green}{(10.64,0.2)}} &                                                  \\ \hline
			\multirow{2}{*}{$(t_3,w_2)$} & \multicolumn{2}{c|}{\multirow{2}{*}{(12.71,0.1)}}                                                                      & \multicolumn{1}{c|}{(12.71,0.1)}                                       & \multicolumn{2}{c|}{\multirow{2}{*}{(12.71,0.1)}}                                                                                             & \multicolumn{1}{c|}{(12.71,0.1)}                                       & \multirow{2}{*}{(12.71,0.1)}                     \\ \cline{4-4} \cline{7-7}
			& \multicolumn{2}{c|}{}                                                                                                  & \multicolumn{1}{c|}{\textcolor{green}{(12.51,0.3)}} & \multicolumn{2}{c|}{}                                                                                                                         & \multicolumn{1}{c|}{\textcolor{green}{(12.51,0.3)}} &                                                  \\ \hline
			$(t_3,w_3)$                  & \multicolumn{7}{c|}{(7.78,5.4)}                                                                                                                                                                                                                                                                                                                                                                                                                                             \\ \hline
		\end{tabular}
	}
\end{table}
\end{example}

\noindent\textbf{Convergence Analysis.}
In order to answer the convergence speed of \solutionB{}, we need to know how many rounds it takes to find a pure Nash equilibrium.
For the corresponding potential game of a \approximateBasicProblem{} instance, $\mathcal{G}=<\entity{W},\vectorfont{S},\vectorfont{UT}>$, we assume there is an equivalent game with potential function $\Phi_\mathbb{Z}(\vectorfont{st})=d\cdot\Phi(\vectorfont{st})$, where $d$ is a positive multiplicative factor satisfying that $\Phi_\mathbb{Z}(\vectorfont{st})\in\mathbb{Z}$ for $\forall{\vectorfont{st}\in\vectorfont{S}}$. Let $\vectorfont{st}^{\ast}$ be the best strategy the workers can choose in this \approximateBasicProblem{} game instance. Based on the above assumption, we prove that \solutionB{} executes at most $\Phi_\mathbb{Z}(\vectorfont{st}^{\ast})$ rounds.

\begin{theorem}\label{thrm:GameConv}
\solutionB{} executes at most $\Phi_\mathbb{Z}(\vectorfont{st}^{\ast})$ rounds to achieve a pure Nash equilibrium, where $\Phi_\mathbb{Z}(\vectorfont{st}^{\ast}) = d\cdot\Phi(\vectorfont{st}^{\ast})$ is a scaled potential function with integer value $d$ and $\vectorfont{st}^{\ast}$ is the optimal strategy the workers can choose in the potential \approximateBasicProblem{} game instance.
\end{theorem}
\begin{proof}
We say \solutionB{} converges when no workers deviate from their current strategies. If \solutionB{} has not converged, then at least one worker $w_j$ deviates from his current strategy in each round. Besides the new change strategy $st'_j$ of $w_j$ is better than his current strategy $st_j$. And the change will improve at least 1 (i.e., $\Phi_\mathbb{Z}(st'_i,\vectorfont{s}_{-i})-\Phi_\mathbb{Z}(st_i,\vectorfont{s}_{-i})\geq 1$) for potential games. Because the maximum value of scaled potential function is $\Phi_\mathbb{Z}(\vectorfont{st}^{\ast})$, and the total utility is always positive, \solutionB{} needs at most $\Phi_\mathbb{Z}(\vectorfont{st}^{\ast})$ rounds to converge to a pure Nash equilibrium.
\end{proof}

\noindent\textbf{Quality Analysis.}
Since the distance in our game is rather real distance than effective obfuscated distance, here we give the upper bound of expectation of \emph{price of stability} (EPoS) and the lower bound of expectation of \emph{price of anarchy} (EPoA).
Let \vspace{-3ex}
{\scriptsize\begin{align}
	U_{j}^{L}(i) &= v_i-f_d(d_{i,j})-f_p(\sum_{t_k\in R_j}sum(\vectorfont{\epsilon_{k,j}}))\notag\\
	U_{j}^{H}(i) &= v_i-f_d(d_{i,j})-f_p(min(\vectorfont{\epsilon_{i,j}}))\notag\\
	U_{min}^{+}(i) &= \left\{
	\begin{array}{ll}
		\min\limits_{R_j\ni t_i, U_{j}^{L}(i)>0} U_{j}^{L}(i), &\textrm{if there exists}\;U_{j}^{L}(i)>0   \\
		0, &\textrm{otherwise}
	\end{array}
	\right.\notag\\
	U_{max}^{+}(i) &= \left\{
	\begin{array}{ll}
		\max\limits_{R_j\ni t_i} U_{j}^{H}(i),&\textrm{if there exists}\;U_{j}^{H}(i)>0   \\
		0,&\textrm{otherwise}
	\end{array}
	\right.\notag
\end{align}}

Then we have Theorem~\ref{thrm:Pos_PoA} as follows.
\begin{theorem}\label{thrm:Pos_PoA}
In the strategic game of \solutionB{}, the lower bound of EPoA is {\scriptsize$\frac{\sum_{t_i\in\entity{T}}U_{min}^{+}(i)}{\sum_{t_i\in\entity{T}}U_{max}^{+}(i)}$ $(\sum_{t_i\in\entity{T}}U_{max}^{+}(i)\neq 0)$}
and the upper bound of EPoS is 1.
\end{theorem}

Please refer to details of the proof of Theorem~\ref{thrm:Pos_PoA} in Appendix B.

\begin{theorem}\label{thrm:Privacy_PGT}
\solutionB{} satisfies {\scriptsize$(\sum_{t_i\in R_j}\vectorfont{b_{i,j}\epsilon_{i,j}}r_j)$}-local differential privacy for each worker $w_j$.
\end{theorem}

The proof is similar to Theorem~\ref{SolutionA_PA}, and please refer to the details in Appendix C.

\begin{figure}[t!]\centering\vspace{-3ex}
	\subfigure[][{\scriptsize Order location distribution}]{
		\scalebox{0.25}[0.25]{\includegraphics{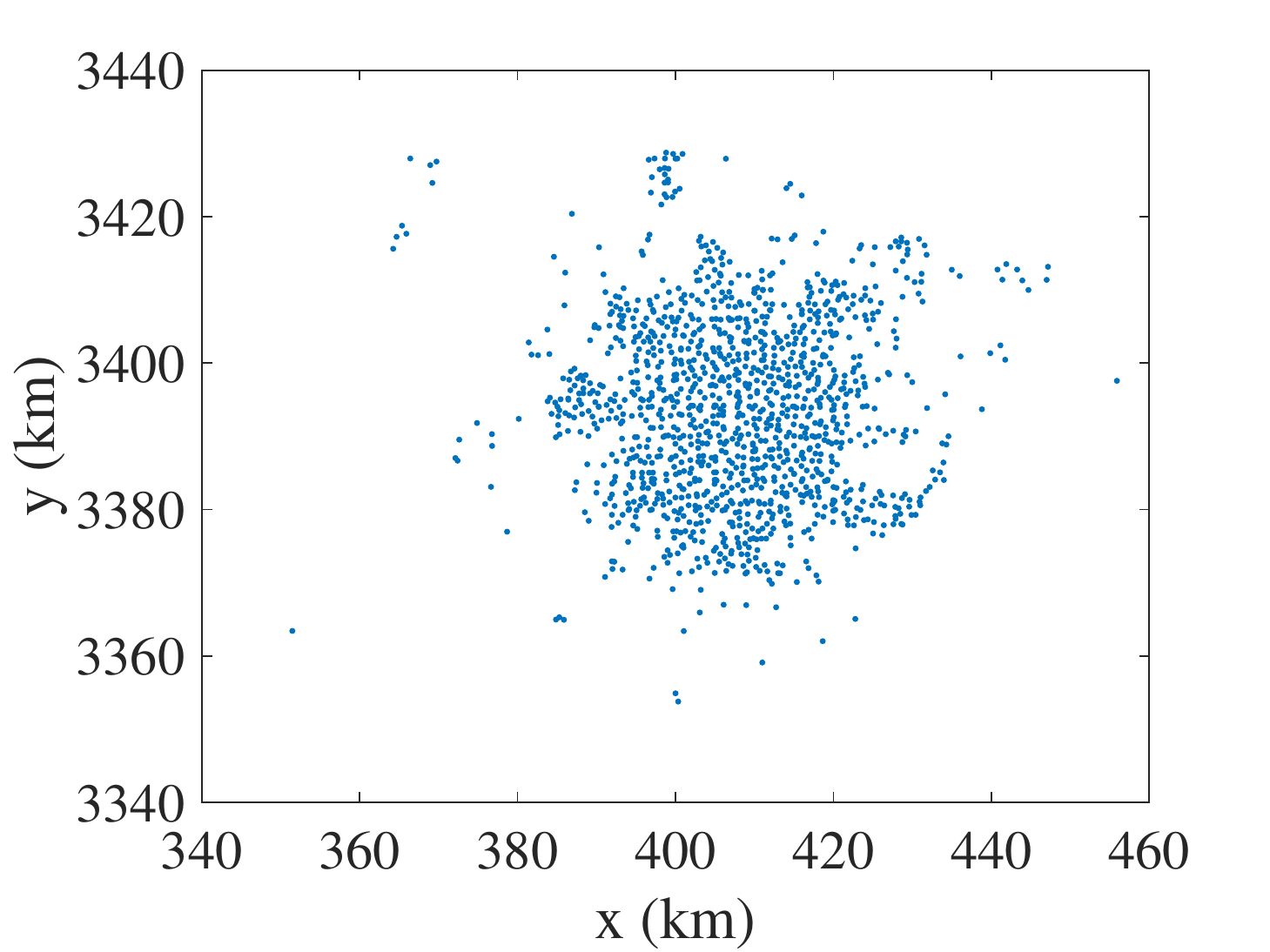}}
		\label{subfig:order_distribution}}
	\subfigure[][{\scriptsize Taxi location distribution}]{
		\scalebox{0.25}[0.25]{\includegraphics{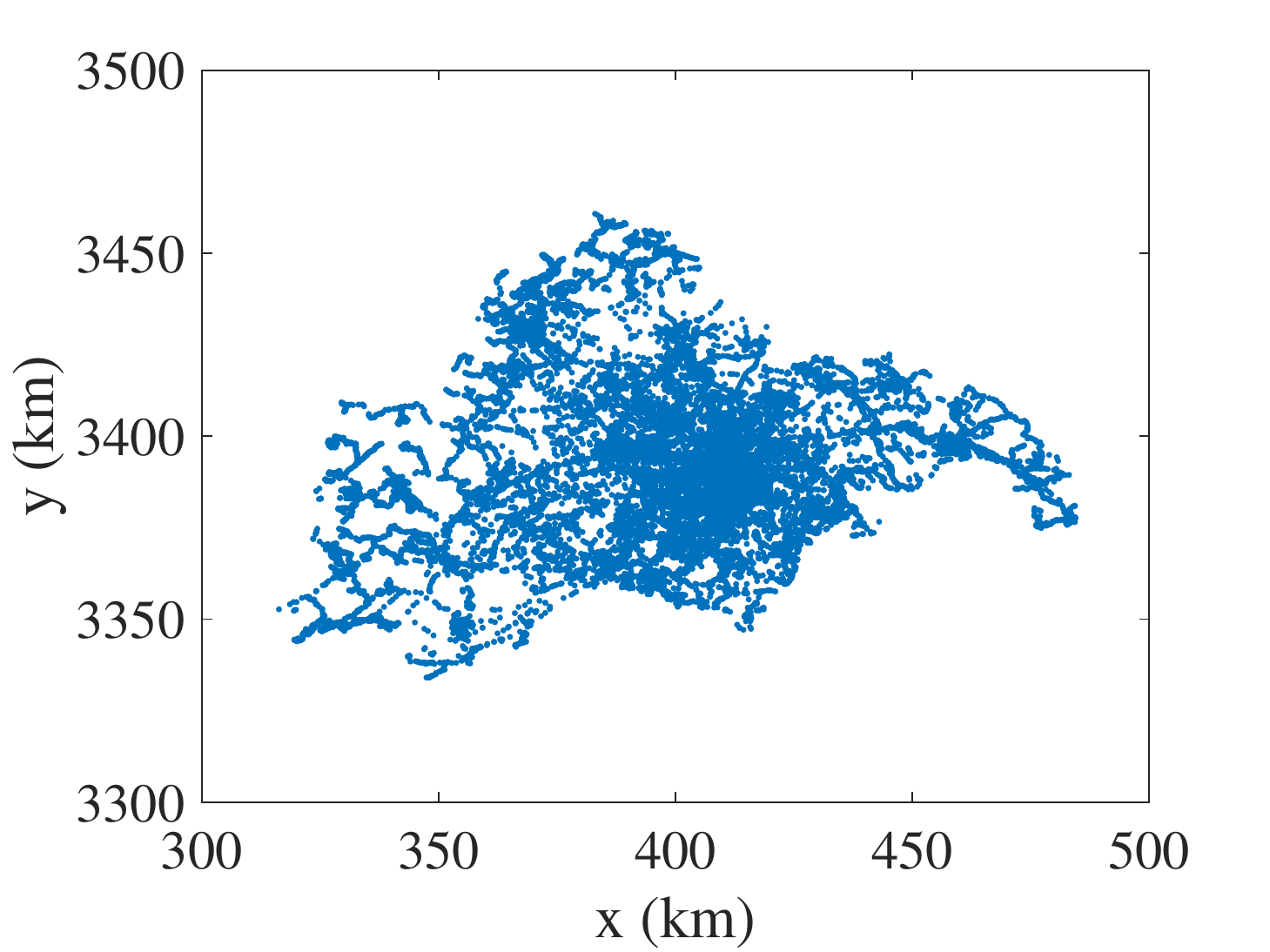}}
		\label{subfig:taxi_distribution}}\figureCaptionMargin
	\caption{\small Orders and taxies of Chengdu Didi data set.}\figureBelowMargin
	\label{fig_experiment_real_dataset}
\end{figure}

\section{Experiment}\label{Experiment}
\subsection{Data Sets}
We test our mechanisms in real and synthetic  data sets.

\noindent\textbf{Real Data Set.}
We use Didi Chuxing\cite{DidiLink} in Chengdu, China, as our real data set.
We choose the day with the most requests for evaluation (November 18, 2016 in Chengdu) and perform the same preprocess in the existing work~\cite{DBLP:journals/pvldb/TongZZCYX18}, which is denoted as \emph{chengdu}.

\emph{Chengdu} contains $259347$ orders and $30000$ taxis.
Each order tuple is a taxi request consisting of a release time, a pickup location, a drop-off location, and some passengers.
Each taxi tuple is a basic message consisting of the original location of the taxi and its capacity.
The location distribution of taxis is shown in Figure~\ref{subfig:taxi_distribution}.

\noindent\textbf{Synthetic Data Set.}
We generate two data sets with 2-dimensional uniform distribution and normal distribution, respectively.
For the uniform distribution data set, we randomly generate {$300$k} points for tasks and {$900$k} for workers in a plane with a range of $100\times 100$.
Each point follows a 2-dimensional uniform distribution with an average of 0.

For the normal distribution data set, we generate $300$k and $900$k points for tasks and workers, respectively. The expectation and variance for all points are 0 and 150, respectively.

\subsection{Experimental Setup}
We split the orders into batches by timestamp.
Each batch contains at most 1000 orders. Figure~\ref{subfig:order_distribution} is a batch example of the order distribution.
We also split the taxis into ten groups for the real data set, each containing 3000 taxis.
We use each worker group circularly for each batch.
We set the {pickup} locations of orders as task locations and the {original} locations of taxis as worker locations.

Let $S_T$ and $S_W$ be two sets for tasks and workers.
We define the value $p_{wt}=\frac{|S_W|}{|S_T|}$ as \emph{worker-task ratio} which stands for the ratio between {the} worker number and {the} task number.

We alter the method in Ref~\cite{DBLP:journals/tmc/WangHLWWYQ19} by constraining the workers' proposing range {in his service area} and replacing PCF with PPCF in order to get reasonable comparison with our \solutionA{}.
We denote this altering method in Ref~\cite{DBLP:journals/tmc/WangHLWWYQ19} as \solutionCMPTotalName{} (\EXPSolutionCMP{}).
The {difference} between \solutionA{} and \solutionCMP{} is the optimization objective.
In \solutionCMP{}, the goal is to minimize all the travelling distance on the platform, which only {considers} the distance variable.
However, in \solutionA{}, the goal is to maximize the utility function of the platform, which considers the task value, travel distance and privacy budget.

\begin{table}[t!]\vspace{-2ex}
	\begin{center}
		{\small\scriptsize  \vspace{1ex}
				\caption{\small Methods.} \label{all_methods}
				\begin{tabular}{c|c|c|c}\hline
					& Private version           & Non-Private version   &   Non-PPCF version\\ \hline
					Distance Elimination & \EXPSolutionCMP{} \cite{DBLP:journals/tmc/WangHLWWYQ19} & \EXPSolutionNPCMP{}  & \EXPSolutionCMP-nppcf\\
					Utility Elimination  & \EXPSolutionA{}      & \EXPSolutionNPA{}  & \EXPSolutionA{}-nppcf   \\
					Game Theory          & \EXPSolutionB{}      & \EXPSolutionNPB{}  & --- \\
					Greedy               & ---                  & \EXPSolutionNPGD{} & --- \\ \hline
				\end{tabular}
		}\vspace{-1ex}
	\end{center}
\end{table}

\begin{table}[t!]\vspace{-2ex}
	\begin{center}
		{\small\scriptsize  \vspace{1ex}
			\caption{\small Experimental settings.} \label{tab:settings}
			\label{tab_experiment_parameter_settings}
			\begin{tabular}{l|l}\hline
				{\bf \qquad Parameters \qquad \quad } & {\bf \qquad  \qquad Values \qquad } \\ \hline
				worker-task ratio           &       $1, 1.5, \textbf{2}, 2.5, 3$ \\
				task values                 &       $1.5, 3, \textbf{4.5}, 6, 7.5$\\
				worker range                &       $0.8, 1.1, \textbf{1.4}, 1.7, 2.0$\\
				\multirow{2}{*}{privacy budget}              &       $[0.5,0.75], [0.75,1.00], [1.00,1.25],$ \\
				\;                &       $[1.25,1.50], [1.50, 1.75]; \textbf{[0.5,1.75]}$\\
				privacy budget group size   &       $\textbf{7}$\\
				\hline
			\end{tabular}
		}\vspace{-1ex}
	\end{center}
\end{table}

We compare our \solutionA{} and \solutionB{} with \EXPSolutionCMP{}.
Besides, we construct the {non-private} solution of each {private} solution by eliminating the privacy budget cost in the utility function and replacing obfuscated distance with real distance. These {non-private} solutions are \solutionNPATotalName{} (\EXPSolutionNPA{}), \solutionNPBTotalName{} (\EXPSolutionNPB{}), \solutionNPCMPTotalName{} (\EXPSolutionNPCMP{}) and \solutionNPGDTotalName{} (\EXPSolutionNPGD). Here, GRD always greedily chooses the current best worker-task pair (with the highest utility) for each worker.
We also construct the non-PPCF solution of \EXPSolutionA{} and \EXPSolutionCMP{} by replacing the PPCF part with the PCF part.
We denote these non-PPCF solutions as \EXPSolutionA{}-nppcf and \EXPSolutionCMP{}-nppcf.
We compare all these methods above and summarize them in Table~\ref{all_methods}.

We show the parameter settings in table~\ref{tab_experiment_parameter_settings}, where the default values are marked in bold.
As for distance value function $f_d$ and privacy budget value function $f_p$, we model them as linear functions and use $f_d(x)=\alpha x$ and $f_p(x)=\beta x$ in our experiment. We set $\alpha=1$ and $\beta=1$.

\begin{figure}[t!]\centering\vspace{-3ex}
	\subfigure{
		\scalebox{0.33}[0.33]{\includegraphics{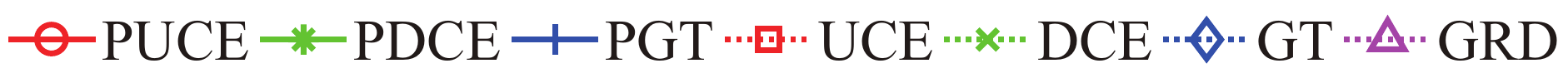}}}\hfill\\
	\addtocounter{subfigure}{-1}\vspace{-2.5ex}
	\subfigure[][{\scriptsize \emph{chengdu}}]{	\scalebox{0.25}[0.25]{\includegraphics{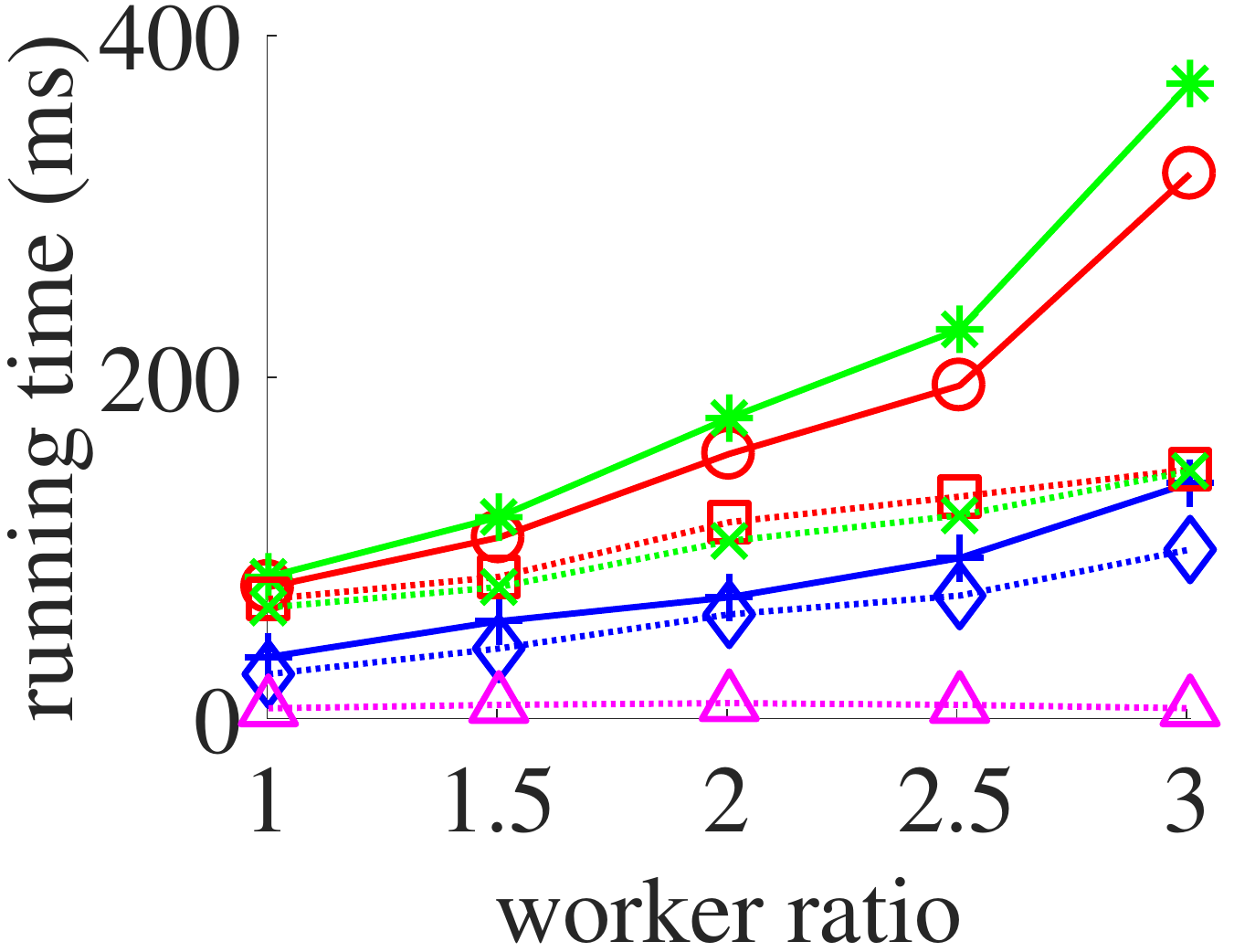}}
		\label{subfig:task_value_time_cost_chengdu}}
	\subfigure[][{\scriptsize \emph{normal}}]{
		\scalebox{0.25}[0.25]{\includegraphics{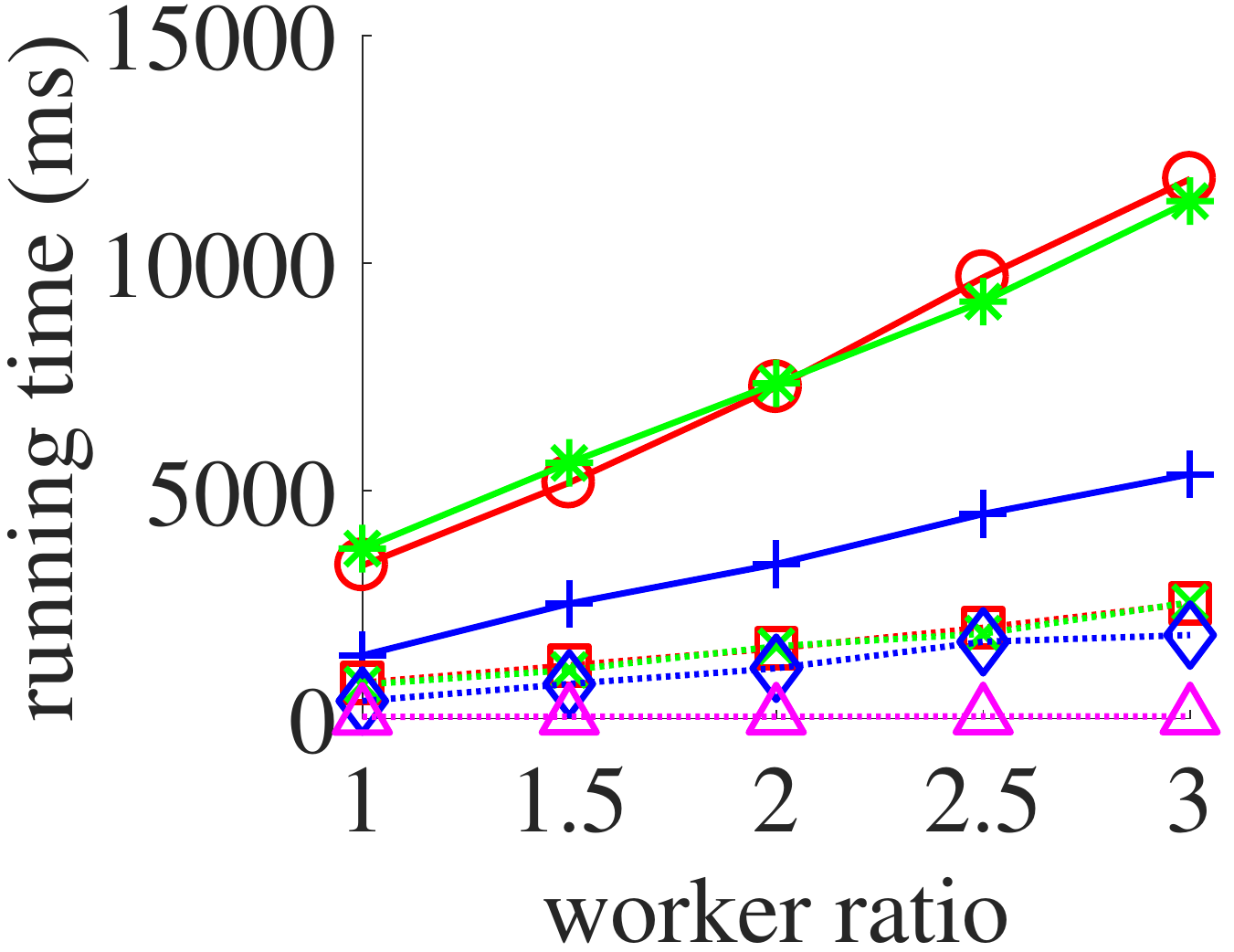}}
		\label{subfig:worker_distribution}}\figureCaptionMargin
	\caption{\small The impact of the worker ratio on the time cost.}\vspace{-2ex}
	\label{fig:worker_ratio_time_cost}
\end{figure}

\begin{figure}[t!]\centering 
\subfigure{
  \scalebox{0.33}[0.33]{\includegraphics{bar2-eps-converted-to.pdf}}}\hfill\\
\addtocounter{subfigure}{-1}\vspace{-2.5ex}
	\subfigure[][{\scriptsize Average Utility}]{
		\scalebox{0.25}[0.25]{\includegraphics{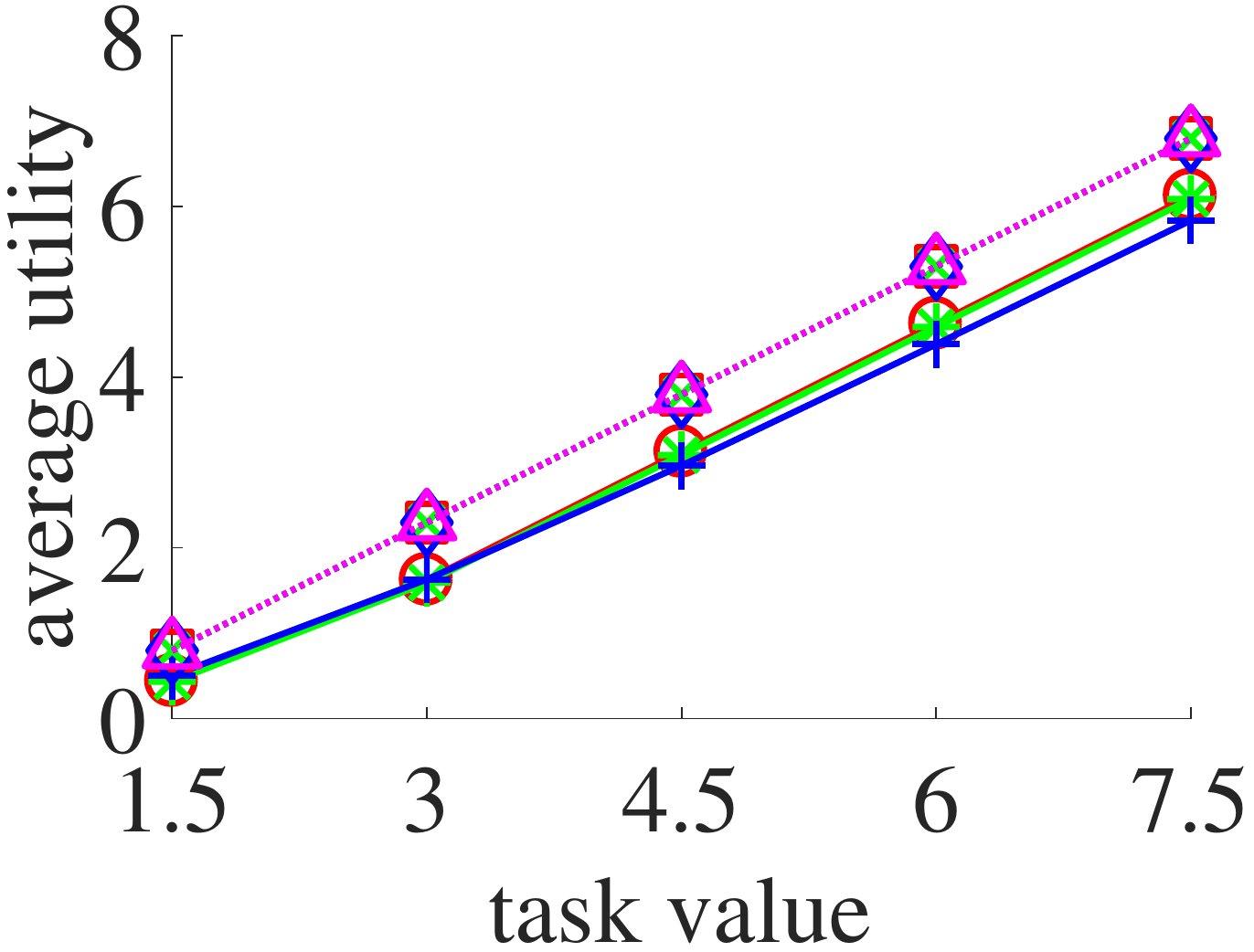}}
		\label{subfig:task_value_utility_chengdu}}
	\subfigure[][{\scriptsize Relative Deviation of Utility}]{
		\scalebox{0.25}[0.25]{\includegraphics{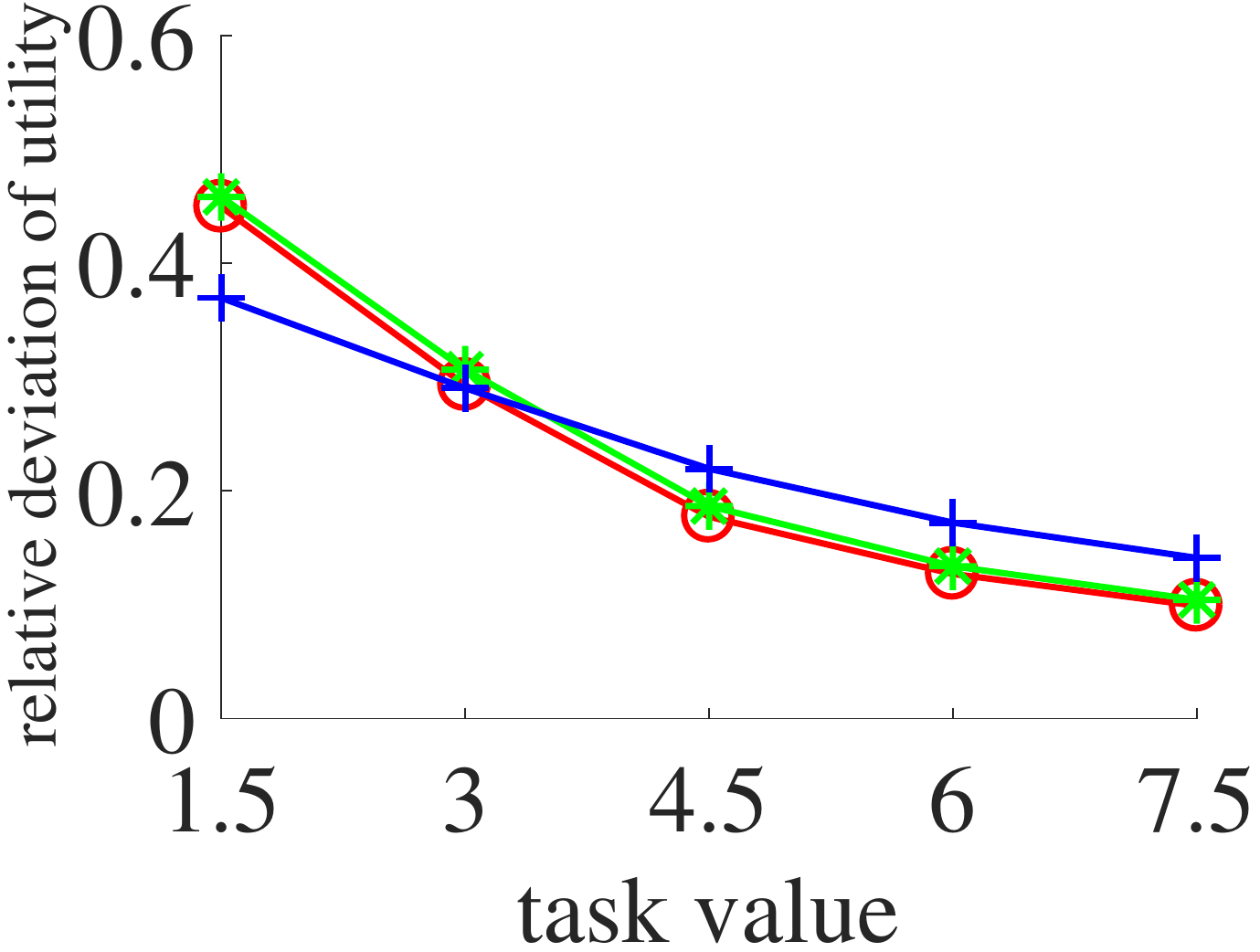}}
		\label{subfig:task_value_utility_deviation_chengdu}}\figureCaptionMargin
	\caption{\small The impact of the task value on the utility  (\emph{chengdu}).}
	\label{fig:task_value_utility_chengdu}
\end{figure}

We run our experiment on an Intel(R) Xeon(R) Silver 4210R CPU @ 2.4GHz with 128 GB RAM in Java.

\subsection{Measures}
We design a utility-based empirical measure of the efficiency of our proposed mechanisms.

\textbf{Average Utility:}
We define the average utility $U_\textrm{AVG}$ as $\frac{\sum_{(i,j)\in M}U_j(i)}{|M|}$, which means the average utility value of a successful task-worker pair.

\textbf{Relative Deviation of Utility:}
Let the utility of {non-private} solutions be $U_\textrm{NP}$ and privacy ones be $U_\textrm{P}$.
We define the relative deviation of utility $U_\textrm{RD}$ as $\frac{U_\textrm{NP}-U_\textrm{P}}{U_\textrm{NP}}$.

\textbf{Average Travel Distance:} We define the average travel distance $D_\textrm{AVG}$ as $\frac{\sum_{(i,j)\in M}d_{i,j}}{|M|}$, which means the average travel distance of a successful task-worker pair.

\textbf{Relative Deviation of Distance:}
Let the distance of {non-private} solutions be $D_\textrm{NP}$ and privacy ones be $D_\textrm{P}$.
We define the relative deviation of distance $D_\textrm{RD}$ as $\frac{D_{\textrm{P}} - D_\textrm{NP}}{D_\textrm{NP}}$.

\subsection{Experimental Result}

\subsubsection{Time Cost}
Figure~\ref{fig:worker_ratio_time_cost} shows the time cost on different worker ratio from 1 to 3 while the other parameters are in the default values in Table~\ref{tab_experiment_parameter_settings}.
We can see that the time cost increases linearly with the worker ratio.
That is because when we fix the task quantity, as the worker ratio becomes larger, the competition between workers will become more fierce, and it will cost more time to finish the whole competition.

Besides, we can find that \EXPSolutionA{} costs nearly the same time over the change of worker ratio.
\EXPSolutionB{} costs much less time than \EXPSolutionA{} and \EXPSolutionCMP{}.
Compared with \EXPSolutionCMP{}, \EXPSolutionB{} costs about 52\%--63\% less time in \emph{chengdu} and 50\%--63\% in \emph{normal}.

\begin{figure}[t!]\centering\vspace{-3ex}
	\subfigure{
		\scalebox{0.33}[0.33]{\includegraphics{bar2-eps-converted-to.pdf}}}\hfill\\
	\addtocounter{subfigure}{-1}\vspace{-2.5ex}
	\subfigure[][{\scriptsize Average Utility}]{
		\scalebox{0.25}[0.25]{\includegraphics{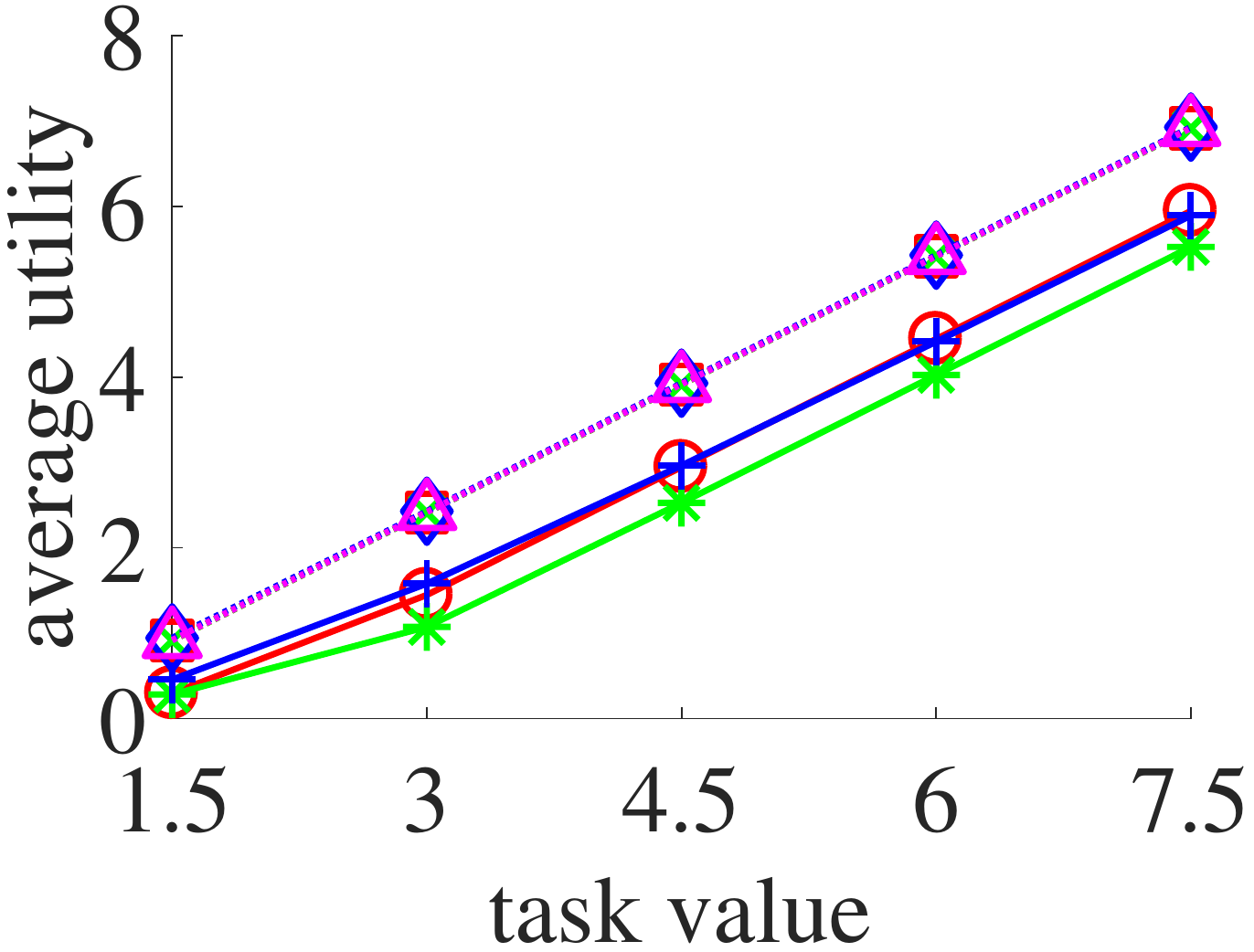}}
		\label{subfig:task_value_utility_normal}}
	\subfigure[][{\scriptsize Relative Deviation of Utility}]{
		\scalebox{0.25}[0.25]{\includegraphics{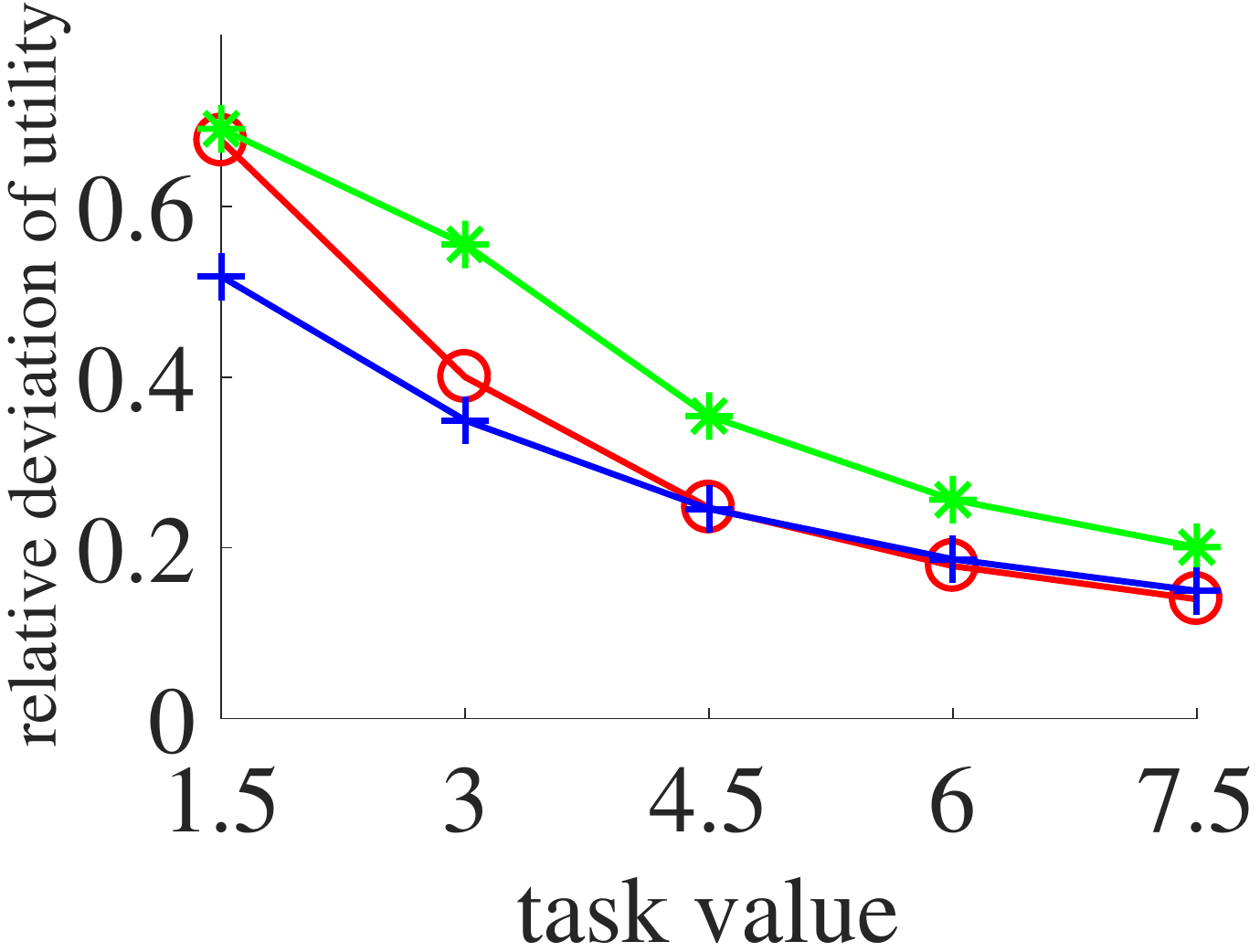}}
		\label{subfig:task_value_utility_deviation_normal}}\figureCaptionMargin
	\caption{\small The impact of the task value on the utility  (\emph{normal}).}\vspace{-2ex}
	\label{fig:task_value_utility_normal}
\end{figure}

\begin{figure}[t!]\centering 
\subfigure{
  \scalebox{0.33}[0.33]{\includegraphics{bar2-eps-converted-to.pdf}}}\hfill\\
\addtocounter{subfigure}{-1}\vspace{-2.5ex}
	\subfigure[][{\scriptsize Average Utility}]{
		\scalebox{0.25}[0.25]{\includegraphics{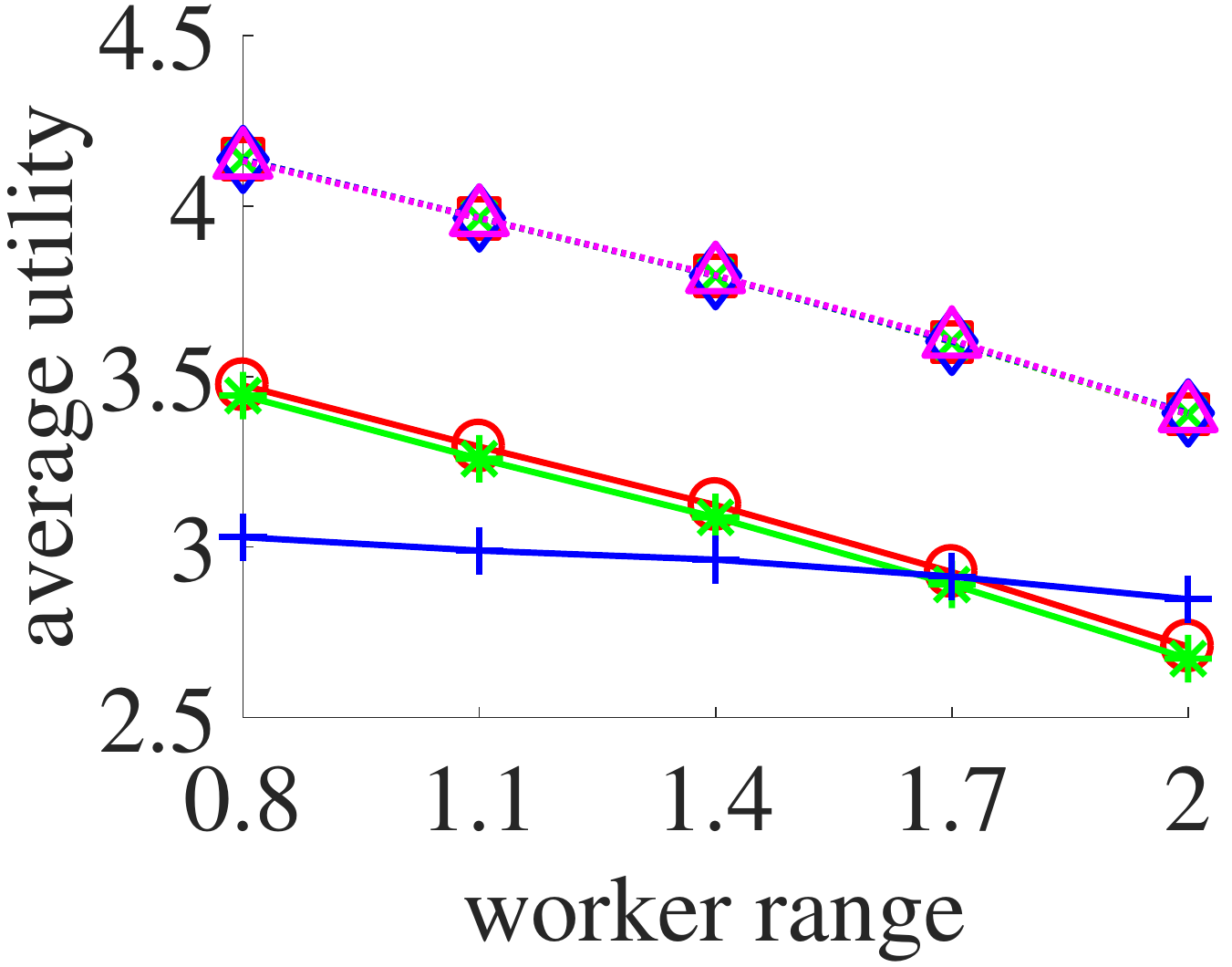}}
		\label{subfig:worker_range_utility_chengdu}}
	\subfigure[][{\scriptsize Relative Deviation of Utility}]{
		\scalebox{0.25}[0.25]{\includegraphics{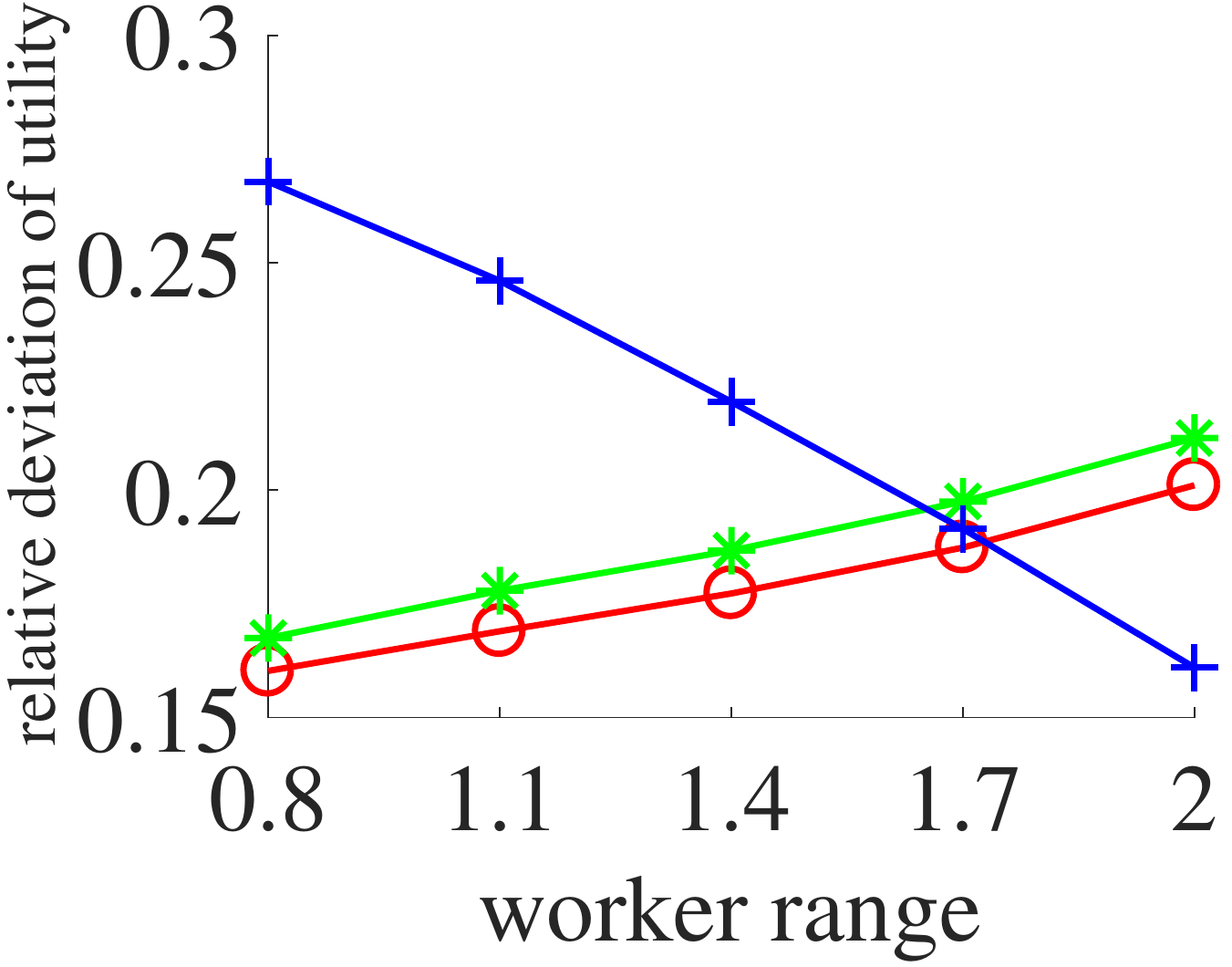}}
		\label{subfig:worker_range_utility_deviation_chengdu}}\figureCaptionMargin
	\caption{\small The impact of the worker range on the utility  (\emph{chengdu}).}
	\label{fig:worker_range_utility_chengdu}
\end{figure}

\subsubsection{Average Utility}
Figure~\ref{fig:task_value_utility_chengdu} and \ref{fig:task_value_utility_normal} show the relation between the utility and the task value on \emph{chengdu} and \emph{normal} respectively.
We change the task value from 1.5 to 7.5 and set other parameters as the default values.

In Figure~\ref{subfig:task_value_utility_chengdu} and \ref{subfig:task_value_utility_normal}, the utility increases approximately linear with the task value.
We can see that \EXPSolutionB{} performs worse than \EXPSolutionCMP{} slightly in \emph{chengdu}, but better in \emph{normal}.
\EXPSolutionB{} even performs better than \EXPSolutionA{} in \emph{normal}.
The reason is that \EXPSolutionB{} takes advantage over the other two when the workers' service area contains many tasks.
The data in \emph{chengdu} is of road network data which is sparser than that in \emph{normal}.
Thus {when we fix} the service area, a worker in \emph{chengdu} can propose to fewer tasks than that in \emph{normal} on average, which leads to poor utility for \EXPSolutionB{}. (The following experiment result in Figure~\ref{subfig:worker_range_utility_chengdu} proves this inference.)
\EXPSolutionA{} performs better than \EXPSolutionCMP{} in both of the two data sets.
The relative deviation of utility impacted by the task value is shown in Figure~\ref{subfig:task_value_utility_deviation_chengdu} and \ref{subfig:task_value_utility_deviation_normal}.
We can see that the relative deviation of utility decreases with the task value increase from 1.5 to 7.5, which means the absolute deviation between the private and non-private solutions keeps nearly stable.
And when the task value becomes larger and larger, the utility of private solutions equals that of non-private solutions asymptotically.

Figure~\ref{fig:worker_range_utility_chengdu} and \ref{fig:worker_range_utility_normal} show the relation between the utility and the worker range on \emph{chengdu} and \emph{normal} respectively.
The worker service area (denoted as worker range) increases from 0.8 to 2, and the other parameters are set as default values.
The average utility depends on the total utility and the matching quantity.
Specifically, in Figure~\ref{subfig:worker_range_utility_chengdu}, the average utility of all solutions decreases when the worker range increases from 0.8 to 2.
It is because when worker service areas become larger, more workers (who have no task to propose to in some small range conditions, denoting them as $\mathcal{W}_L$) will be able to propose to some tasks.
With the ratio of $\mathcal{W}_L$ becoming larger, the average distance to all matching tasks becomes larger, making the average utility smaller.

Besides, we can see that the utility of \EXPSolutionB{} decreases slower than both \EXPSolutionA{} and \EXPSolutionCMP{}. The utility of \EXPSolutionB{} is no less {than} $88\%$ when the worker range is no more than $1.6$. And as the worker range increases, the utility of \EXPSolutionB{} will exceed the other two.
The reason why \EXPSolutionB{} keeps lower decrease is that \EXPSolutionB{} can avoid ineffective competition. When the service area becomes larger, the competition becomes more intense, and the advantage of \EXPSolutionB{} becomes more apparent.

Figure~\ref{subfig:worker_range_utility_deviation_chengdu} shows the relative deviation of utility affected by the worker range. We can see that the utility of \EXPSolutionB{} will tend to that of its {non-private} solution when the worker range becomes larger and larger.
However \EXPSolutionA{} and \EXPSolutionCMP{} deviate more as the worker range becomes larger.
That is because when the worker's service area becomes larger, it has a greater possibility of disturbing a large real distance to a small obfuscated distance or a small real distance to a large obfuscated distance.
Without the guarantee of total utility function $ST$, the total proposing workers' utilities in \EXPSolutionA{} and \EXPSolutionCMP{} decrease dramatically when the worker range increases.

\begin{figure}[t!]\centering\vspace{-3ex}
	\subfigure{
		\scalebox{0.33}[0.33]{\includegraphics{bar2-eps-converted-to.pdf}}}\hfill\\
	\addtocounter{subfigure}{-1}\vspace{-2.5ex}
	\subfigure[][{\scriptsize Average Utility}]{
		\scalebox{0.25}[0.25]{\includegraphics{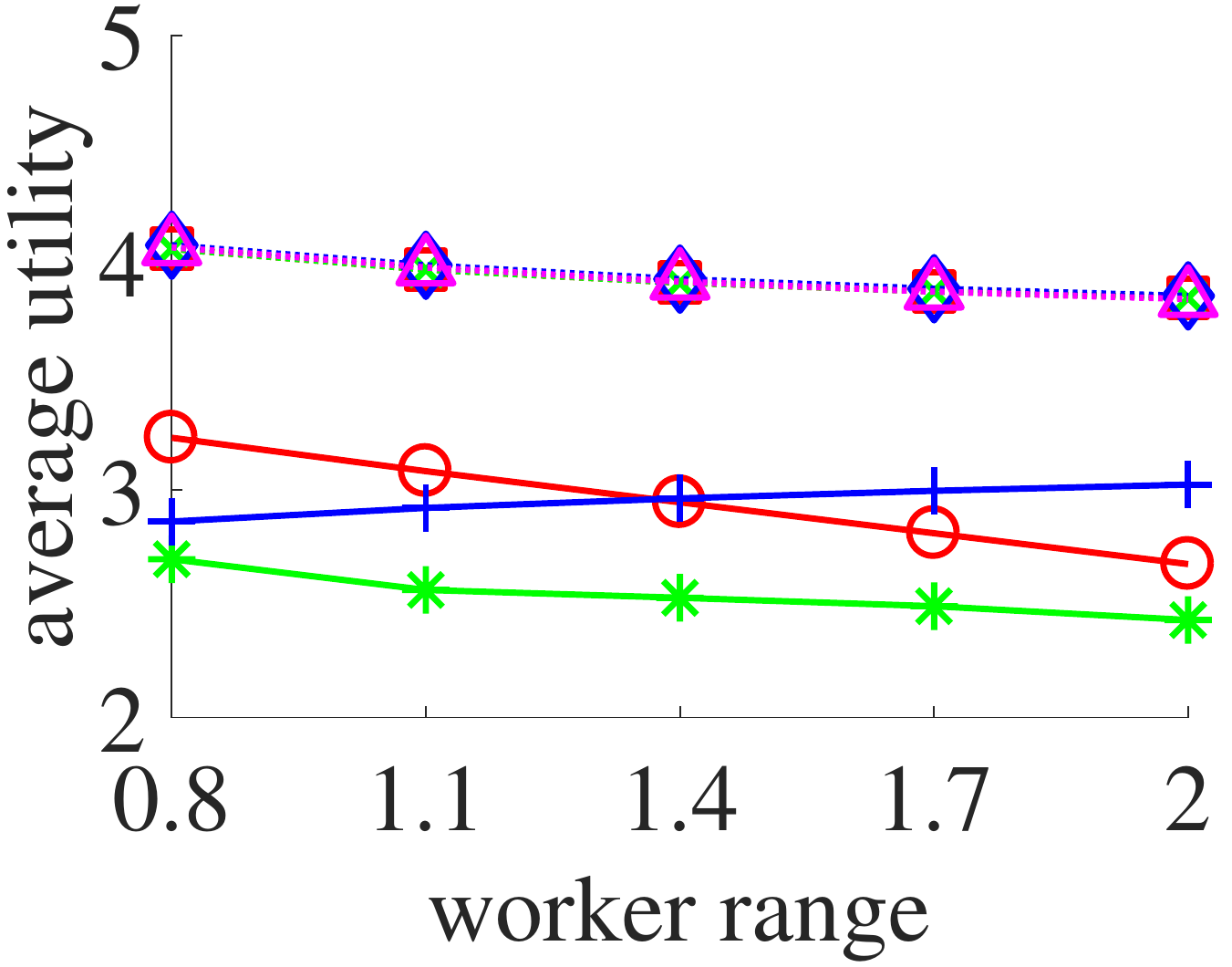}}
		\label{subfig:worker_range_utility_normal}}
	\subfigure[][{\scriptsize Relative Deviation of Utility}]{
		\scalebox{0.25}[0.25]{\includegraphics{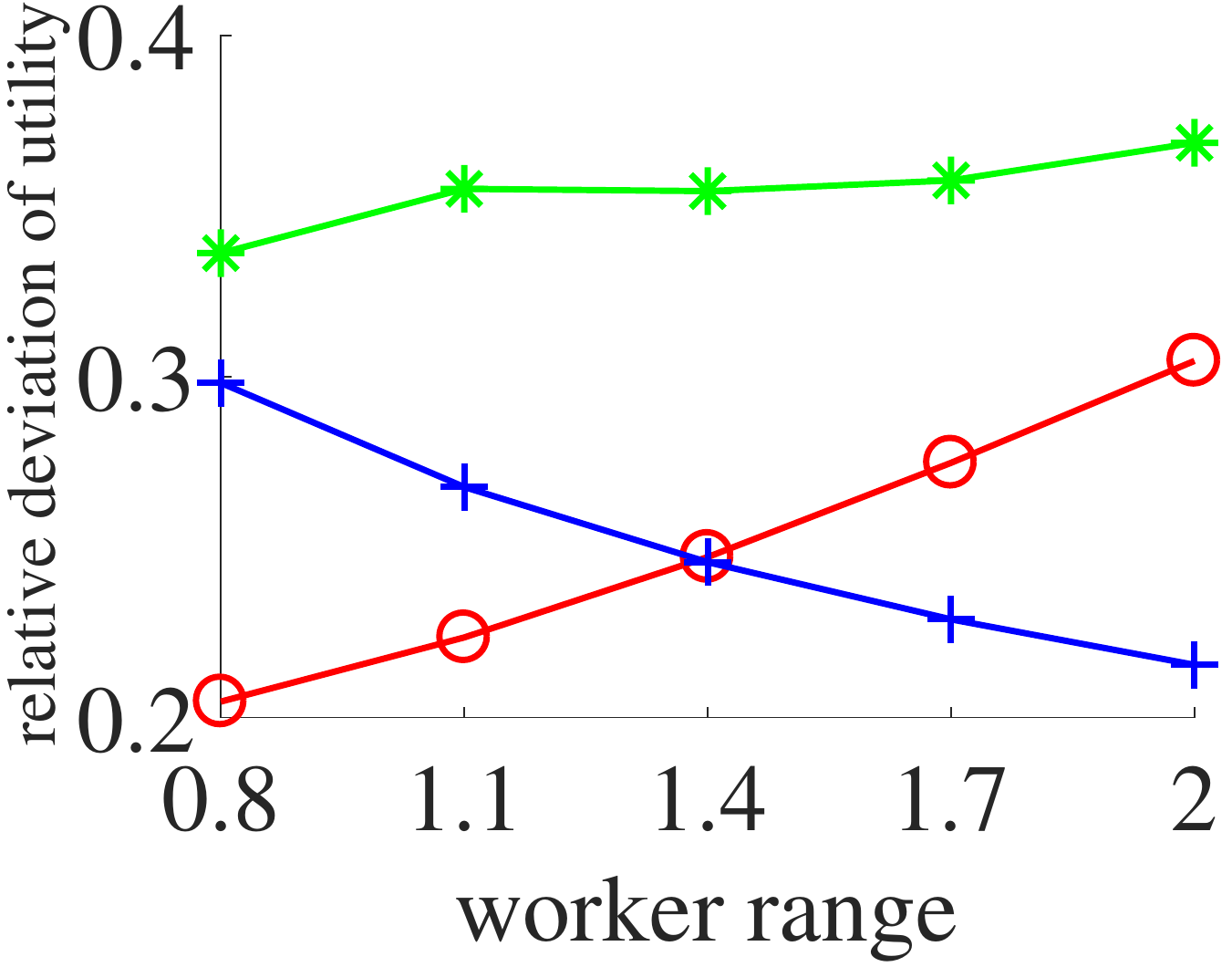}}
		\label{subfig:worker_range_utility_deviation_normal}}\vspace{-1ex}
	\caption{\small The impact of the worker range on the utility  (\emph{normal}).}\vspace{-2ex}
	\label{fig:worker_range_utility_normal}
\end{figure}

\begin{figure}[t!]\centering 
\subfigure{
  \scalebox{0.33}[0.33]{\includegraphics{bar2-eps-converted-to.pdf}}}\hfill\\
\addtocounter{subfigure}{-1}\vspace{-2.5ex}
	\subfigure[][{\scriptsize Average Utility}]{
		\scalebox{0.25}[0.25]{\includegraphics{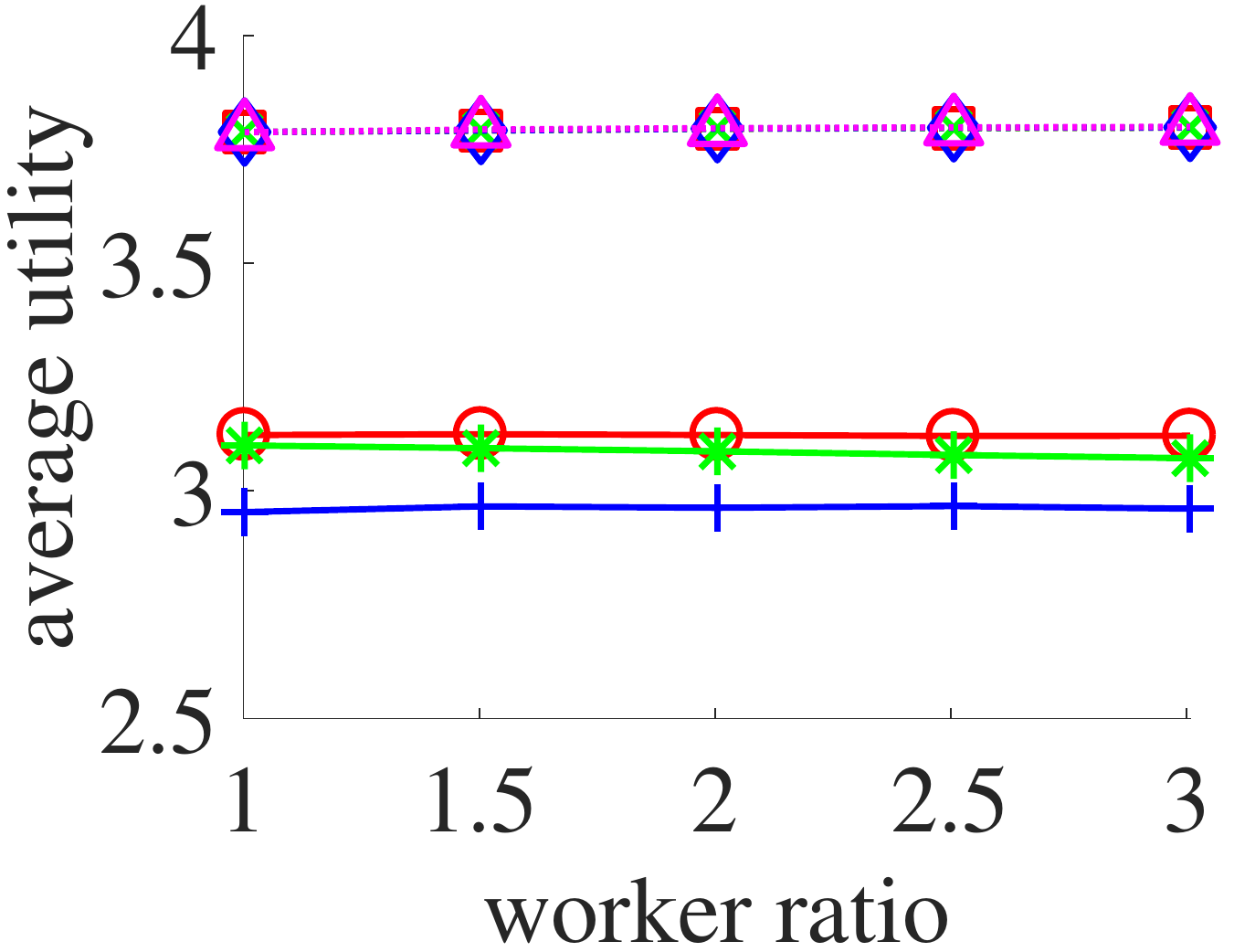}}
		\label{subfig:worker_ratio_utility_chengdu}}
	\subfigure[][{\scriptsize Relative Deviation of Utility}]{
		\scalebox{0.25}[0.25]{\includegraphics{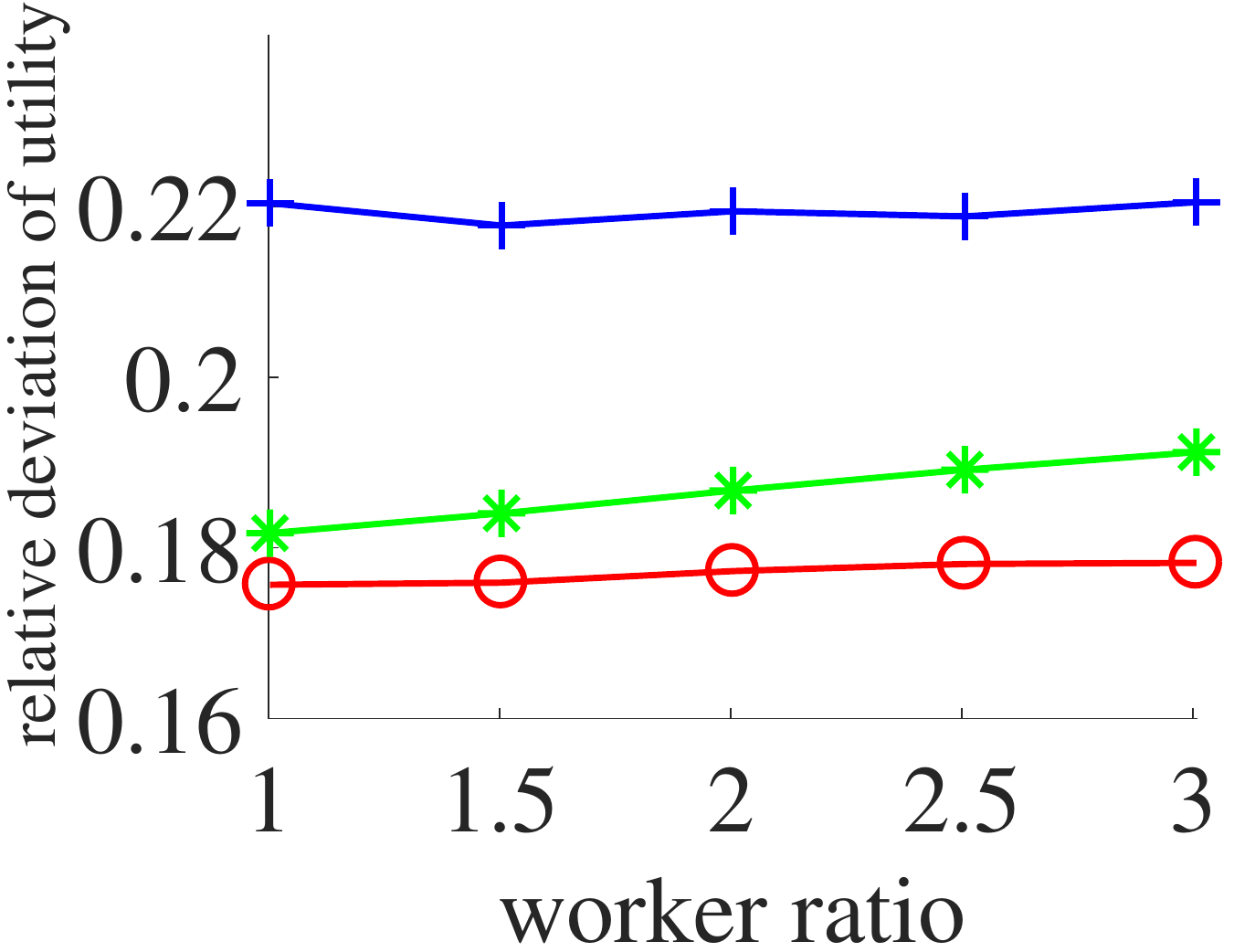}}
		\label{subfig:worker_ratio_utility_deviation_chengdu}}\figureCaptionMargin
	\caption{\small The impact of the worker ratio on the utility  (\emph{chengdu}).}
	\label{fig:worker_ratio_utility_chengdu}
\end{figure}

From Figure~\ref{subfig:worker_range_utility_normal}, we can get the similar conclusion to that in Figure~\ref{subfig:worker_range_utility_chengdu}.
Besides, we can find when the worker range becomes large enough, the decline rate of average utility for \EXPSolutionA{} and \EXPSolutionCMP{} tend to be small.
That is because being too far away will make the utility value non-positive, and the server will not choose.
The average utility of \EXPSolutionB{} increases slightly, which is $16\%$ larger than \EXPSolutionCMP{} on average.
That is because \EXPSolutionB{} can increase the total utility more rapidly than the matching quantity.

Figure~\ref{fig:worker_ratio_utility_chengdu} and \ref{fig:worker_ratio_utility_normal} show the relation between the utility and the worker ratio.
From figure~\ref{subfig:worker_ratio_utility_chengdu} and \ref{subfig:worker_ratio_utility_normal}, we can see that the worker ratio does not affect the average utility very much.
That is because the increase of workers does not significantly increase proposing workers.
Besides, we can see that \EXPSolutionA{} always keeps a higher average utility than \EXPSolutionCMP{}. And \EXPSolutionB{} performs worse than \EXPSolutionCMP{} in \emph{chengdu} but better in \emph{normal}.

\subsubsection{Average Travel Distance}
Figure~\ref{fig:task_value_distance_chengdu} to Figure~\ref{fig:worker_ratio_distance_normal} show the influence of the task value, worker range and worker ratio on the distance.
\EXPSolutionCMP{} is better than \EXPSolutionA{} and \EXPSolutionB{} in most cases.
That is because the goal of \EXPSolutionCMP{} is only to minimize the total travel distance on the platform without considering task value and privacy budget cost.
Besides, we can see that different data sets lead to different comparison results for \EXPSolutionA{}, \EXPSolutionB{} and \EXPSolutionCMP{}.
The average travel distance of \EXPSolutionCMP{} on \emph{normal} outperforms the other two on \emph{chengdu}.

\begin{figure}[t!]\centering\vspace{-3ex}
	\subfigure{
		\scalebox{0.33}[0.33]{\includegraphics{bar2-eps-converted-to.pdf}}}\hfill\\
	\addtocounter{subfigure}{-1}\vspace{-2.5ex}
	\subfigure[][{\scriptsize Average Utility}]{
		\scalebox{0.25}[0.25]{\includegraphics{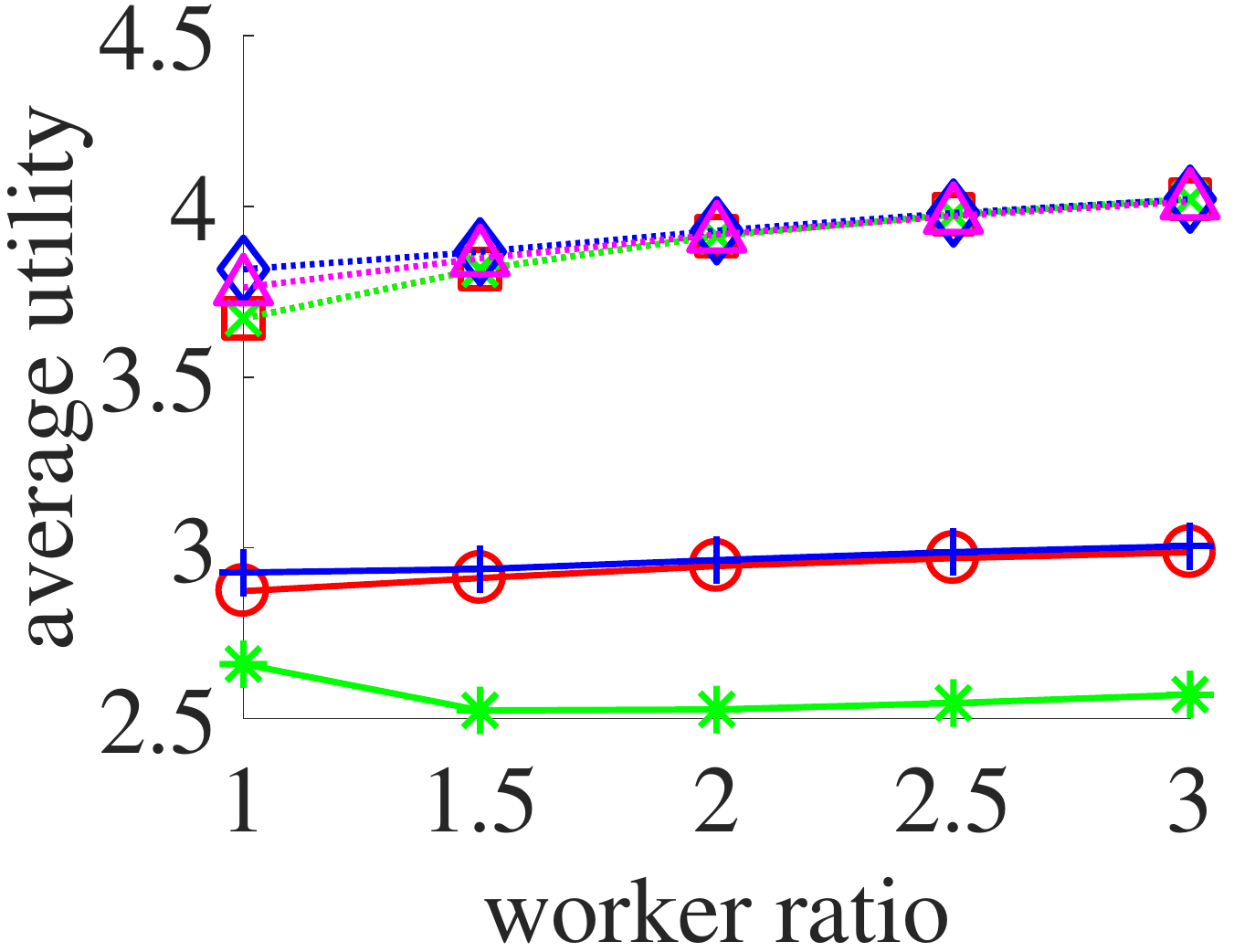}}
		\label{subfig:worker_ratio_utility_normal}}
	\subfigure[][{\scriptsize Relative Deviation of Utility}]{
		\scalebox{0.25}[0.25]{\includegraphics{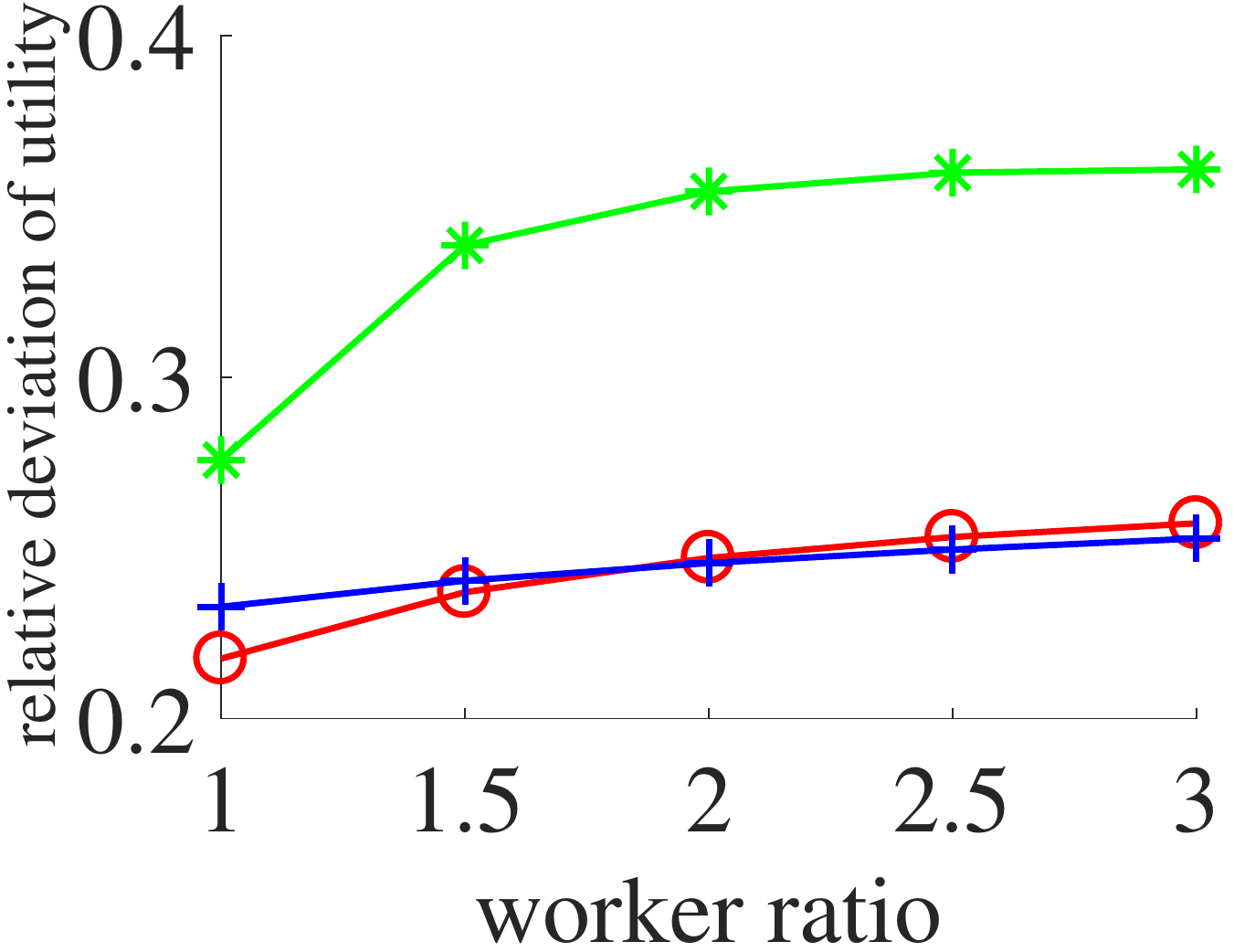}}
		\label{subfig:worker_ratio_utility_deviation_normal}}\figureCaptionMargin
	\caption{\small The impact of the worker ratio on the utility  (\emph{normal}).}\vspace{-2ex}
	\label{fig:worker_ratio_utility_normal}
\end{figure}

\begin{figure}[t!]\centering
	\subfigure{
		\scalebox{0.33}[0.33]{\includegraphics{bar2-eps-converted-to.pdf}}}\hfill\\
	\addtocounter{subfigure}{-1}\vspace{-2.5ex}
	\subfigure[][{\scriptsize Average Distance}]{
		\scalebox{0.25}[0.25]{\includegraphics{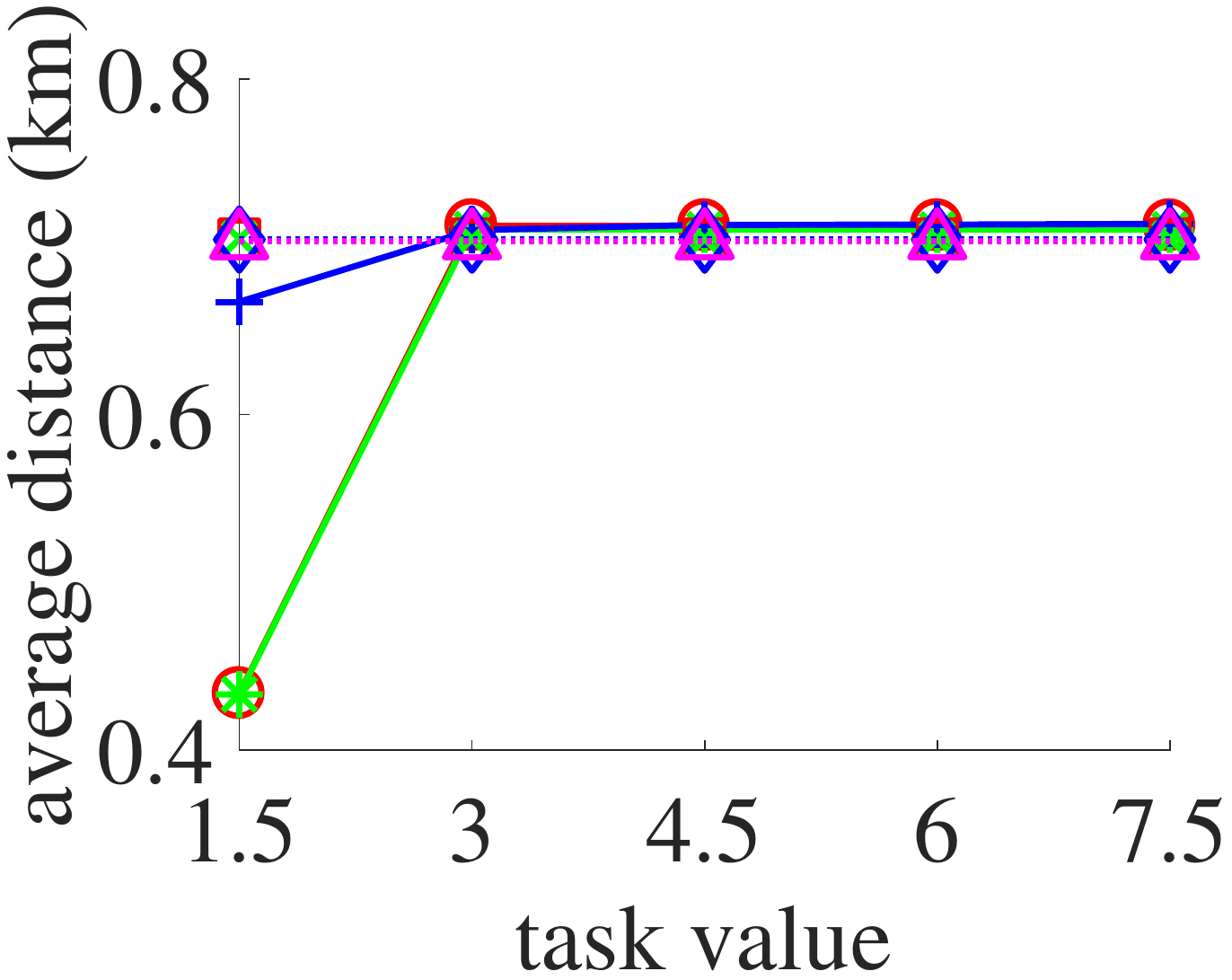}}
		\label{subfig:task_value_distance_chengdu}}
	\subfigure[][{\scriptsize Relative Deviation of Distance}]{
		\scalebox{0.25}[0.25]{\includegraphics{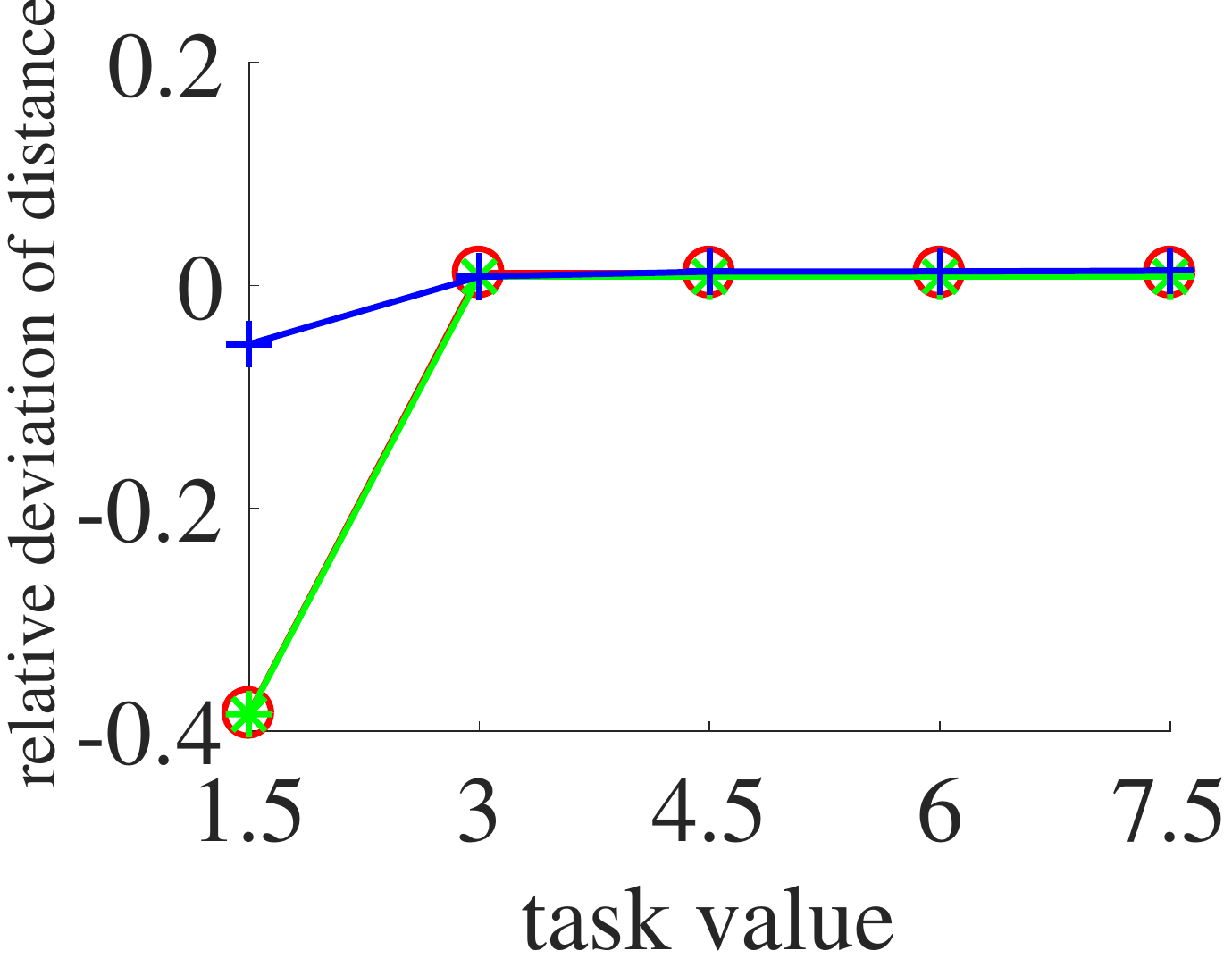}}
		\label{subfig:task_value_distance_deviation_chengdu}}\figureCaptionMargin
	\caption{\small The impact of the task value on the distance  (\emph{chengdu}).}
	\label{fig:task_value_distance_chengdu}
\end{figure}

Figure~\ref{fig:task_value_distance_chengdu} and \ref{fig:task_value_distance_normal} show the relation between the average distance and the task values. We can see that task values do not affect the average distance when the task value is larger than 3. That is because when the task value is large enough, it will not affect the difference between the two utility values. Workers will not choose many tasks in their range when the task value is minimal, leading to a small average distance.
Besides, \EXPSolutionA{} is better than \EXPSolutionB{} slightly but worse than \EXPSolutionCMP{}. However, the difference of the distances between \EXPSolutionA{} and \EXPSolutionCMP{} keeps stable as task value increases.

Figure~\ref{fig:worker_range_distance_chengdu} and \ref{fig:worker_range_distance_normal} show the relation between the average distance and the worker service area.
We can see that the average distance increases when worker range increases.
That is because a larger range will lead to more proposing workers with far distance, making the average distance larger.
The average distances of \solutionA{} and \solutionB{} are nearly equal. They are worse than \EXPSolutionCMP{} with nearly {fixed} difference distance value when the worker range is larger than 1.4.

Figure~\ref{fig:worker_ratio_distance_chengdu} and \ref{fig:worker_ratio_distance_normal} show the relation between the average distance and the worker ratio.
Especially in figure~\ref{subfig:worker_ratio_distance_chengdu}, the average distance in non-privacy solutions decreases when the worker ratio increases. That is because, with the increase in workers, the competition has become rigorous.
The number of tasks limits the increase of workers' proposals, and a task will be allocated to the worker at a small distance. Therefore, the average becomes smaller with the worker ratio becoming larger.
As for privacy solutions, competition will also cost more privacy budget on utility value, which will relieve the reduction in privacy solutions.
Similar to the comparison result before, \EXPSolutionCMP{} is better than the other two schemes when the worker ratio is larger than 1.5.

\subsubsection{PPCF and Non-PPCF}
We compare our \solutionA{} and the \solutionCMP{} with non-PPCF ones (\solutionA{}-nppcf and \solutionCMP-nppcf).
We fix the task value as 4.5, the worker range distance as 1.4, and the worker ratio as 2.
We divide the privacy budget range into 5 groups shown in Table~\ref{tab_experiment_parameter_settings}.

Figure~\ref{fig:privacy_budget_utility} shows the relation between the average utility and the privacy budget.
We mark the median of each interval as the value of the x-axis.

\begin{figure}[t!]\centering\vspace{-3ex}
	\subfigure{
		\scalebox{0.33}[0.33]{\includegraphics{bar2-eps-converted-to.pdf}}}\hfill\\
	\addtocounter{subfigure}{-1}\vspace{-2.5ex}
	\subfigure[][{\scriptsize Average Distance}]{
		\scalebox{0.25}[0.25]{\includegraphics{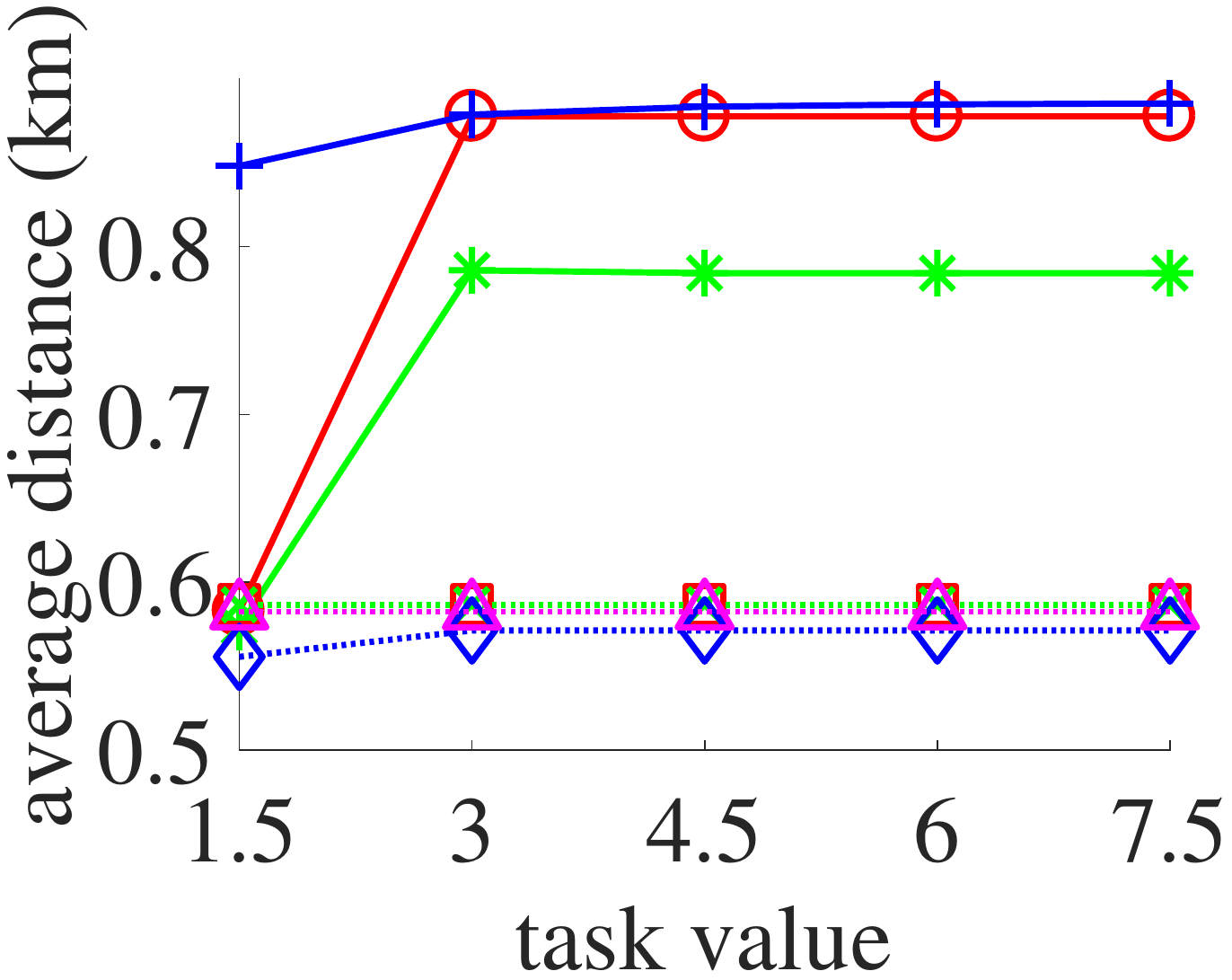}}
		\label{subfig:task_value_distance_normal}}
	\subfigure[][{\scriptsize Relative Deviation of Distance}]{
		\scalebox{0.25}[0.25]{\includegraphics{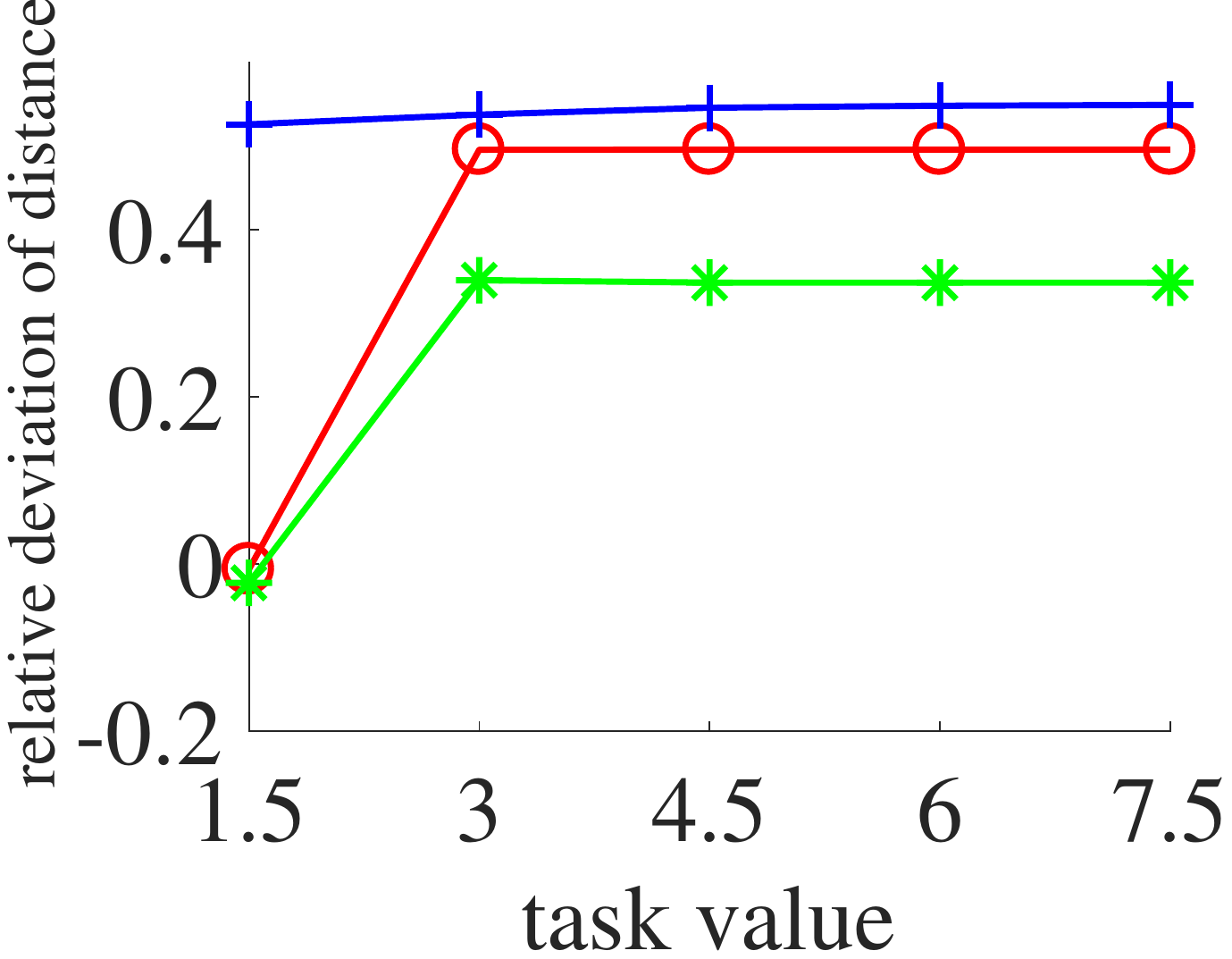}}
		\label{subfig:task_value_distance_deviation_normal}}\figureCaptionMargin
	\caption{\small The impact of the task value on the distance (\emph{normal}).}\vspace{-2ex}
	\label{fig:task_value_distance_normal}
\end{figure}

\begin{figure}[t!]\centering
	\subfigure{
		\scalebox{0.33}[0.33]{\includegraphics{bar2-eps-converted-to.pdf}}}\hfill\\
	\addtocounter{subfigure}{-1}\vspace{-2.5ex}
	\subfigure[][{\scriptsize Average Distance}]{
		\scalebox{0.25}[0.25]{\includegraphics{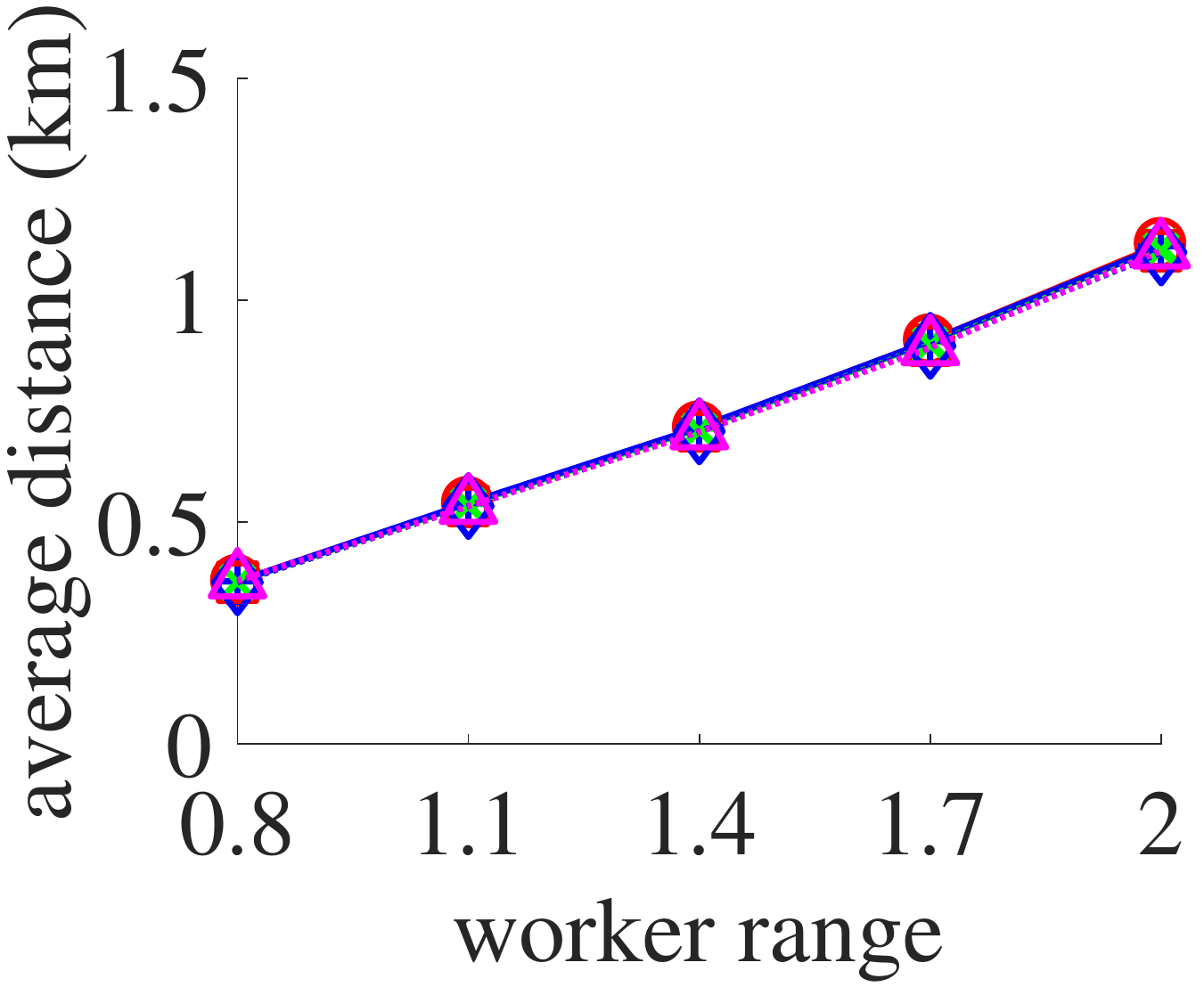}}
		\label{subfig:worker_range_distance_chengdu}}
	\subfigure[][{\scriptsize Relative Deviation of Distance}]{
		\scalebox{0.25}[0.25]{\includegraphics{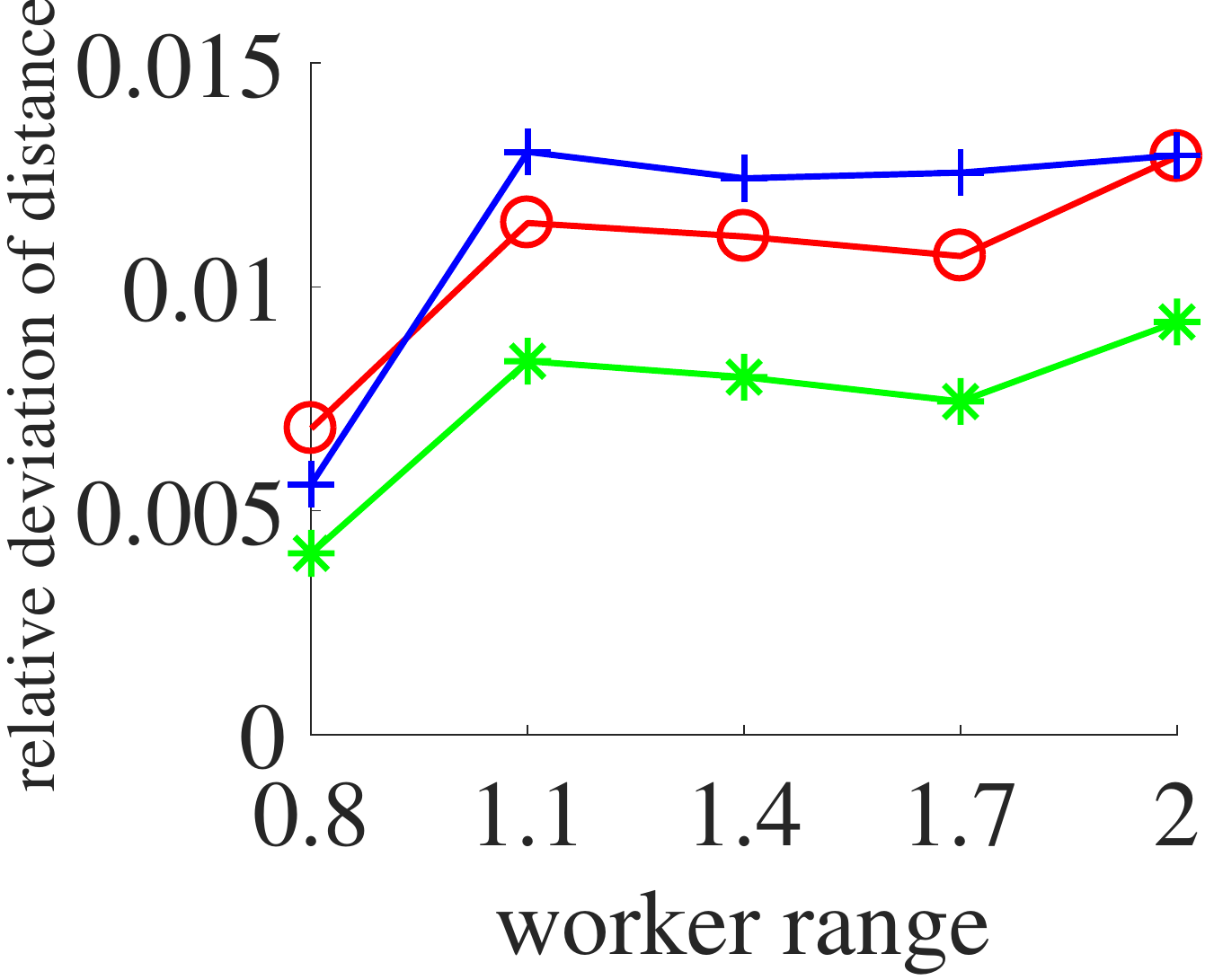}}
		\label{subfig:worker_range_distance_deviation_chengdu}}\figureCaptionMargin
	\caption{\small The impact of the worker range on the distance (\emph{chengdu}).}
	\label{fig:worker_range_distance_chengdu}
\end{figure}

\begin{figure}[t!]\centering 
	\subfigure{
		\scalebox{0.33}[0.33]{\includegraphics{bar2-eps-converted-to.pdf}}}\hfill\\
	\addtocounter{subfigure}{-1}\vspace{-2.5ex}
	\subfigure[][{\scriptsize Average Distance}]{
		\scalebox{0.25}[0.25]{\includegraphics{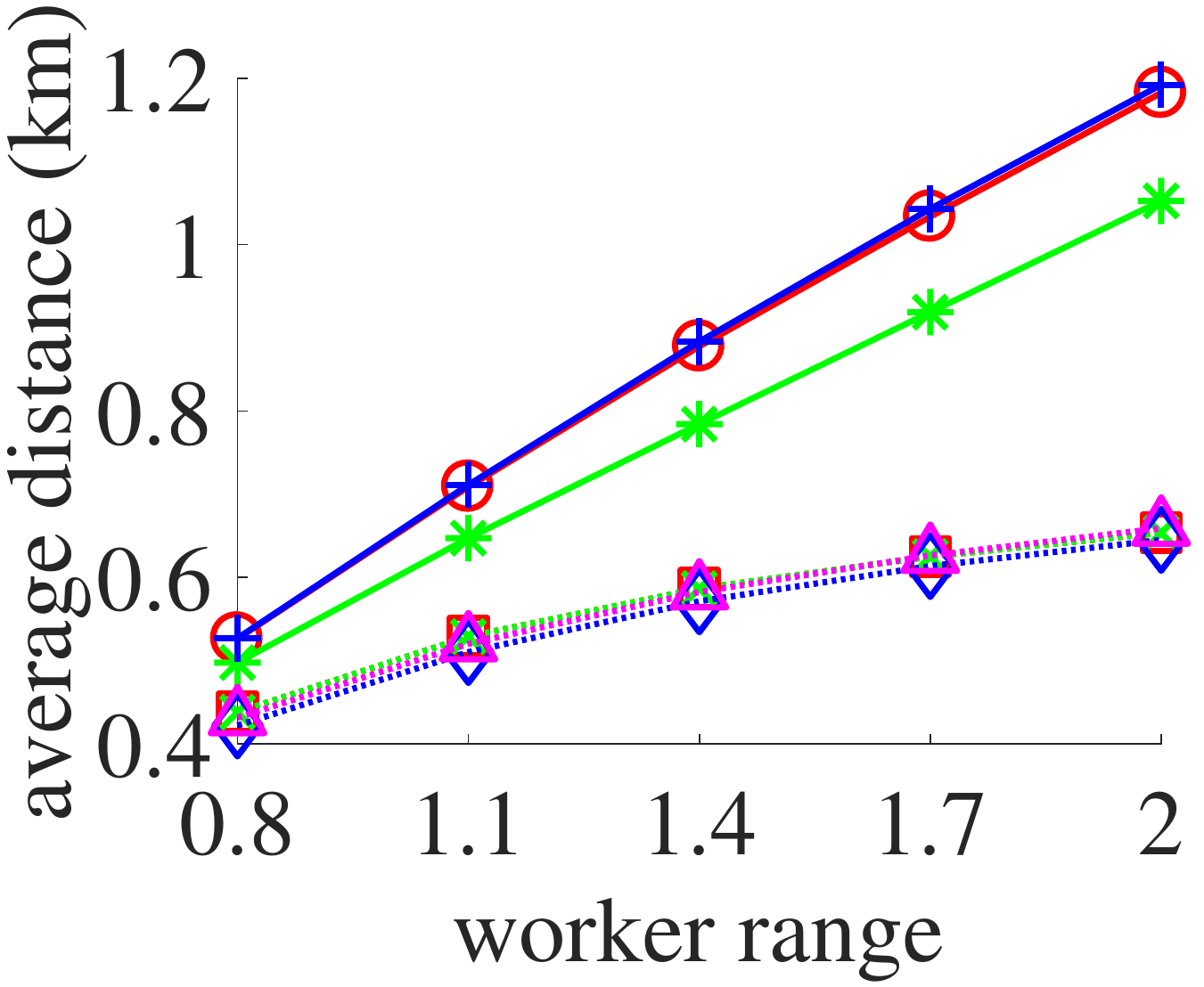}}
		\label{subfig:worker_range_distance_normal}}
	\subfigure[][{\scriptsize Relative Deviation of Distance}]{
		\scalebox{0.25}[0.25]{\includegraphics{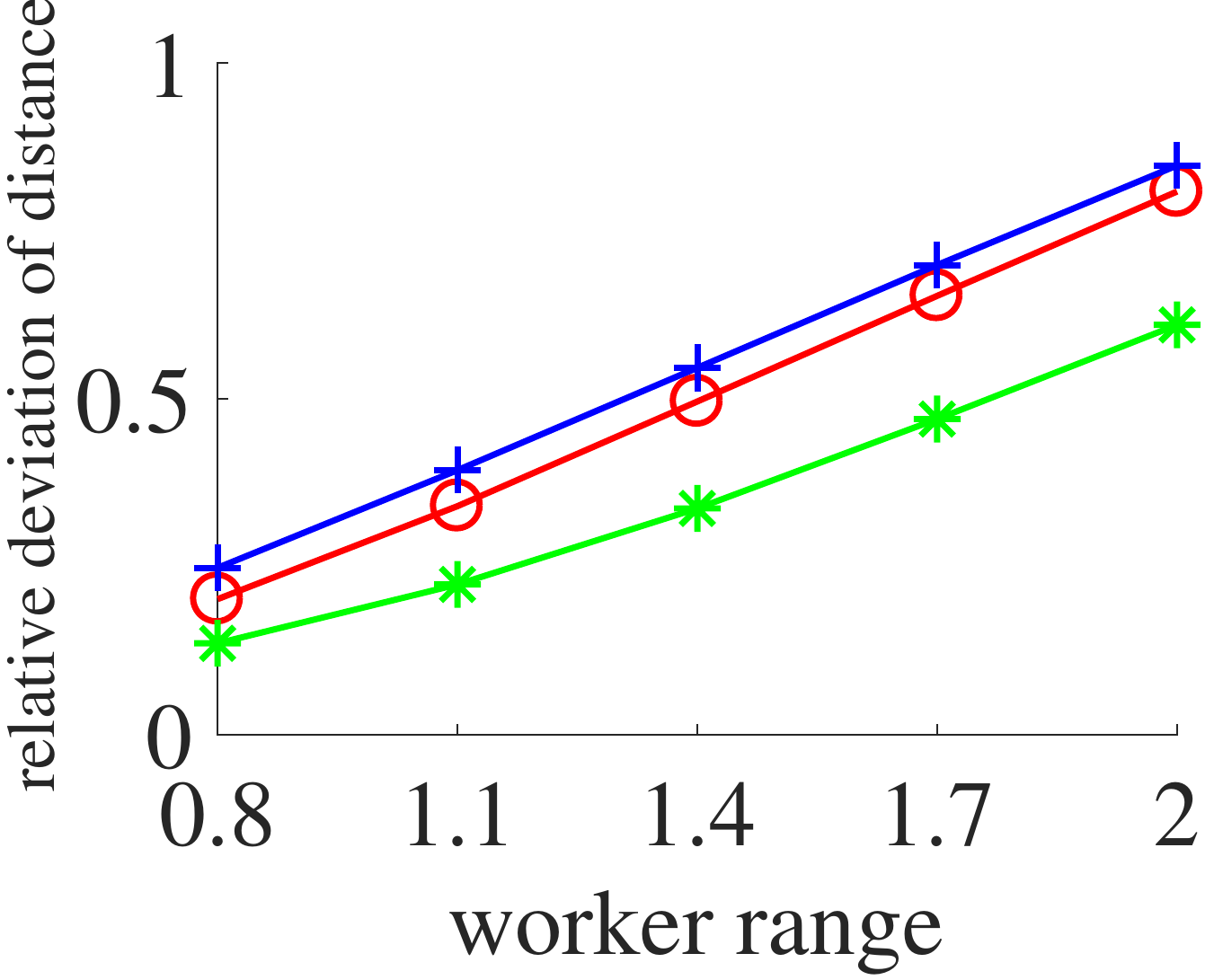}}
		\label{subfig:worker_range_distance_deviation_normal}}\figureCaptionMargin
	\caption{\small The impact of the worker range on the distance (\emph{normal}).}\vspace{-1ex}
	\label{fig:worker_range_distance_normal}
\end{figure}

The solutions with PPCF are better than that without PPCF when the privacy budget is small. It means PPCF is suitable for high-privacy situations and is continuously more effective than that without PPCF.
As the privacy budget increases, the average utility decreases.
That is because large privacy budgets give large average privacy budget cost for workers.
Although high privacy budgets are able to lead to high utility match, it also leads to high privacy budget cost.
Besides, as the privacy budget increases, the difference between PPCF and non-PPCF is eliminated.
That is because the larger the privacy budget, the more accurate the obfuscated distance, and the smaller difference between PPCF and PCF.

\section{Conclusion}\label{Conclusion}
In this paper, we formalize the Privacy-aware Task Assignment (\basicProblem{}) Problem, which assigns a task to a worker to get a high utility value. In order to make use of obfuscated distance published by workers, we propose new notations called \emph{effective obfuscated distance} and \emph{effective privacy budget}. To get a higher utility value, we offer a new comparison function called PPCF and prove that it {achieves} better effectiveness than PCF in both theory and practice. Besides, we propose another game theoretic approach to solve the problem. Extensive experiments have been conducted to show the efficiency and effectiveness of our methods on both real and synthetic data sets.

\begin{figure}[t!]\centering\vspace{-3ex}
	\subfigure{
		\scalebox{0.33}[0.33]{\includegraphics{bar2-eps-converted-to.pdf}}}\hfill\\
	\addtocounter{subfigure}{-1}\vspace{-2.5ex}
	\subfigure[][{\scriptsize Average Distance}]{
		\scalebox{0.25}[0.25]{\includegraphics{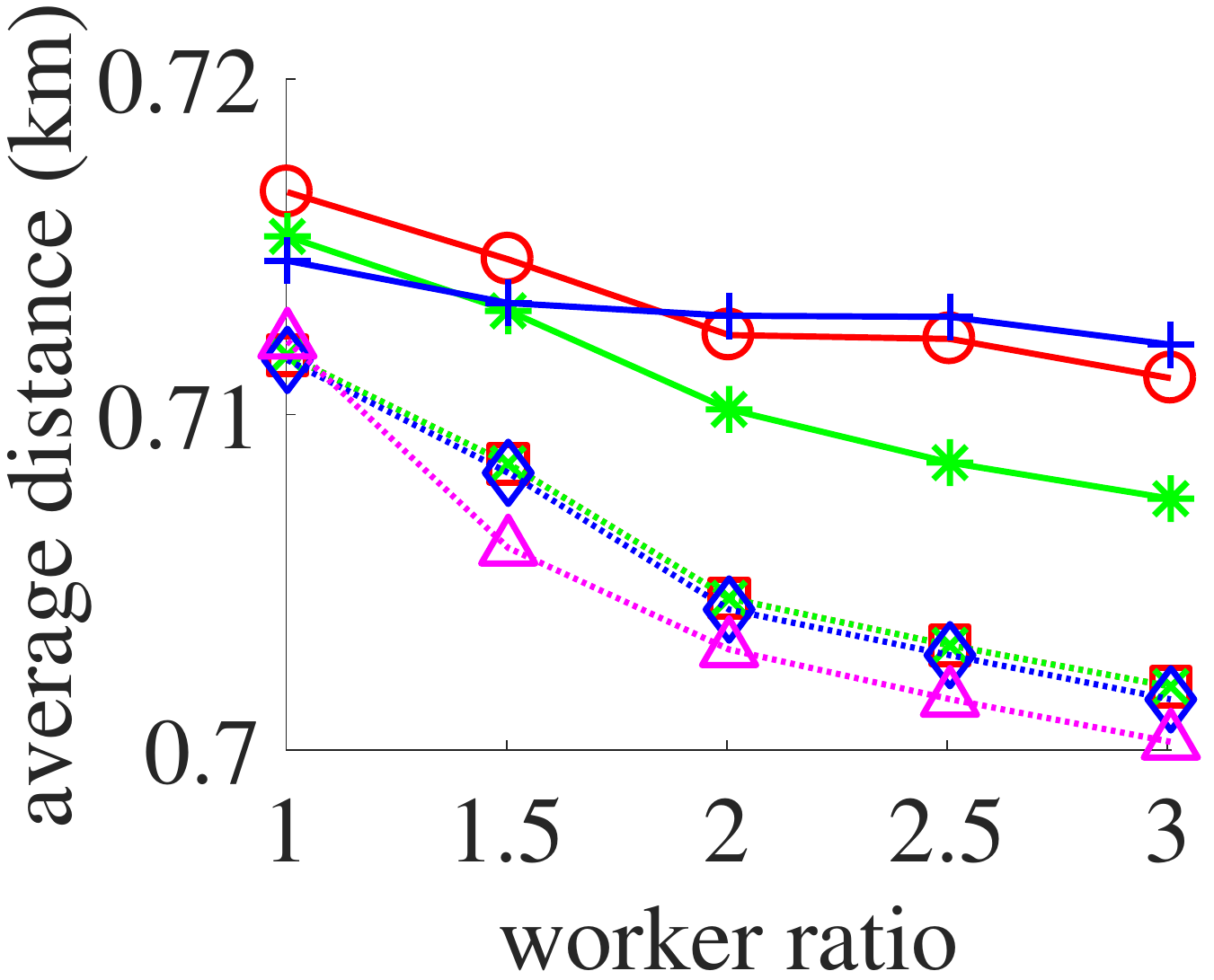}}
		\label{subfig:worker_ratio_distance_chengdu}}
	\subfigure[][{\scriptsize Relative Deviation of Distance}]{
		\scalebox{0.25}[0.25]{\includegraphics{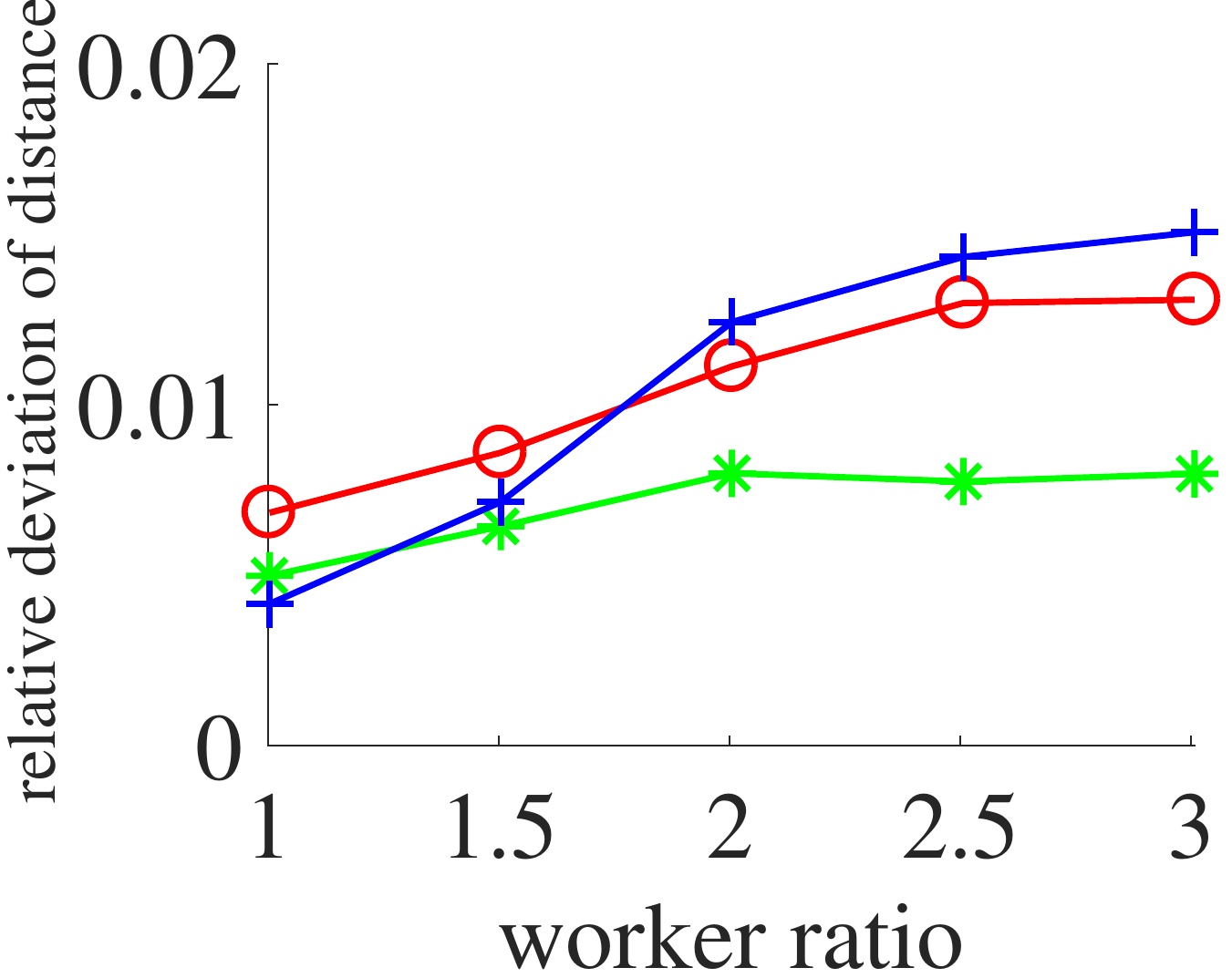}}
		\label{subfig:worker_ratio_distance_deviation_chengdu}}\figureCaptionMargin
	\caption{\small The impact of the worker ratio on the distance (\emph{chengdu}).}\vspace{-2ex}
	\label{fig:worker_ratio_distance_chengdu}
\end{figure}

\begin{figure}[t!]\centering
	\subfigure{
		\scalebox{0.33}[0.33]{\includegraphics{bar2-eps-converted-to.pdf}}}\hfill\\
	\addtocounter{subfigure}{-1}\vspace{-2.5ex}
	\subfigure[][{\scriptsize Average Distance}]{
		\scalebox{0.25}[0.25]{\includegraphics{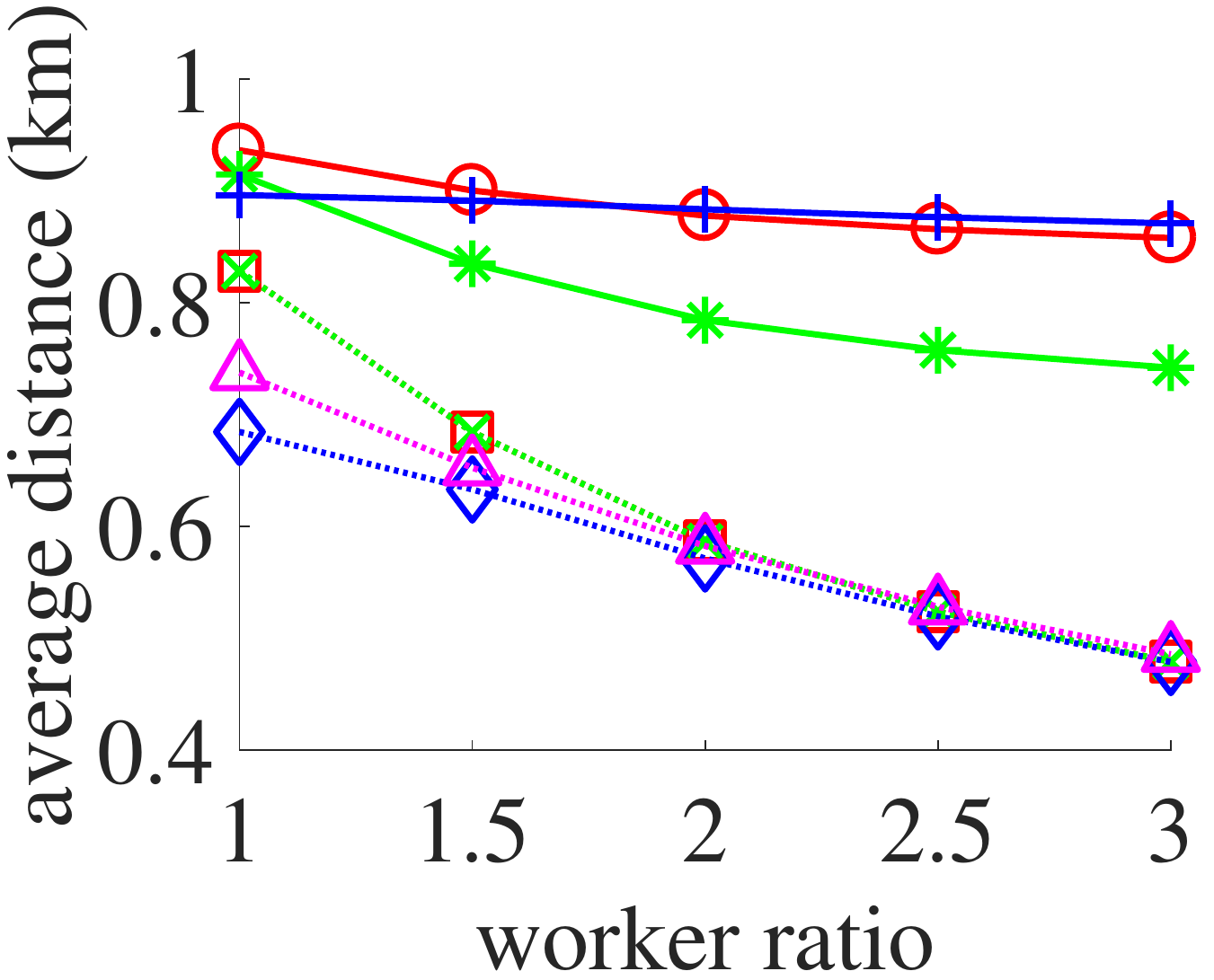}}
		\label{subfig:worker_ratio_distance_normal}}
	\subfigure[][{\scriptsize Relative Deviation of Distance}]{
		\scalebox{0.25}[0.25]{\includegraphics{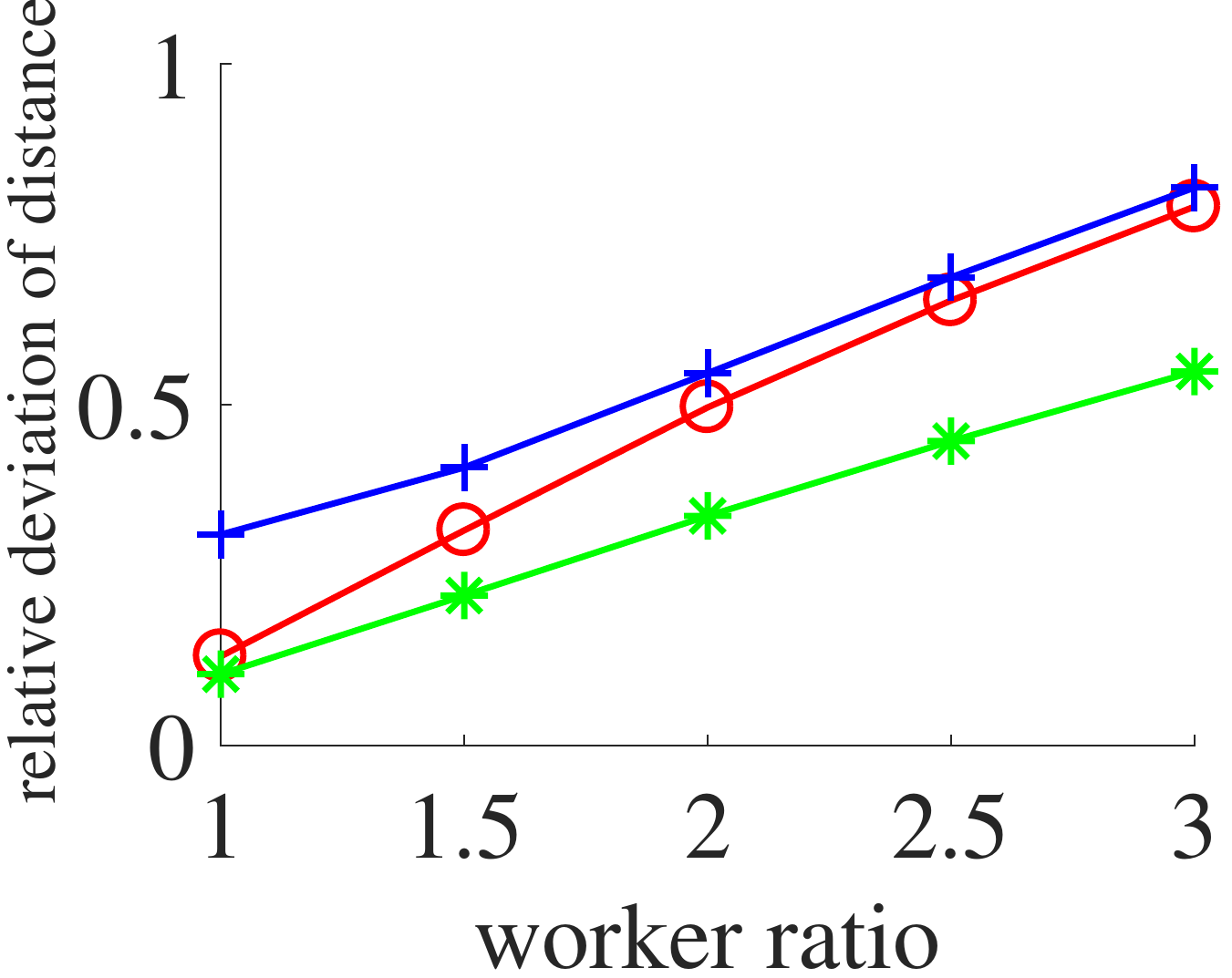}}
		\label{subfig:worker_ratio_distance_deviation_normal}}\figureCaptionMargin
	\caption{\small The impact of the worker ratio on the distance (\emph{normal}).}\vspace{-2ex}
	\label{fig:worker_ratio_distance_normal}
\end{figure}
\begin{figure}[t!]\centering
	\subfigure{
		\scalebox{0.33}[0.33]{\includegraphics{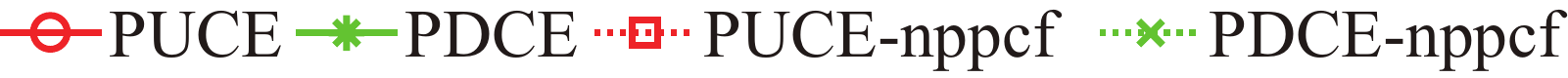}}}\hfill\\
	\addtocounter{subfigure}{-1}\vspace{-2.5ex}
	\subfigure[][{\scriptsize \emph{chengdu}}]{
		\scalebox{0.25}[0.25]{\includegraphics{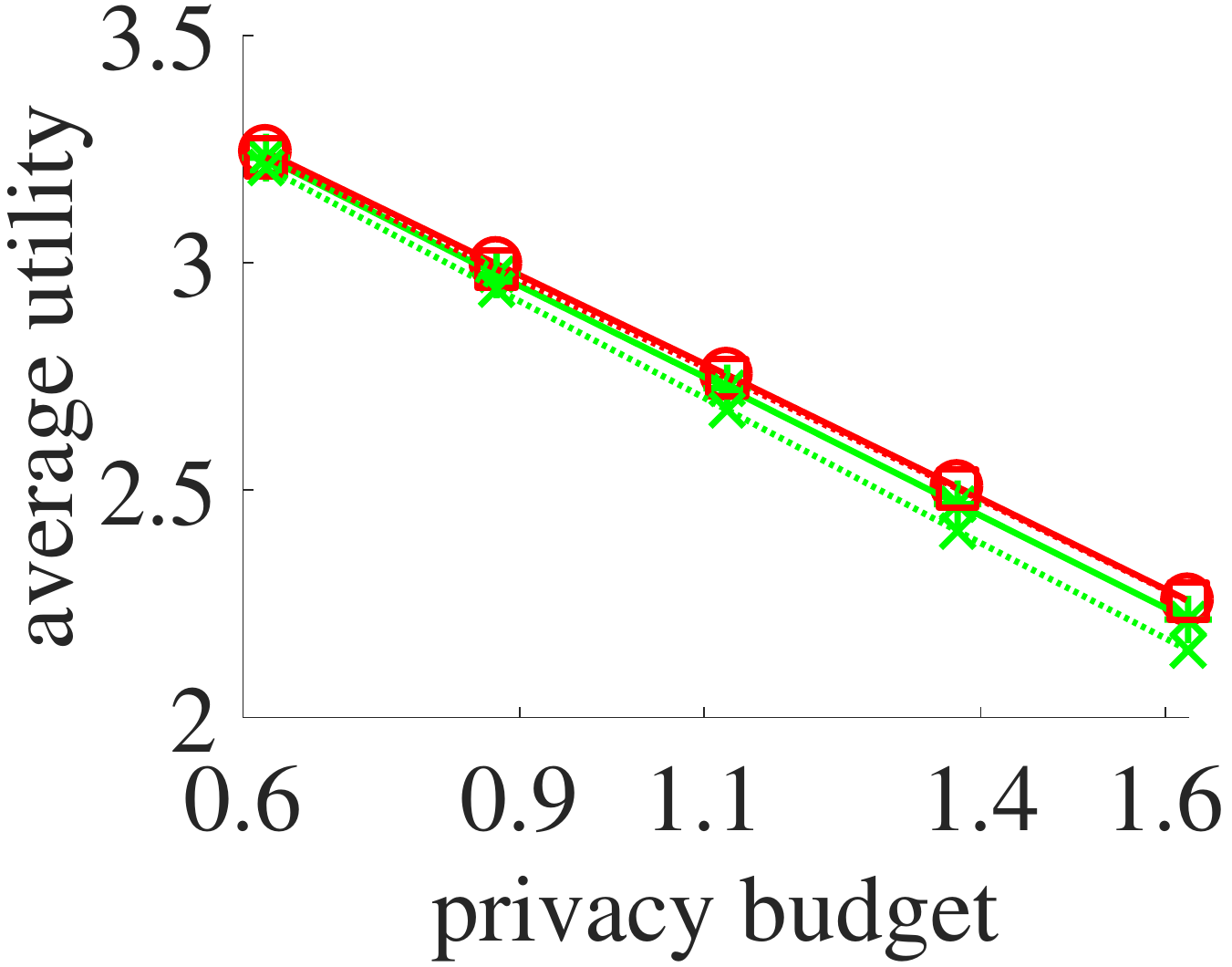}}
		\label{subfig:privacy_budget_utility_chengdu}}
	\subfigure[][{\scriptsize \emph{normal}}]{
		\scalebox{0.25}[0.25]{\includegraphics{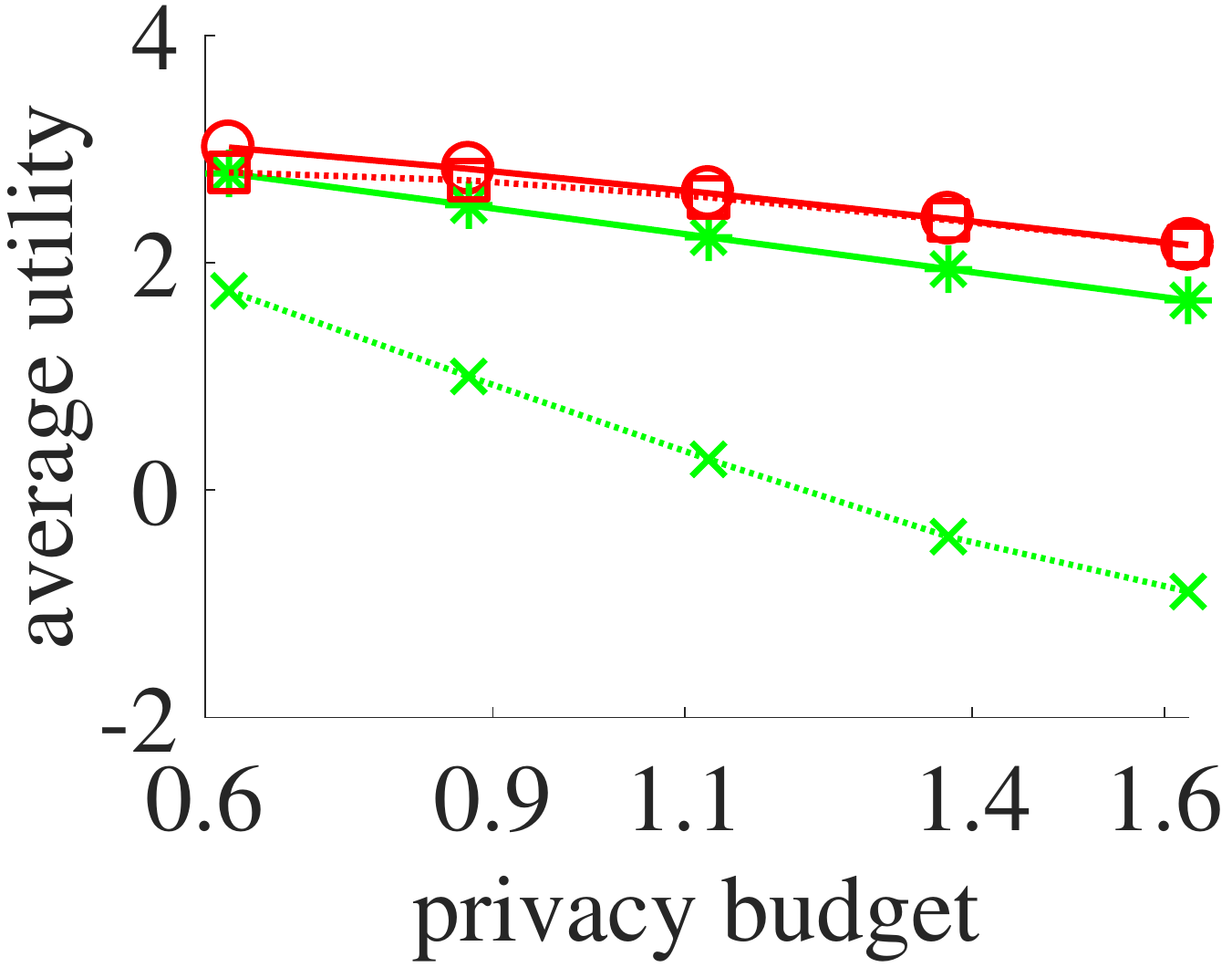}}
		\label{subfig:privacy_budget_utility_normal}}
	\caption{\small The impact of privacy on the utility.}\vspace{-1ex}
	\label{fig:privacy_budget_utility}
\end{figure}

Our \solutionA{} and \solutionB{} only consider the distance privacy of one worker in his service area.
If the service area of a worker is small enough and the quantity of tasks in this area is large enough, attackers can locate the worker's position through trilateration by viewing the entire area as a position.
That is because too much effective obfuscated distance from a worker to many tasks will outline the worker's service area.
Our subsequent work will focus on this problem and consider how to hide correlation privacy caused by the relation between different worker service areas.
Besides, in our goal function, we suppose the task value is only related to the task itself. Our subsequent work will extract the payment from the task value and research on the assumption that: the task value is related to task itself, travel distance and privacy cost.

\section{acknowledgment}
Peng Cheng's work is partially supported by the National NaturalScience Foundation of China under Grant No. 62102149 and Open Foundation of Key Laboratory of Transport Industry of Big Data Application Technologiesfor Comprehensive Transport. Libin Zheng is supported by the National Natural Science Foundation of China No. 62102463 and the Natural Science Foundation of Guangdong Province of China No.2022A1515011135. Wei Xi's work is partially supported by National Key R\&D Program of China Grant No. 2018AAA0101100. Xuemin Lin's work is partially supported by the National Key R\&D Program of China under grant 2018AAA0102502. Wenjie Zhang's work is partially supported by the ARC Future Fellowship FT210100303. Corresponding author: Peng Cheng.

\bibliographystyle{ieeetr}
\bibliography{reference}

\section{Appendix}
\subsection{Proof for Theorem~\ref{thrm:PPCF_PCF}}\label{Appendix_PPCF_PCF}
Before the proof of Theorem~\ref{thrm:PPCF_PCF}, we declare and prove Lemma~\ref{la:PCF} and Lemma~\ref{la:area} as follows.
\begin{lemma}\label{la:PCF}
For any  {\scriptsize$d_x, d_y, \epsilon_x, \epsilon_y$, $\hat{d}_x = d_x + Lap(0,1/\epsilon_x), \hat{d}_y = d_y + Lap(0,1/\epsilon_y)$, we have $PCF(\hat{d}_x,\hat{d}_y,\epsilon_x,\epsilon_y)>\frac{1}{2}\Leftrightarrow \hat{d}_x < \hat{d}_y$}.
\end{lemma}
\begin{proof}
Let $\eta_x\sim\epsilon_x, \eta_y\sim\epsilon_y$. Then we have
{\scriptsize$$PCF(\hat{d}_x,\hat{d}_y,\epsilon_x,\epsilon_y) = \iint_{D_R} f(\eta_x, \eta_y)$$}
where $f(\eta_x, \eta_y)=\frac{\epsilon_x\epsilon_y}{4}e^{-\epsilon_x|\eta_x|-\epsilon_y|\eta_y|}$ and $D_R$ is the plane set satisfying $D_R=\{(\eta_x, \eta_y): \eta_y-\eta_x<\hat{d}_y-\hat{d}_x\}$. Note that $f(\eta_x, \eta_y)$ is symmetry about both x-axis and y-axis and $D_R$ is part of plane split by line $l_{\eta}: \eta_y=\eta_x + \hat{d}_y-\hat{d}_x$. Thus, we know that only when $l_{\eta}$ crosses the origin ($\hat{d}_y=\hat{d}_x$), $PCF(\hat{d}_x,\hat{d}_y,\epsilon_x,\epsilon_y)$ equals $\frac{1}{2}$. When $\hat{d}_y-\hat{d}_x < 0$, $PCF(\hat{d}_x,\hat{d}_y,\epsilon_x,\epsilon_y)<\frac{1}{2}$, and $\hat{d}_y-\hat{d}_x > 0$, $PCF(\hat{d}_x,\hat{d}_y,\epsilon_x,\epsilon_y)>\frac{1}{2}$. Therefore, $PCF(\hat{d}_x,\hat{d}_y,\epsilon_x,\epsilon_y)>\frac{1}{2}\Leftrightarrow \hat{d}_x < \hat{d}_y$.
\end{proof}

\begin{lemma}\label{la:area}
For any two continue and differentiable non-negative functions $f,g$ defined in $\mathbb{R}$, if there exists an interval $[a,+\infty)$ satisfying that $\int_{a}^{+\infty}f(x)dx = \int_{a}^{+\infty}g(x)dx$ and there exists a point $x_0\in(a,+\infty)$ satisfying $f(x)\geq g(x)$ for $x\in(a,x_0]$ and $f(x)\leq g(x)$ for $x\in(x_0,+\infty)$, then $\int_{a}^{x}f(x)dx \geq \int_{a}^{x}g(x)dx$ for all $x\in[a,+\infty)$.
\end{lemma}
\begin{proof}
For any $x\in[a,+\infty)$, we can divide it into two cases: (1) $x\in [a,x_0]$; (2) $x\in (x_0,+\infty]$.
If (1) holds, according to \emph{$f(x)\geq g(x)$ for $x\in(a,x_0]$}, we have
{\scriptsize\begin{equation}\label{fxgeqgx_1}
    \begin{aligned}
        \int_{a}^{x}f(x)dx \geq \int_{a}^{x}g(x)dx\;\;\;\textrm{for}\;x\in [a,x_0].
    \end{aligned}
\end{equation}}
If (2) holds, then we have $\int_{x_0}^{+\infty}f(x)dx \leq \int_{x_0}^{+\infty}g(x)dx$. And we can get
{\scriptsize\begin{equation}\notag
    \begin{aligned}
        &\int_{a}^{x}f(x)dx - \int_{a}^{x}g(x)dx \\
        =& \int_{a}^{+\infty}f(x)dx - \int_{x_0}^{+\infty}f(x)dx - (\int_{a}^{+\infty}g(x)dx - \int_{x_0}^{+\infty}g(x)dx)\\
        =& \int_{x_0}^{+\infty}g(x)dx - \int_{x_0}^{+\infty}f(x)dx \geq 0.
    \end{aligned}
\end{equation}}
Therefore, we have
{\scriptsize\begin{equation}\label{fxgeqgx_2}
    \begin{aligned}
        \int_{a}^{x}f(x)dx \geq \int_{a}^{x}g(x)dx\;\;\;\textrm{for}\;x\in (x_0,+\infty).
    \end{aligned}
\end{equation}}
From Equation~\ref{fxgeqgx_1} and \ref{fxgeqgx_2}, we can have $\int_{a}^{x}f(x)dx \geq \int_{a}^{x}g(x)dx$ for $x\in [a,+\infty)$.
\end{proof}

Based on Lemma~\ref{la:PCF} and Lemma~\ref{la:area}, we give the proof of Theorem~\ref{thrm:PPCF_PCF} as follows.
\begin{proof}
From Lemma~\ref{la:PCF}, we have $PCF(\hat{d}_x,\hat{d}_y,\epsilon_x,\epsilon_y)>\frac{1}{2}\Leftrightarrow \hat{d}_x < \hat{d}_y$.
From Equation~\ref{PPCF_equation}, we have $\textrm{Pr}[d_x<d_y]>\frac{1}{2}\Leftrightarrow d_x<\hat{d}_y$. Therefore, we only need to prove $Pr[\hat{d}_x < \hat{d}_y] \leq Pr[d_x<\hat{d}_y]$ for any $d_x, d_y$ satisfying $d_x<d_y$.

According to the definition, we have
{\scriptsize\begin{equation}\notag
    \begin{aligned}
        \textrm{Pr}[\hat{d}_x < \hat{d}_y] &= \textrm{Pr}[d_x + \eta_x < d_y + \eta_y] = \textrm{Pr}[\eta_y > \eta_x + d_x - d_y] \\
                                           &= \int_{-\infty}^{+\infty}\left(\int_{-\infty}^{\eta_y-d_x+d_y}\frac{\epsilon_x\epsilon_y}{4}e^{-(\epsilon_x|\eta_x|+\epsilon_y|\eta_y|)}d\eta_x\right)d\eta_y
    \end{aligned}
\end{equation}}
and
{\scriptsize\begin{equation}\notag
    \begin{aligned}
        \textrm{Pr}[d_x < \hat{d}_y] &= \textrm{Pr}[d_x < d_y + \eta_y] = \textrm{Pr}[\eta_y > d_x - d_y] \\
       &= \int_{d_x-d_y}^{+\infty}\frac{\epsilon_y}{2}e^{-\epsilon_y|\eta_y|}d\eta_y.
    \end{aligned}
\end{equation}}
Let $s = d_y - d_x$.
Let $F: s\to\textrm{Pr}[\hat{d}_x < \hat{d}_y]$ and $G: s\to\textrm{Pr}[d_x < \hat{d}_y]$.
From the definition, we know $s>0$, $\lim\limits_{s\to 0}F(s) = \lim\limits_{s\to 0}G(s) = \frac{1}{2}$ and $\lim\limits_{s\to +\infty}F(s) = \lim\limits_{s\to +\infty}G(s) = 1$. And we have
{\scriptsize\begin{align}
	\frac{\partial F(s)}{\partial s} &= \frac{\epsilon_x\epsilon_y}{4}(\frac{e^{-s\epsilon_x}+e^{-s\epsilon_y}}{\epsilon_x+\epsilon_y} - \frac{e^{-s\epsilon_x}-e^{-s\epsilon_y}}{\epsilon_x-\epsilon_y})\notag\\
	&= \frac{\epsilon_x\epsilon_y}{2}\cdot\frac{e^{-s\epsilon_y}\epsilon_x - e^{-s\epsilon_x}\epsilon_y}{(\epsilon_x+\epsilon_y)(\epsilon_x-\epsilon_y)}>0, \notag\\
	\frac{\partial G(s)}{\partial s} &= \frac{\epsilon_y}{2}e^{-s\epsilon_y}
	> 0, \notag\\
	\frac{\partial F(s)}{\partial s}/\frac{\partial G(s)}{\partial s} &= \frac{\epsilon_x(\epsilon_x-e^{s(\epsilon_y-\epsilon_x)}\epsilon_y)}{(\epsilon_x+\epsilon_y)(\epsilon_x-\epsilon_y)}.\notag
\end{align}}

Let $\frac{\partial F(s)}{\partial s}/\frac{\partial G(s)}{\partial s}\leq 1$. Then we have $s\leq\frac{1}{\epsilon_x-\epsilon_y}\textrm{ln}\frac{\epsilon_x}{\epsilon_y}$. Let $\frac{\partial F(s)}{\partial s}/\frac{\partial G(s)}{\partial s}\geq 1$. Then we have $s\geq\frac{1}{\epsilon_x-\epsilon_y}\textrm{ln}\frac{\epsilon_x}{\epsilon_y}$.
That is to say $\frac{\partial G(s)}{\partial s} \geq \frac{\partial F(s)}{\partial s}$ for $s\in(0,\frac{1}{\epsilon_x-\epsilon_y}\textrm{ln}\frac{\epsilon_x}{\epsilon_y})$ and $\frac{\partial G(s)}{\partial s} \leq \frac{\partial F(s)}{\partial s}$ for $s\in(\frac{1}{\epsilon_x-\epsilon_y}\textrm{ln}\frac{\epsilon_x}{\epsilon_y},+\infty)$. According to Lemma~\ref{la:area}, we have $F(s)\leq G(s)$ for $s\in(0,+\infty)$.
\end{proof}
\subsection{Proof for Theorem~\ref{thrm:Pos_PoA}}\label{proof_converge}
\begin{proof}
Let $\hat{U}(\vectorfont{st})$ be the overall utility of the strategy $\vectorfont{st}$ with  (i.e., $\hat{U}(\vectorfont{st})=\Phi(\vectorfont{st})$).
Besides, we note the global optimal strategy as $\hat{\vectorfont{st}}$, the strategy of achieving best competing utility value as $\vectorfont{st}^{\ast}$ and the worst competing utility value as $\vectorfont{st}^{\sharp}$.
Then we have $\hat{U}(\hat{\vectorfont{st}})=\Phi(\hat{\vectorfont{st}})$, $\hat{U}(\vectorfont{st}^{\ast})=\Phi(\vectorfont{st}^{\ast})$ and $\hat{U}(\vectorfont{st}^{\sharp})=\Phi(\vectorfont{st}^{\sharp})$. Thus,
{\scriptsize\begin{equation}\notag
    \begin{aligned}
        EPoS = \frac{E(\hat{U}(\vectorfont{st}^{\ast}))}{E(OPT)}
             = \frac{E(\hat{U}(\vectorfont{st}^{\ast}))}{E(\hat{U}(\hat{\vectorfont{st}}))}
             \leq 1.
    \end{aligned}
\end{equation}}

If we get the lower bound of $E(\hat{U}(\vectorfont{st}^{\sharp}))$ and upper bound of $E(\hat{U}(\hat{\vectorfont{st}}))$, then we can get the value of EPoA.
As for $E(\hat{U}(\vectorfont{st}^{\sharp}))$, we have
{\scriptsize\begin{equation}\notag
    \begin{aligned}
        E(\hat{U}(\vectorfont{st}^{\sharp})) &\geq \min\limits_{k} \sum_{t_i\in\entity{T}}\sum_{w_j\in\entity{W}}(s_{i,j}^{(k)}\cdot (v_i-f_d(d_{i,j}))-f_p(\vectorfont{b_{i,j}^{(k)}}\cdot\vectorfont{\epsilon_{i,j}})) \\
                                             &\geq \sum_{t_i\in\entity{T}} \min\limits_{R_j\ni t_i, U_{j}^{L}(i)>0} U_{j}^{L}(i) = \sum_{t_i\in\entity{T}} U_{min}^{+}(i)
    \end{aligned}
\end{equation}}
As for $E(\hat{U}(\hat{\vectorfont{st}}))$, we have
{\scriptsize\begin{equation}\notag
    \begin{aligned}
        E(\hat{U}(\hat{\vectorfont{st}})) &\leq OPT(\sum_{t_i\in\entity{T}}\sum_{w_j\in\entity{W}}(s_{i,j}\cdot (v_i-f_d(\tilde{d}_{i,j}))-f_p(\vectorfont{b_{i,j}}\cdot\vectorfont{\epsilon_{i,j}}))) \\
                                          &\leq \sum_{t_i\in\entity{T}} \max\limits_{R_j\ni t_i} U_{j}^{H}(i) = \sum_{t_i\in\entity{T}} U_{max}^{+}(i)
    \end{aligned}
\end{equation}}
Therefore, we have
{\scriptsize\begin{equation}\notag
    \begin{aligned}
        EPoA = \frac{E(\hat{U}(\vectorfont{st}^{\ast}))}{E(OPT)}
             \geq \frac{\sum_{t_i\in\entity{T}} U_{min}^{+}(i)}{\sum_{t_i\in\entity{T}} U_{max}^{+}(i)}
    \end{aligned}
\end{equation}}
\end{proof}

\subsection{Proof for Theorem~\ref{thrm:Privacy_PGT}}\label{proof_privacy_PGT}
\begin{proof}
Let $\algvar{A}_j$ be the mechanism \solutionB{} applying to $w_j$ with query $f$ defined above.
Let $X_j$ be the location of $w_j$.
For query $f(X_j)=[d_{i_1,j},...,d_{i_{|R_j|},j}]$, we extend it to an equivalent query $\hat{f}(X_j)=f(X_j)\cdot \algvar{J}$, where
{\scriptsize\begin{equation}\notag
\algvar{J}=
    \begin{bmatrix}
        CP(\vectorfont{b_{i_1,j}})  &       &       &   \\
        &       CP(\vectorfont{b_{i_2,j}})  &       &   \\
        &       &       \ddots  &   \\
        &       &       &       CP(\vectorfont{b_{i_{|R_j|},j}})
    \end{bmatrix}
\end{equation}}
is a block diagonal matrix. Actually, $\hat{f}(X_j)$ means query $d_{i_u,j}$ for $sum(\vectorfont{b_{i_u,j}})$ times for $u\in[|R_j|]$. We denote the size of $\hat{f}(X_j)$ as $|\hat{f}|$ and the $a$-th element of $\hat{f}(X_j)$ as $\hat{f}(X_j)_a$.

Let $Y_j$ denote the set of all published obfuscated distances of the worker $w_j$ to tasks in $R_j$.
Then we have $Y_j=\hat{f}(X_j)+[\eta_1, \eta_2,...,\eta_{|\hat{f}|}]$, where $\eta_a(1\leq a\leq |\hat{f}|)$ is an i.i.d random variable drawn from $Lap(1/\epsilon_a)$.
Hence we have
{\scriptsize\begin{equation}\notag
    \begin{aligned}
        \frac{\textrm{Pr}[\algvar{A}_j(X_j)=Y_j]}{\textrm{Pr}[\algvar{A}_j(X'_j)=Y_j]}
        &= \prod_{a\in[|\hat{f}|]} (\frac{\textrm{exp}(-\epsilon_a|Y_{j,a}-\hat{f}(X_{j})_a|)}{\textrm{exp}(-\epsilon_a|Y_{j,a}-\hat{f}(X_{j}')_a|)})
           \\
        &=\prod_{t_i\in R_j}\prod_{u\in[sum(\vectorfont{b_{i,j}})]}
           (\frac{\textrm{exp}(-\epsilon_{i,j}^{(u)}|\tilde{d}_{i,j}^{(u)}-d_{i,j}|)}{\textrm{exp}(-\epsilon_{i,j}^{(u)}|\tilde{d}_{i,j}^{(u)}-d'_{i,j}|)}) \\
        &\leq\prod_{t_i\in R_j}\prod_{u\in[sum(\vectorfont{b_{i,j}})]}
           (\textrm{exp}(\epsilon_{i,j}^{(u)}(|d_{i,j}-d'_{i,j}|))) \\
        &=\prod_{t_i\in R_j}
           \textrm{exp}(\vectorfont{b_{i,j}\epsilon_{i,j}}(|d_{i,j}-d'_{i,j}|)) \\
        &\leq\textrm{exp}(\sum_{t_i\in R_j}\vectorfont{b_{i,j}\epsilon_{i,j}}r_j).
    \end{aligned}
\end{equation}}
Because $X_j$ contains only one element, then we have \solutionB{} satisfies $(\sum_{t_i\in R_j}\vectorfont{b_{i,j}\epsilon_{i,j}}r_j)$-local differential privacy for each worker $w_j$.
\end{proof}

\subsection{Experiment Result for the Uniform Data set}
The experiment result of the uniform data set is shown in this section.

The time cost is shown in Figure~\ref{fig:worker_ratio_time_cost_just_for_uniform}.
\begin{figure}[ht!]\centering 
\subfigure{
  \scalebox{0.33}[0.33]{\includegraphics{bar2-eps-converted-to.pdf}}}\hfill\\
\addtocounter{subfigure}{-1}\vspace{-2ex}
	\subfigure[][{\scriptsize \emph{uniform}}]{
		\scalebox{0.25}[0.25]{\includegraphics{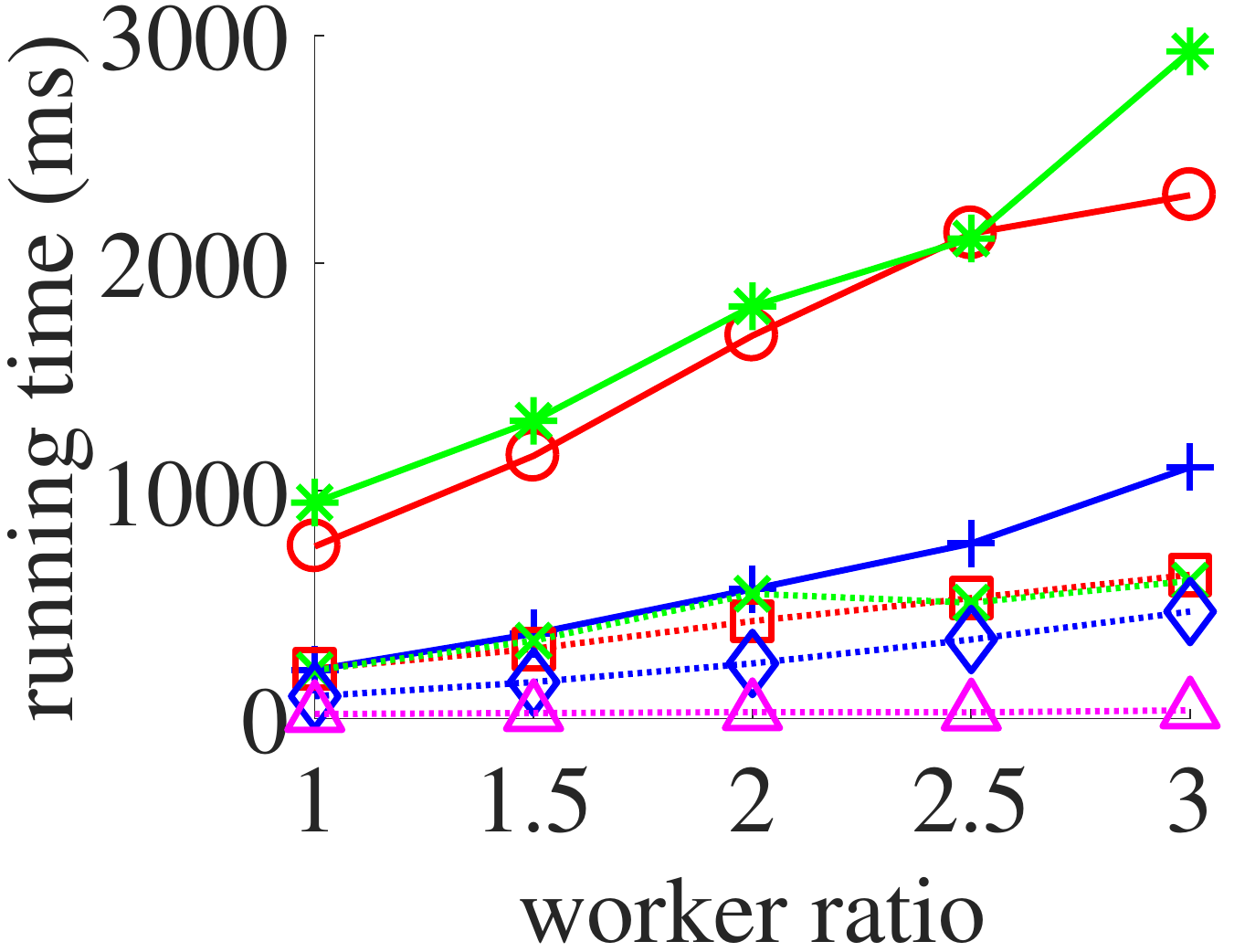}}
		\label{subfig:task_value_time_cost_uniform}}\figureCaptionMargin
	\caption{\small The impact of the worker ratio on the time cost.}
	\label{fig:worker_ratio_time_cost_just_for_uniform}
\end{figure}

The impact of task value on utility is shown in Figure~\ref{fig:task_value_utility_uniform}.
\begin{figure}[ht!]\centering
\subfigure{
  \scalebox{0.33}[0.33]{\includegraphics{bar2-eps-converted-to.pdf}}}\hfill\\
\addtocounter{subfigure}{-1}\vspace{-2ex}
	\subfigure[][{\scriptsize Average Utility}]{
		\scalebox{0.25}[0.25]{\includegraphics{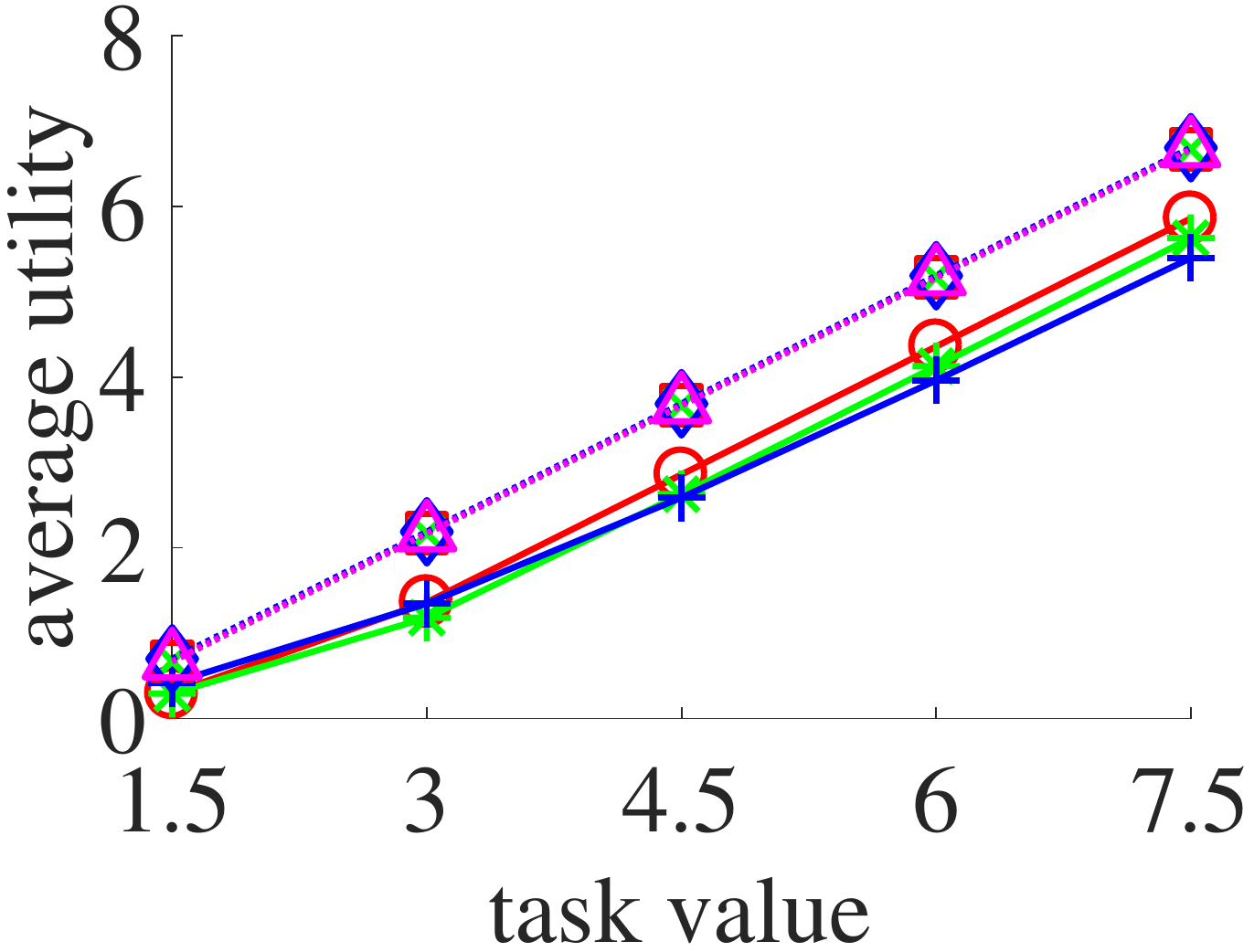}}
		\label{subfig:task_value_utility_uniform}}
	\subfigure[][{\scriptsize Relative Deviation of Utility}]{
		\scalebox{0.25}[0.25]{\includegraphics{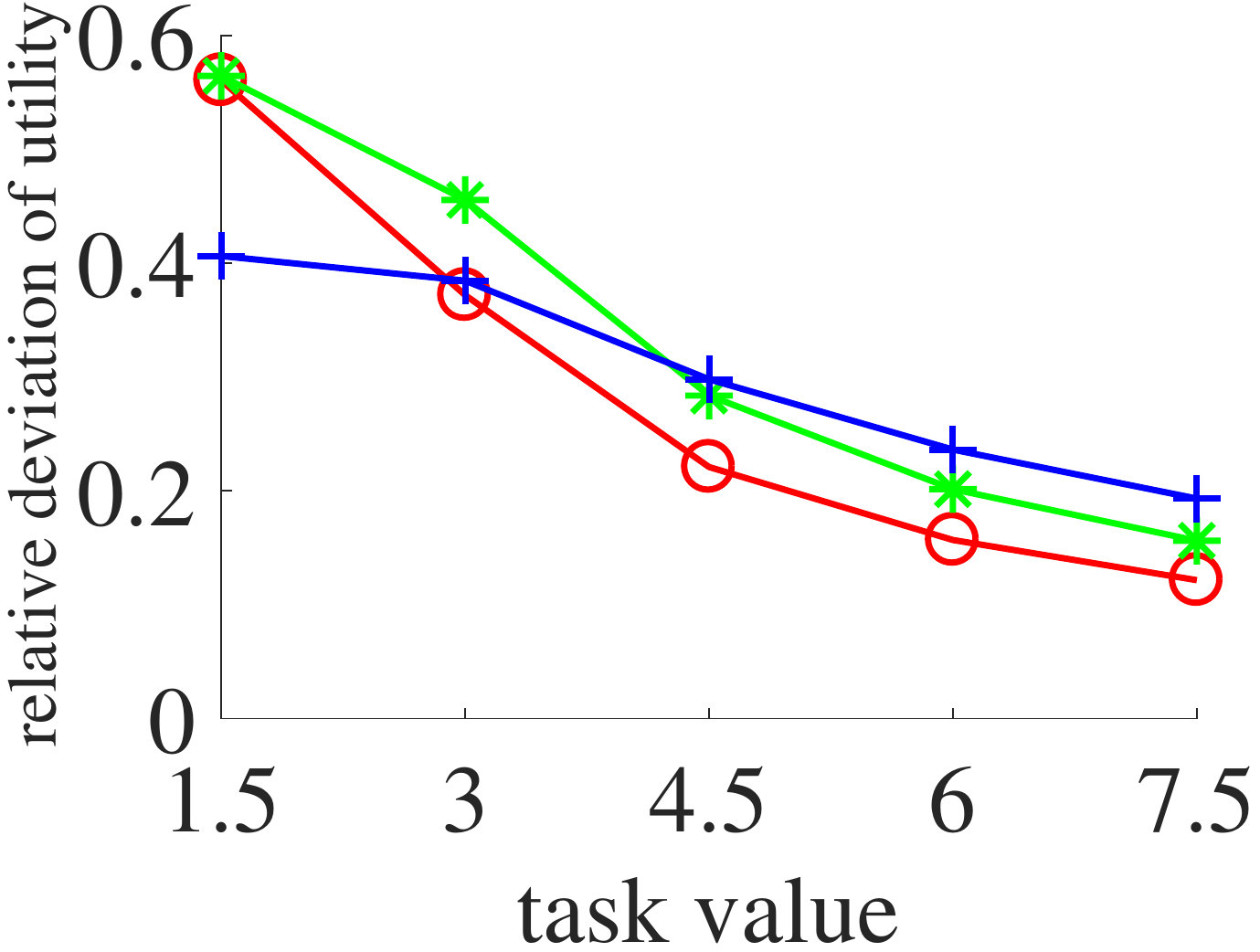}}
		\label{subfig:task_value_utility_deviation_uniform}}\figureCaptionMargin
	\caption{\small The impact of the task value on the utility for \emph{uniform}.}
	\label{fig:task_value_utility_uniform}
\end{figure}

The impact of worker range on utility is shown in Figure~\ref{fig:worker_range_utility_uniform}.

\begin{figure}[ht!]\centering
\subfigure{
  \scalebox{0.33}[0.33]{\includegraphics{bar2-eps-converted-to.pdf}}}\hfill\\
\addtocounter{subfigure}{-1}\vspace{-2ex}
	\subfigure[][{\scriptsize Average Utility}]{
		\scalebox{0.25}[0.25]{\includegraphics{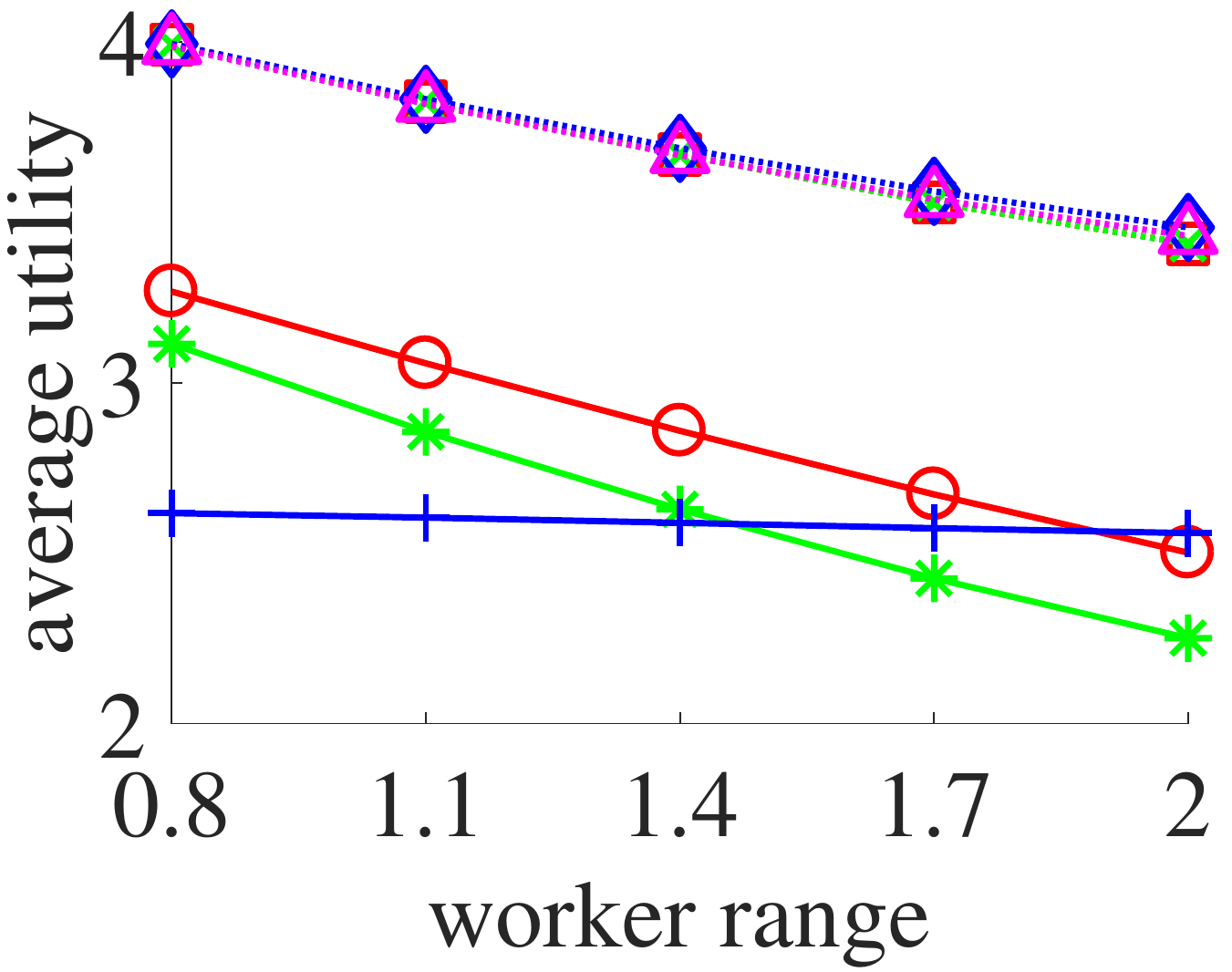}}
		\label{subfig:worker_range_utility_uniform}}
	\subfigure[][{\scriptsize Relative Deviation of Utility}]{
		\scalebox{0.25}[0.25]{\includegraphics{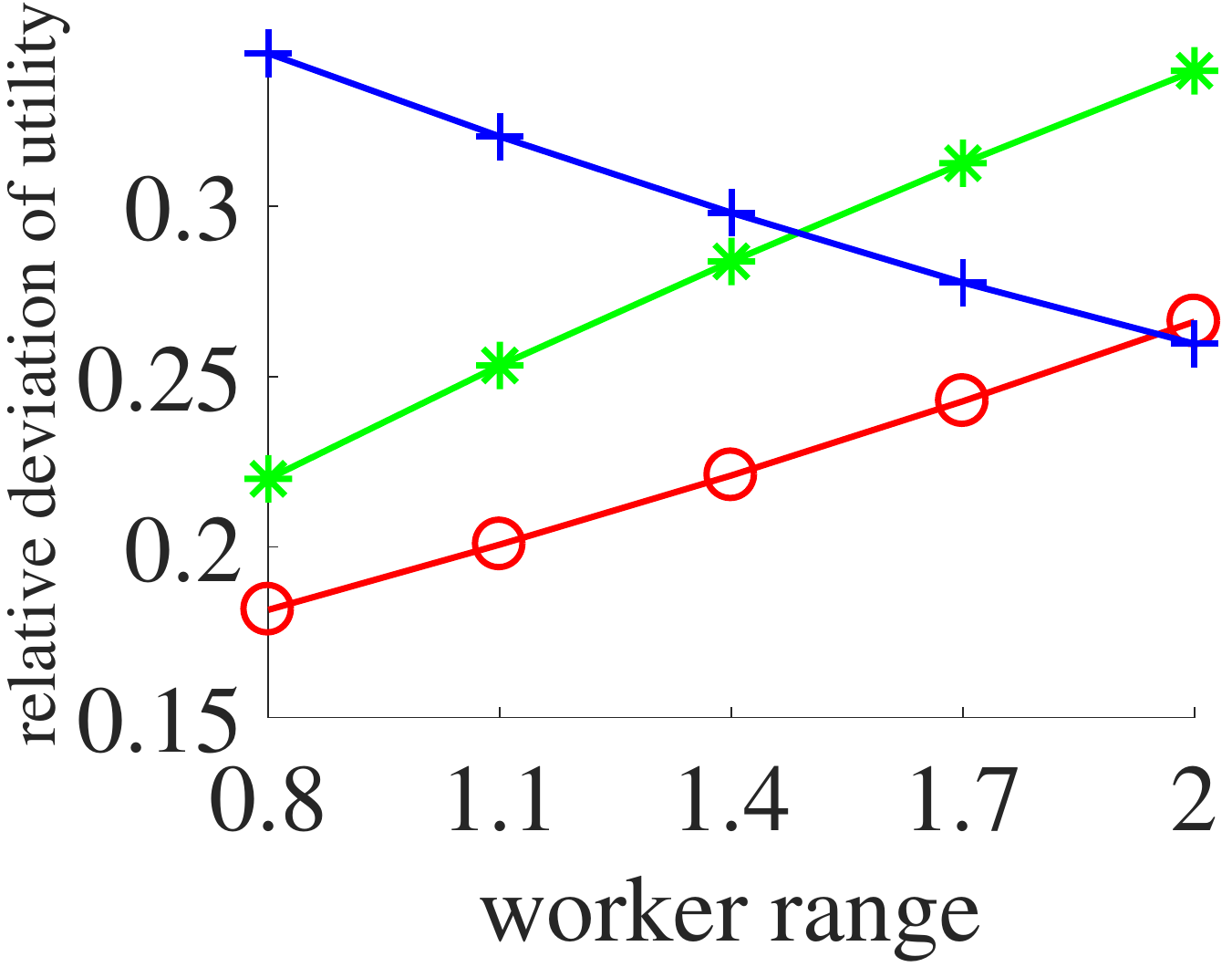}}
		\label{subfig:worker_range_utility_deviation_uniform}}
	\caption{\small The impact of the worker range on the utility for \emph{uniform}.}
	\label{fig:worker_range_utility_uniform}
\end{figure}

The impact of worker ratio on utility is shown in Figure~\ref{fig:worker_ratio_utility_uniform}.

\begin{figure}[ht!]\centering 
\subfigure{
  \scalebox{0.33}[0.33]{\includegraphics{bar2-eps-converted-to.pdf}}}\hfill\\
\addtocounter{subfigure}{-1}\vspace{-2ex}
	\subfigure[][{\scriptsize Average Utility}]{
		\scalebox{0.25}[0.25]{\includegraphics{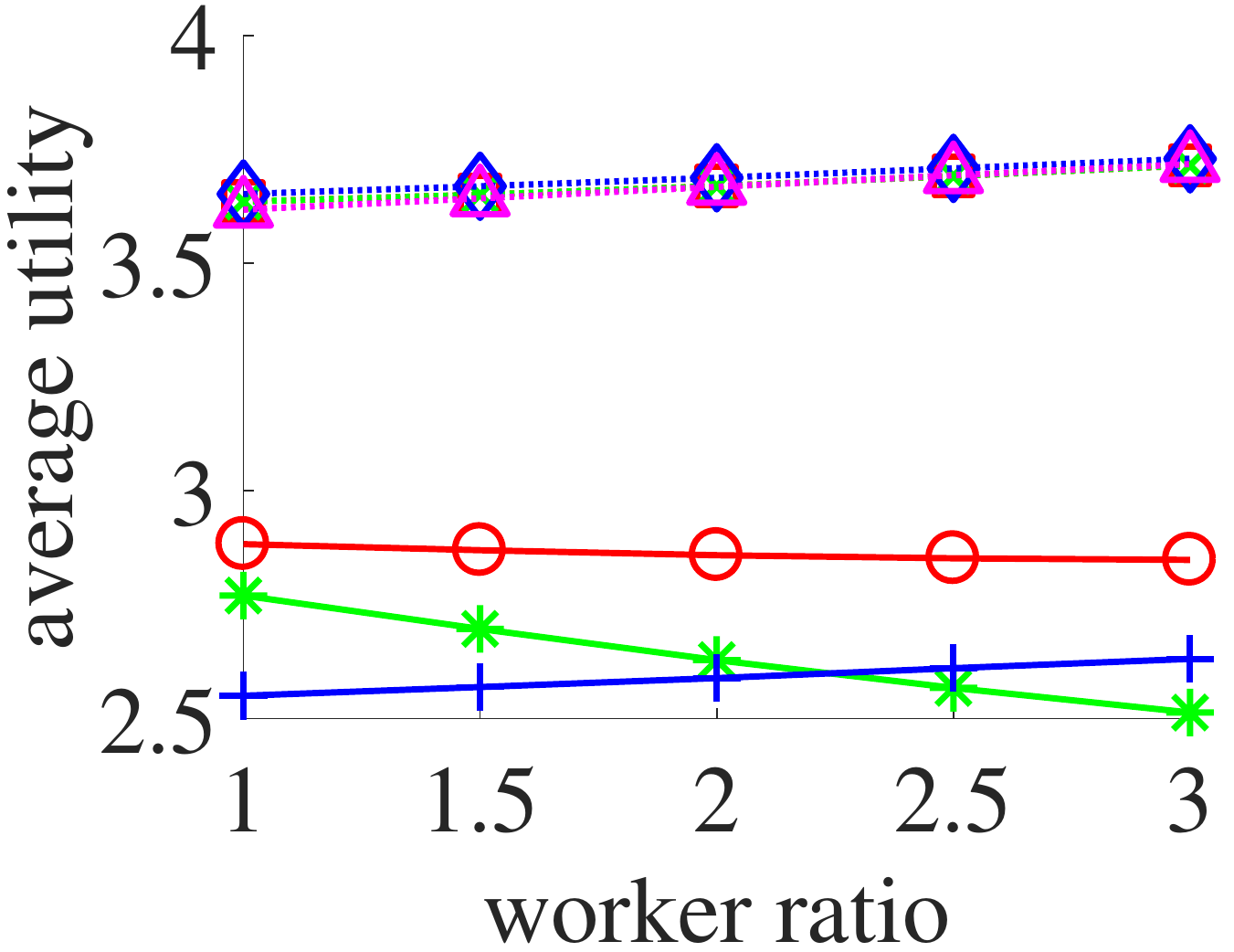}}
		\label{subfig:worker_ratio_utility_uniform}}
	\subfigure[][{\scriptsize Relative Deviation of Utility}]{
		\scalebox{0.25}[0.25]{\includegraphics{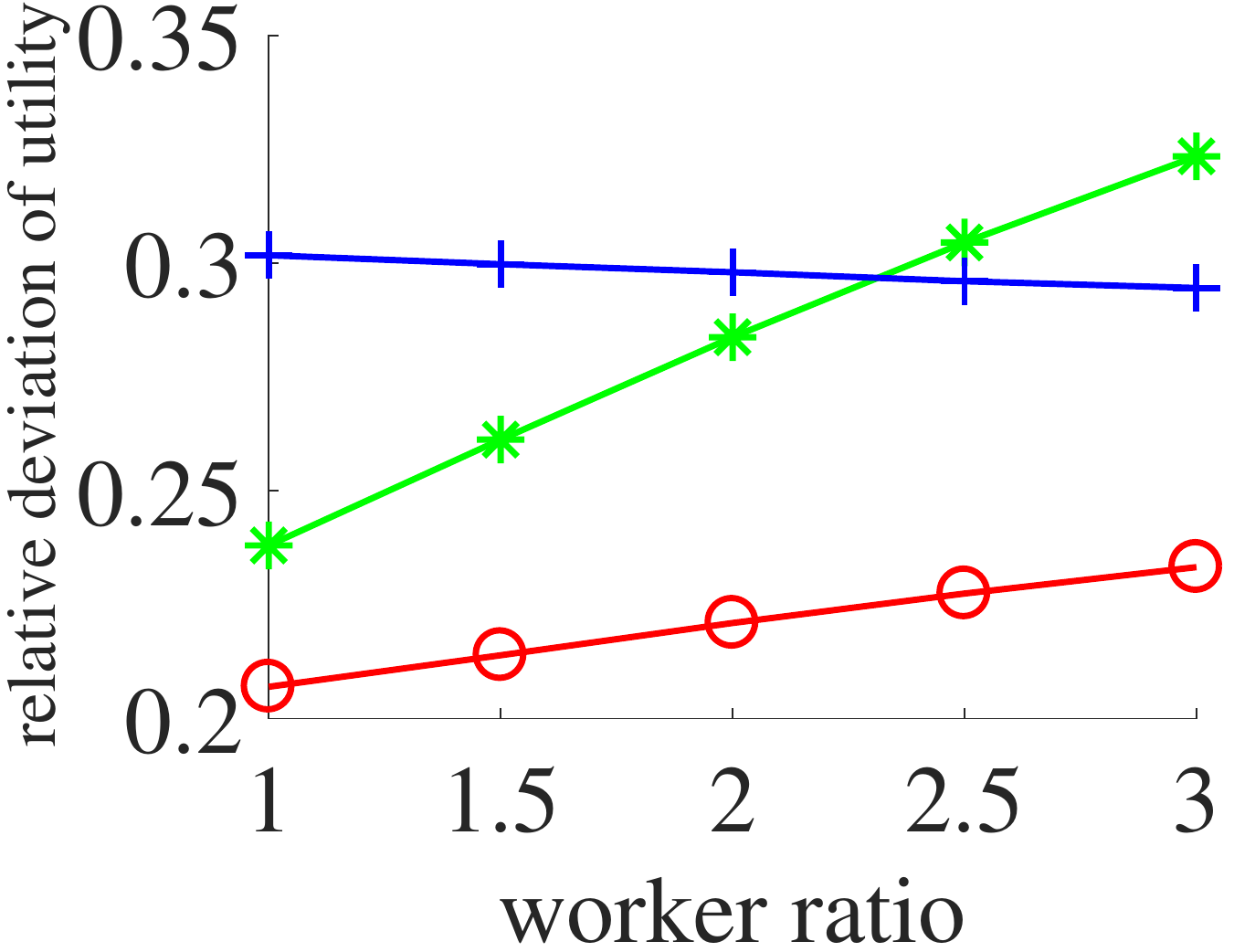}}
		\label{subfig:worker_ratio_utility_deviation_uniform}}\figureCaptionMargin
	\caption{\small The impact of the worker ratio on the utility for \emph{uniform}.}
	\label{fig:worker_ratio_utility_uniform}
\end{figure}

The impact of task value on travel distance is shown in Figure~\ref{fig:task_value_distance_uniform}.
\begin{figure}[ht!]\centering 
\subfigure{
  \scalebox{0.33}[0.33]{\includegraphics{bar2-eps-converted-to.pdf}}}\hfill\\
\addtocounter{subfigure}{-1}\vspace{-2ex}
	\subfigure[][{\scriptsize Average Distance}]{
		\scalebox{0.25}[0.25]{\includegraphics{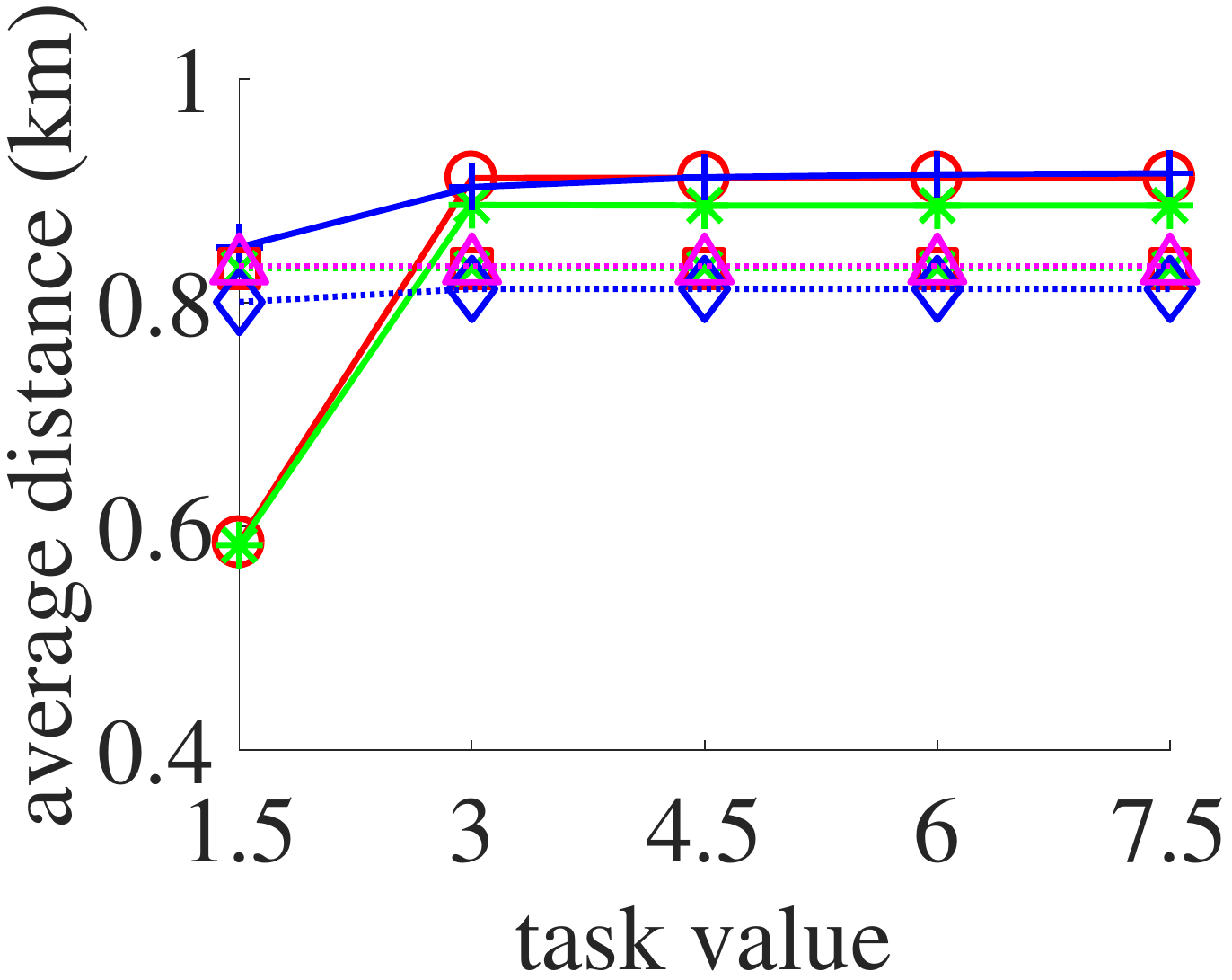}}
		\label{subfig:task_value_distance_uniform}}
	\subfigure[][{\scriptsize Relative Deviation of Distance}]{
		\scalebox{0.25}[0.25]{\includegraphics{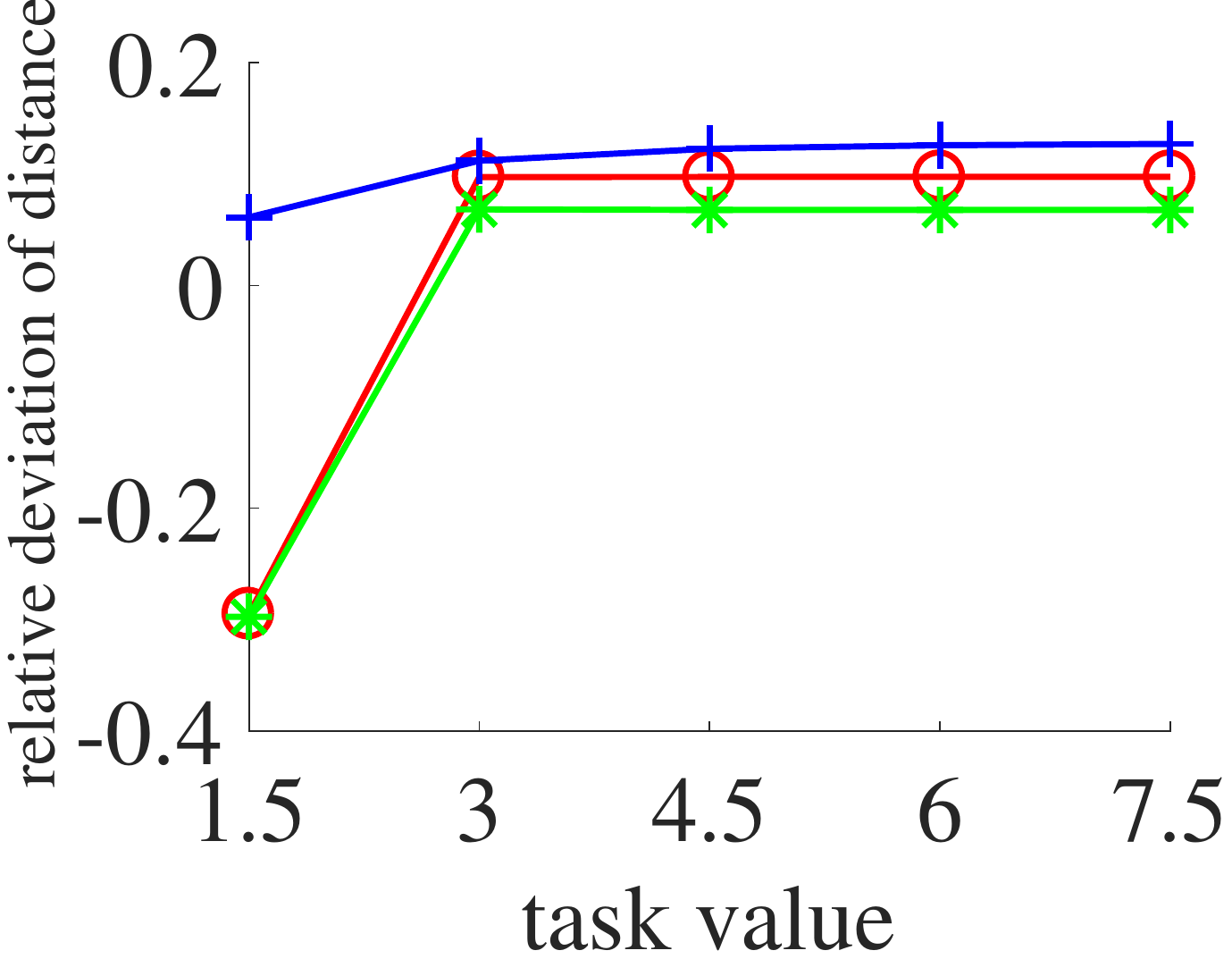}}
		\label{subfig:task_value_distance_deviation_uniform}}\figureCaptionMargin
	\caption{\small The impact of the task value on the distance for \emph{uniform}.}
	\label{fig:task_value_distance_uniform}
\end{figure}

The impact of worker range on travel distance is shown in Figure~\ref{fig:worker_range_distance_uniform}.
\begin{figure}[ht!]\centering 
\subfigure{
  \scalebox{0.33}[0.33]{\includegraphics{bar2-eps-converted-to.pdf}}}\hfill\\
\addtocounter{subfigure}{-1}\vspace{-2ex}
	\subfigure[][{\scriptsize Average Distance}]{
		\scalebox{0.25}[0.25]{\includegraphics{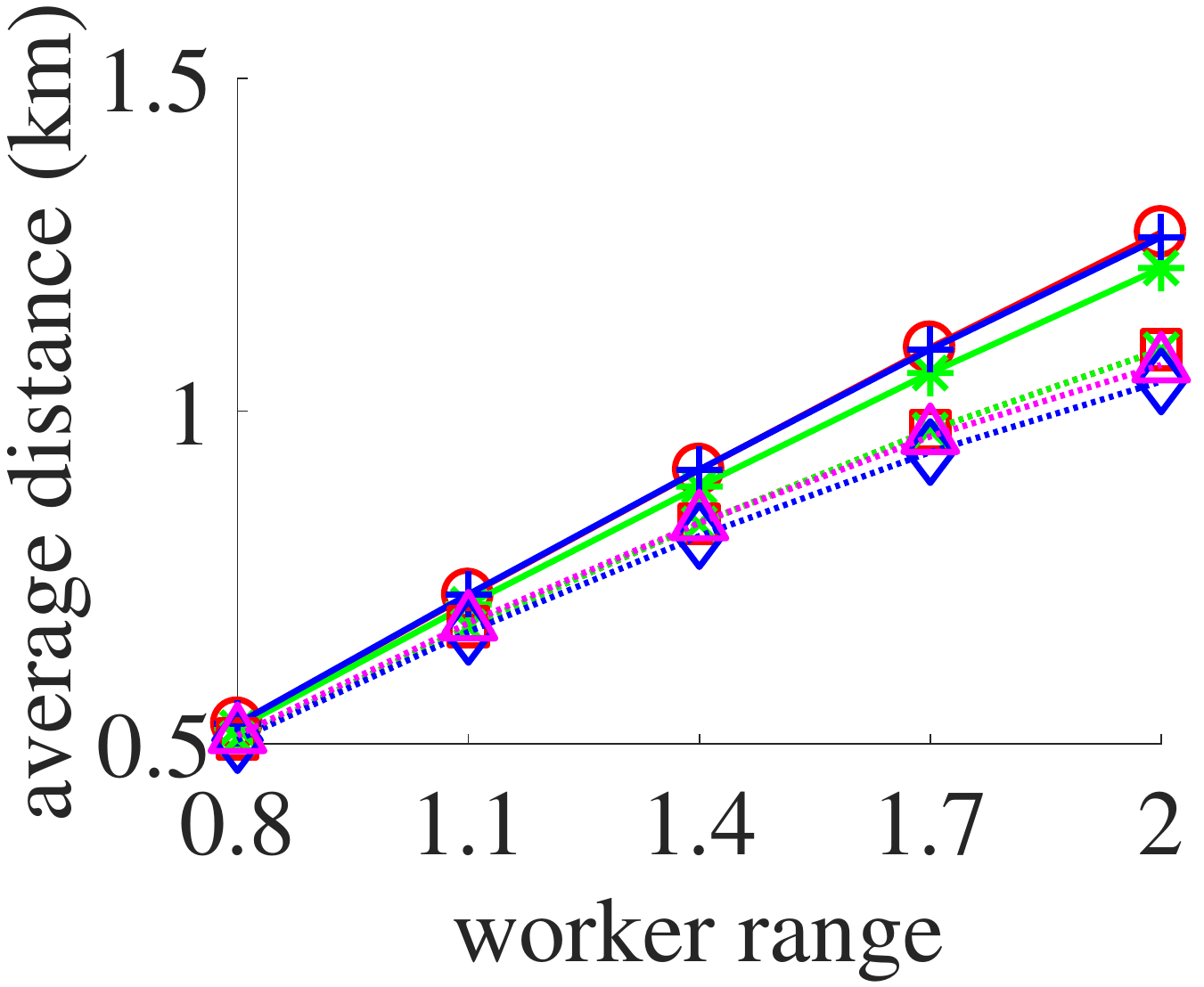}}
		\label{subfig:worker_range_distance_uniform}}
	\subfigure[][{\scriptsize Relative Deviation of Distance}]{
		\scalebox{0.25}[0.25]{\includegraphics{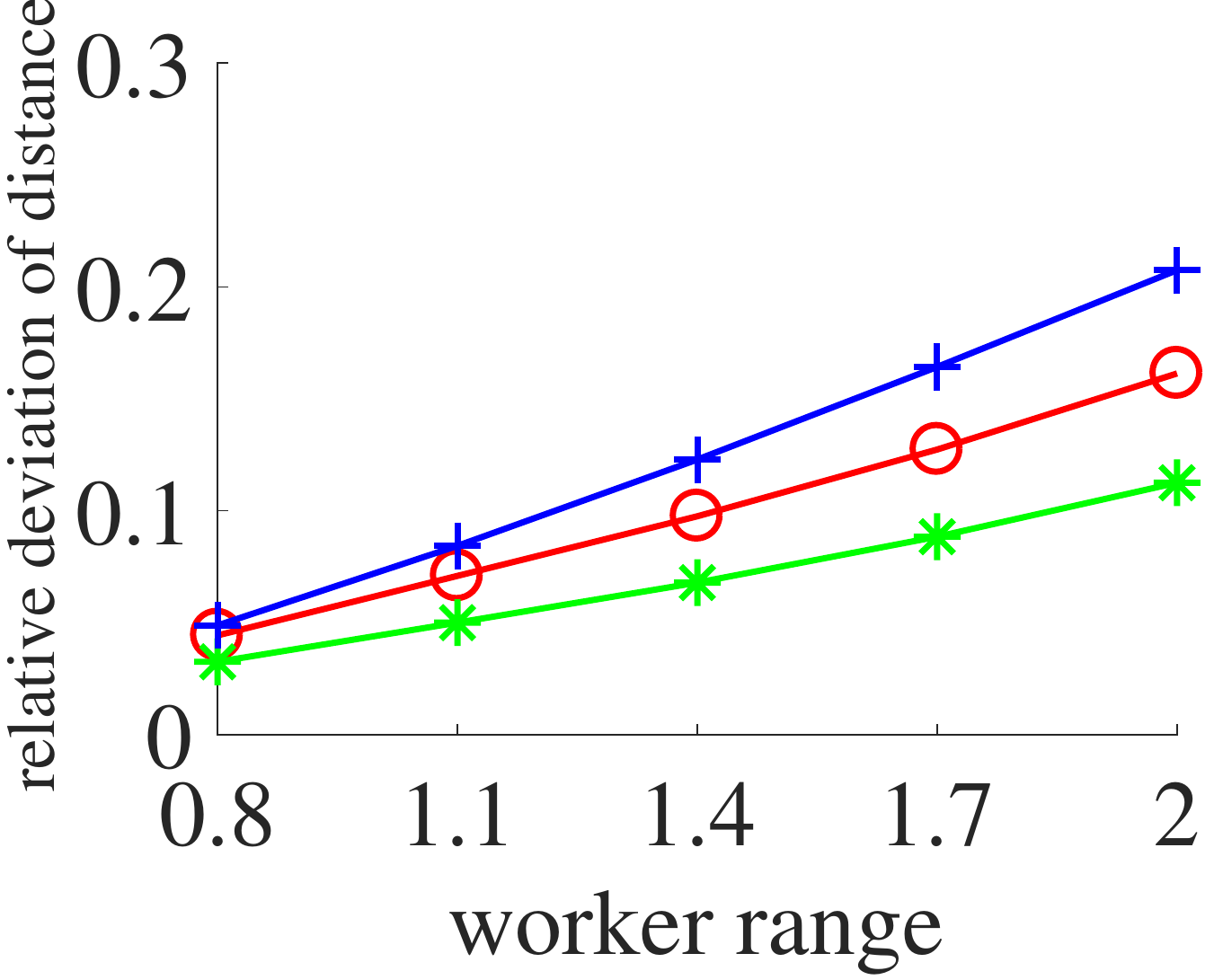}}
		\label{subfig:worker_range_distance_deviation_uniform}}\figureCaptionMargin
	\caption{\small The impact of the worker ratio on the distance for \emph{uniform}.}
	\label{fig:worker_range_distance_uniform}
\end{figure}

The impact of worker ratio on travel distance is shown in Figure~\ref{fig:worker_ratio_distance_uniform}.
\begin{figure}[ht!]\centering
\subfigure{
  \scalebox{0.33}[0.33]{\includegraphics{bar2-eps-converted-to.pdf}}}\hfill\\
\addtocounter{subfigure}{-1}\vspace{-2ex}
	\subfigure[][{\scriptsize Average Distance}]{
		\scalebox{0.25}[0.25]{\includegraphics{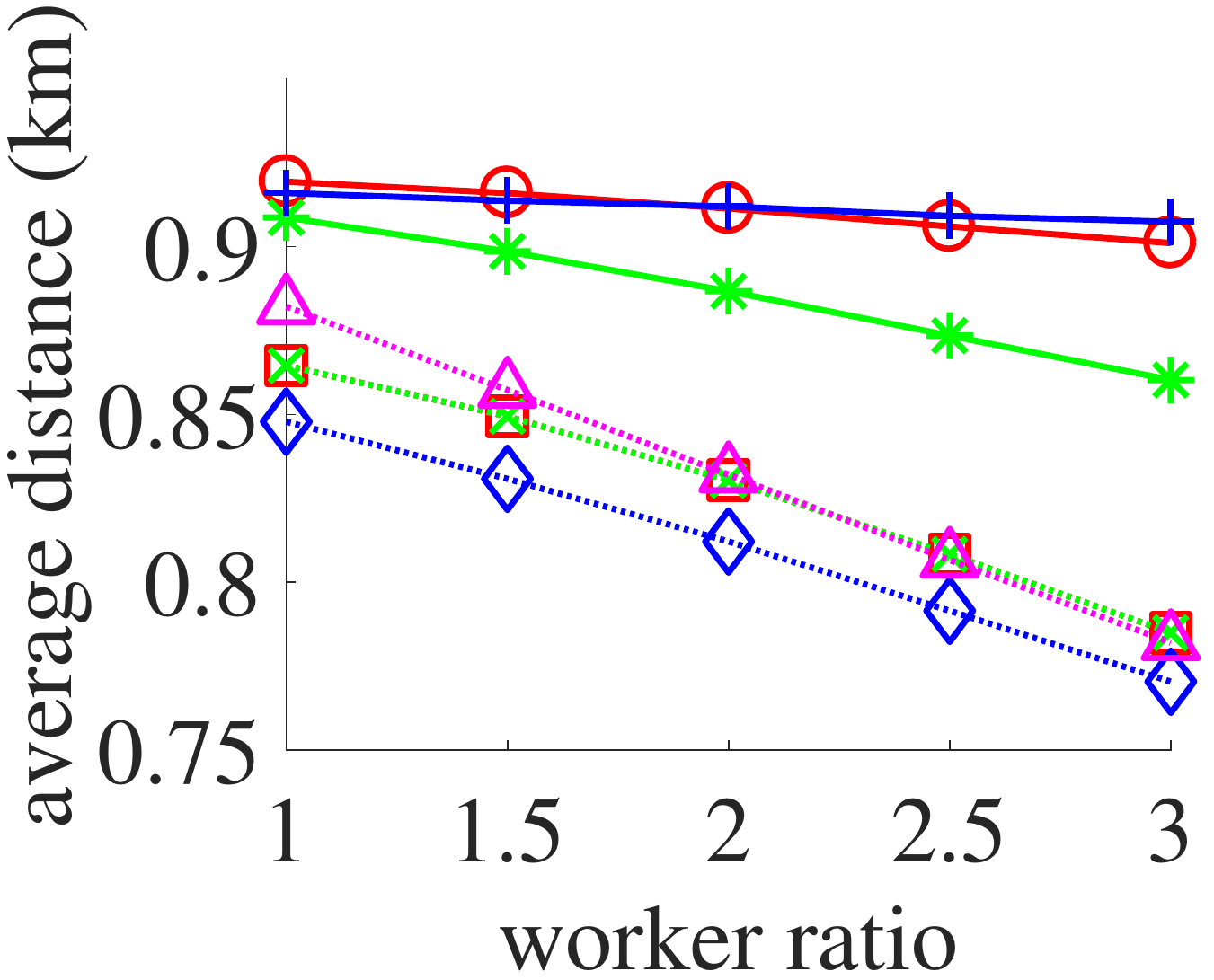}}
		\label{subfig:worker_ratio_distance_uniform}}
	\subfigure[][{\scriptsize Relative Deviation of Distance}]{
        \scalebox{0.25}[0.25]{\includegraphics{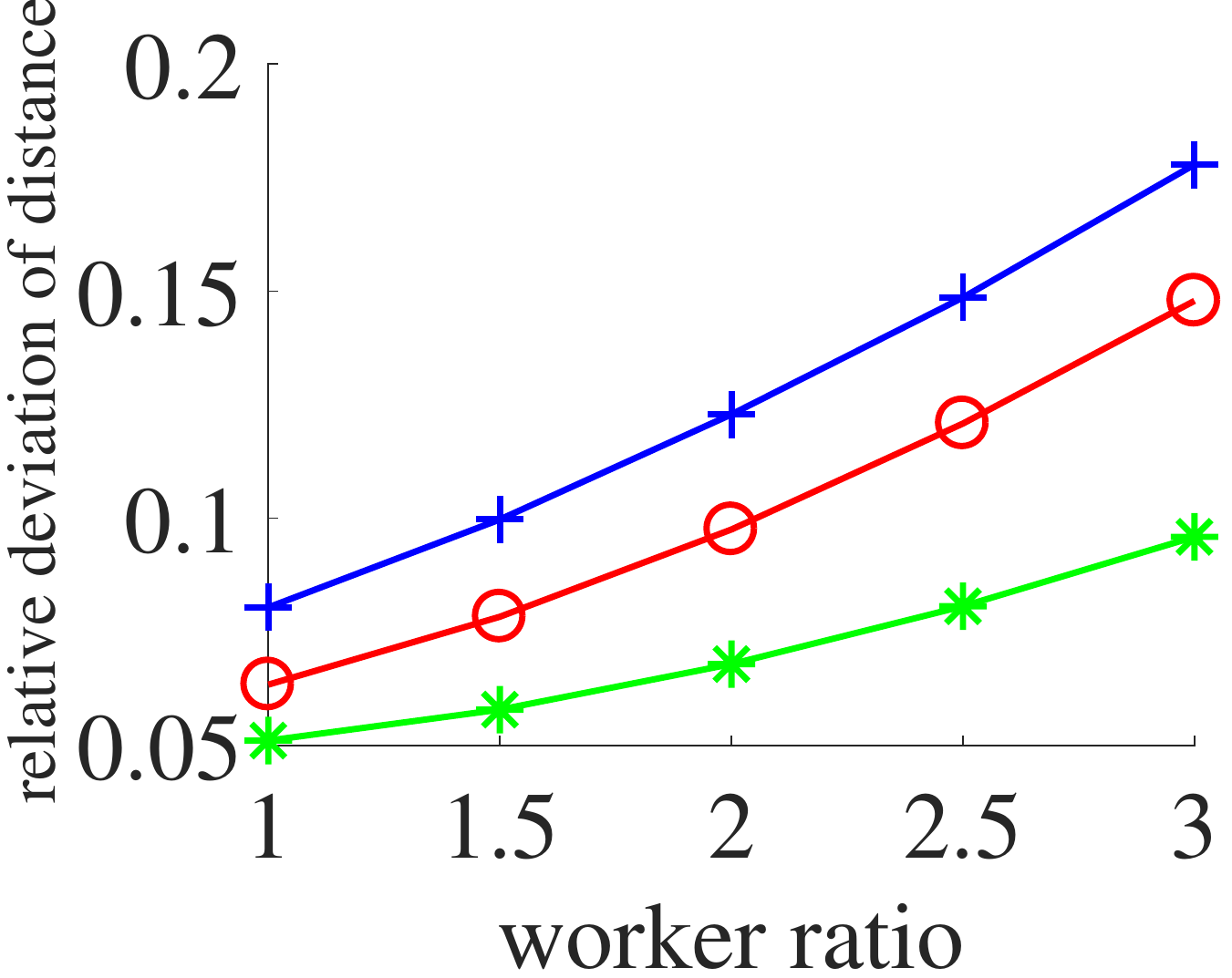}}
		\label{subfig:worker_ratio_distance_deviation_uniform}}\figureCaptionMargin
	\caption{\small The impact of the worker ratio on the distance for \emph{uniform}.}
	\label{fig:worker_ratio_distance_uniform}
\end{figure}

The impact of worker ratio on PPCF and non-PPCF is shown in Figure~\ref{fig:privacy_budget_utility_just_for_uniform}.
\begin{figure}[ht!]\centering 
\subfigure{
  \scalebox{0.33}[0.33]{\includegraphics{bar3-eps-converted-to.pdf}}}\hfill\\
\addtocounter{subfigure}{-1}\vspace{-2ex}
	\subfigure[][{\scriptsize Average Distance}]{
		\scalebox{0.25}[0.25]{\includegraphics{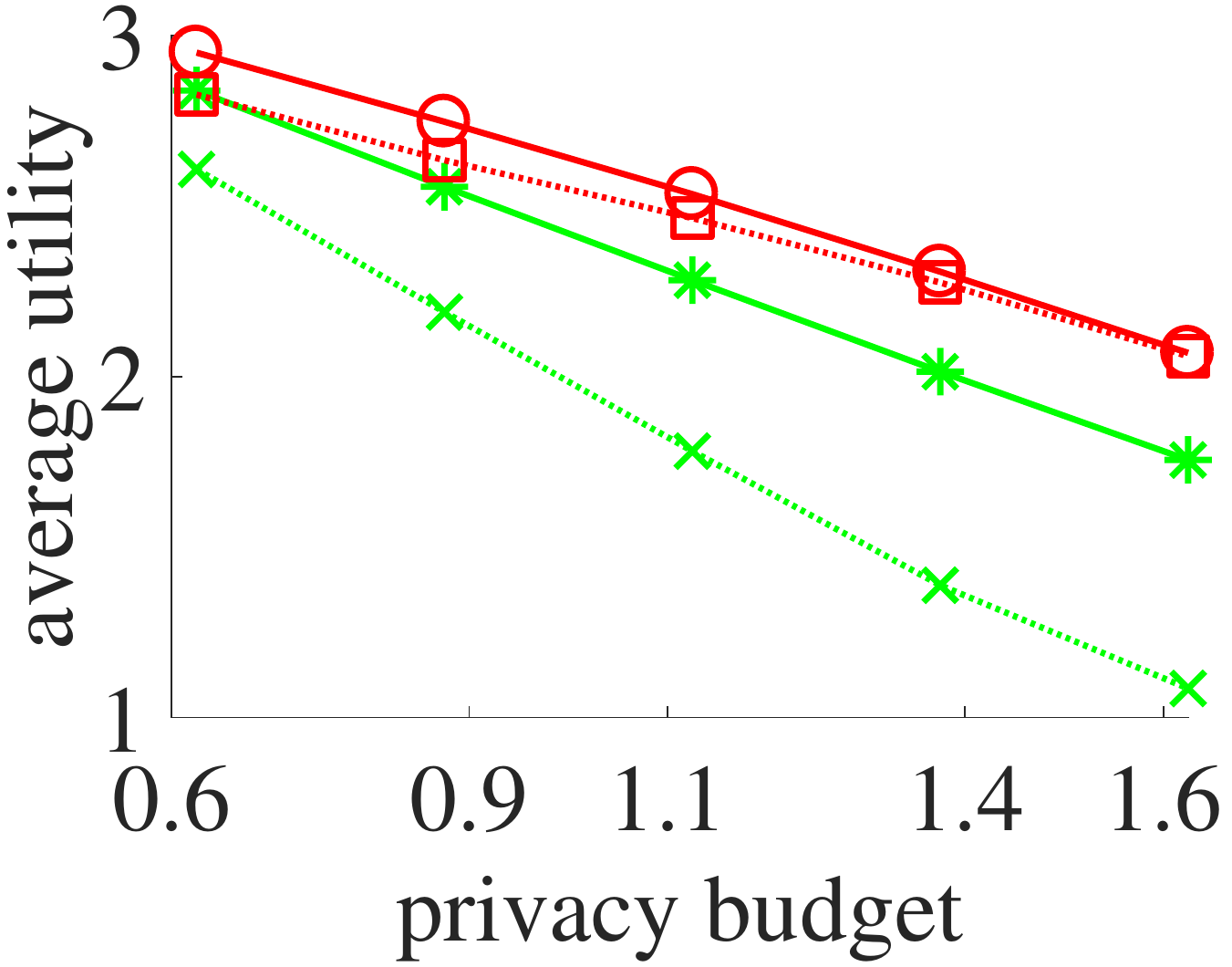}}
		\label{subfig:privacy_budget_utility_uniform}}
	\caption{\small The impact of privacy on the utility.}
	\label{fig:privacy_budget_utility_just_for_uniform}
\end{figure}

\subsection{Additional Definitions}
	\begin{definition}[One-to-one Match]\label{O2OM} Let {$G=(U,E,V)$} be a bipartite graph, and {$M\subseteq E$} be a match in {$G$}.
		{$M$} is called a one-to-one match if for any two different edges {$e_{u,v},e_{u',v'}\in M$},  {$e_{u,v}\cap e_{u',v'}=\phi$}.
	\end{definition}

	Let $S_{G}(M)=[s_{u,v}]_{u\in U, v\in V}$ denote a state matrix of $M$. Here, $s_{u,v}=1$ when $e_{u,v}\in M$, $s_{u,v}=1$; otherwise, $s_{u,v}=0$.

	We utilize Differential Privacy \cite{DBLP:journals/fttcs/DworkR14} to disturb the raw data and measure the privacy cost of workers' proposals for tasks.

\begin{definition}[Differential Privacy \cite{DBLP:journals/fttcs/DworkR14}, DP]\label{def_DP} A randomized algorithm {$\algvar{A}$} with domain {$\constvar{N}^{|\algvar{X}|}$} is {$(\epsilon, \delta)$}-differential private if for all {$\algvar{S}\subseteq \textrm{Range}(\algvar{A})$} and for all {$x,y\in\constvar{N}^{|\algvar{X}|}$}
	such that {$||x-y||_1\leq 1$}:
	{\scriptsize$$\textrm{Pr}[\algvar{A}(x)\in\algvar{S}]\leq\textrm{exp}(\epsilon)\textrm{Pr}[\algvar{A}(y)\in\algvar{S}]+\delta,$$}
where $||\cdot||_1$ is the $\ell_1$ norm of an vector, $\textrm{Pr}[\cdot]$ denotes the probability of an event.
	Especially, when {$\delta=0$}, {$\algvar{A}$} is {$\epsilon$}-differential private.
\end{definition}

	When $x$ and $y$ consists of single element, $\algvar{S}$ is also called a local randomizer, which provides \emph{local differential privacy} (LDP) guarantees \cite{DBLP:conf/sigmod/Liew0TK0Y22}.

	The Laplace mechanism \cite{DBLP:journals/fttcs/DworkR14} is the most well know perturbation methods for numeric values that satisfy the definition of differential privacy.
	Given a function $f$ outputting a numeric value vector $f(\cdot)$, the laplace mechanism is able to transform $f$ into a differentially private algorithm by adding random noise to each entry of $f(\cdot)$.
	The random noise is sampled from laplace distribution.

	The scale of the random noise is relevant to the $\ell_1$-sensitivity of $f$ as well as a predetermined privacy budget $\epsilon$.
	The $\ell_1$-sensitivity of $f$ is defined as the maximum possible difference between any two vector $x$ and $y$ with their $\ell_1$ distance as 1:

	\begin{definition}[{$\ell_1$}-sensitivity \cite{DBLP:journals/fttcs/DworkR14}] The {$\ell_1$}-sensitivity of a function  {$f:$} {$\constvar{N}^{|\algvar{X}|}\to\constvar{R}^{k}$} is:\vspace{-1ex}
		{\small$$\Delta f = \textrm{max}_{x,y\in\constvar{N}^{|\algvar{X}|},||x-y||_1=1}||f(x)-f(y)||_1.$$}\end{definition}

Thus, the laplace mechanism is stated as follows:
\begin{definition}[The Laplace Mechanism \cite{DBLP:journals/fttcs/DworkR14}] Given any function {$f:$} {$\constvar{N}^{|\algvar{X}|}\to\constvar{R}^{k}$}, the Laplace mechanism is defined as:\vspace{-2ex}
	{\small$$\algvar{A}_L(x,f(\cdot),\epsilon)=f(x)+(Y_1,...,Y_k)$$}
	where {$Y_i$} ($i\in[k]$) is an i.i.d. random variable drawn from  {$Lap(\Delta f/\epsilon)$}.
\end{definition}

			In order to compare two disturbed value conveniently, we introduce the Probability Compare Function \cite{DBLP:journals/tmc/WangHLWWYQ19} in Definition~\ref{def_PCF}.

\end{document}